\documentclass[aps,prb,twocolumn,superscriptaddress,showpacs]{revtex4}

\usepackage{color}
\usepackage{times}
\usepackage{epsfig}
\usepackage{amsmath}
\usepackage{amssymb}
\usepackage{amsmath, amsthm, amssymb,  mathrsfs, epsfig}
\usepackage{dcpic, pictex, pinlabel, setspace, rotating}
\usepackage{faktor}
\newcommand{\be}{\begin{equation}}
\newcommand{\ee}{\end{equation}}
\newcommand{\vspan}{\text{span}}

\newtheorem{theorem}{Theorem}[section]

\newtheorem{proposition}[theorem]{Proposition}

\newtheorem{fact}[theorem]{Fact}
\newtheorem{consequence}[theorem]{Consequence}
\newtheorem*{claim}{Claim}

\theoremstyle{definition}

\theoremstyle{remark}

\newcommand{\id}{{\rm{Id}}}
\newcommand{\R}{{\mathbb R}}

\newcommand{\Z}{{\mathbb Z}}

\newcommand{\Owe}{{\mathcal O}} % the ArXiV doesn't like you to redefine \O

\newcommand{\C}{{\mathbb C}}
\newcommand{\bbar}{\overline}

\newcommand{\til}{\widetilde}

\newcommand{\UO}{\til{U/O}}

\def\hline{\bigskip\hrule\bigskip}  % temporary def
\begin{document}

\title{Projective Ribbon Permutation Statistics: a Remnant
of non-Abelian Braiding in Higher Dimensions}

\author{Michael Freedman}
\affiliation{Microsoft Research, Station Q, Elings Hall, University of California, Santa Barbara, CA 93106}
\author{Matthew B. Hastings}
\affiliation{Microsoft Research, Station Q, Elings Hall, University of California, Santa Barbara, CA 93106}
\author{Chetan Nayak}
\affiliation{Microsoft Research, Station Q, Elings Hall, University of California, Santa Barbara, CA 93106}
\affiliation{Department of Physics, University of California, Santa Barbara, CA 93106}
\author{Xiao-Liang Qi}
\affiliation{Microsoft Research, Station Q, Elings Hall, University of California, Santa Barbara, CA 93106}
\affiliation{Department of Physics, Stanford University, Stanford, CA 94305, USA}
\author{Kevin Walker}
\affiliation{Microsoft Research, Station Q, Elings Hall, University of California, Santa Barbara, CA 93106}
\author{Zhenghan Wang}
\affiliation{Microsoft Research, Station Q, Elings Hall, University of California, Santa Barbara, CA 93106}

\begin{abstract}
In a recent paper, Teo and Kane proposed
a 3D model in which the defects support Majorana
fermion zero modes. They argued that exchanging and twisting
these defects would implement a set ${\cal R}$
of unitary transformations on the zero mode Hilbert space
which is a `ghostly' recollection of the
action of the braid group on Ising anyons in 2D.
In this paper, we find the group ${\cal T}_{2n}$
which governs the statistics of these defects by
analyzing the topology of the space $K_{2n}$ of configurations
of $2n$ defects in a slowly spatially-varying
gapped free fermion Hamiltonian:
${\cal T}_{2n}\equiv {\pi_1}(K_{2n})$. We find that the group
${\cal T}_{2n}=\Z \times {\cal T}^r_{2n}$, where
the `ribbon permutation group' ${\cal T}^r_{2n}$
is a mild enhancement of the permutation
group $S_{2n}$: ${\cal T}^r_{2n} \equiv \Z_2 \times E((\mathbb{Z}_2)^{2n}\rtimes S_{2n})$. Here, $E((\mathbb{Z}_2)^{2n}\rtimes S_{2n})$
is the `even part' of $(\mathbb{Z}_2)^{2n} \rtimes S_{2n}$,
namely those elements for which the total parity of
the element in $(\mathbb{Z}_2)^{2n}$ added to
the parity of the permutation is even.
Surprisingly, ${\cal R}$ is only a projective representation
of ${\cal T}_{2n}$, a possibility proposed by Wilczek.
Thus, Teo and Kane's defects realize
`Projective Ribbon Permutation Statistics',
which we show to be consistent with locality.
We extend this phenomenon to other dimensions,
co-dimensions, and symmetry classes. Since it is an
essential input for our calculation, we review the
topological classification of gapped free fermion systems
and its relation to Bott periodicity.
\end{abstract}

\maketitle

%%%%%%%%%%%%%%%%%%%%%%%%%%%%%%%%%%

\section{Introduction}

In two dimensions, the configuration space
of $n$ point-like particles ${\cal C}_n^{2D}$
is multiply-connected. Its first homotopy
group, or {\it fundamental group},
is the $n$-particle braid group,
${\pi_1}({\cal C}_n^{2D})={\cal B}_n$.
The braid group ${\cal B}_n$ is generated
by counter-clockwise exchanges $\sigma_i$
of the $i^\text{th}$ and $(i+1)^\text{th}$ particles
satisfying the defining relations:
\begin{eqnarray}
{\sigma_i} {\sigma_j} &=& {\sigma_j} {\sigma_i} \hskip 0.5 cm
\mbox{for } |i-j|\geq 2\cr {\sigma_i} \sigma_{i+1} {\sigma_i} &=&
\sigma_{i+1} {\sigma_i}\, \sigma_{i+1} \hskip 0.5 cm \mbox{for }
1\leq i \leq n-2
\label{eq:braidrelation1}
\end{eqnarray}
This is an infinite group, even for only two particles,
since $(\sigma_i)^m$ is a non-trivial element of the
group for any $m>0$. In fact, even if we consider
distinguishable particles, the resulting group,
called the `pure Braid group' is non-trivial.
(For two particles, the pure braid group
consists of all even powers
of $\sigma_1$.)

In quantum mechanics, the equation
${\pi_1}({\cal C}_n^{2D})={\cal B}_n$
opens the door to the possibility of anyons\cite{Leinaas77,Wilczek82a}.
Higher-dimensional representations
of the braid group give rise to
non-Abelian anyons \cite{Bais80,Goldin85,Frohlich90}.
There has recently been intense effort directed towards
observing non-Abelian anyons due, in part,
to their potential use for fault-tolerant quantum
computation \cite{Kitaev97,Nayak08}.
One of the simplest models of non-Abelian
anyons is called {\it Ising anyons}. They arise
in theoretical models of the $\nu=5/2$ fractional quantum
Hall state \cite{Moore91,Nayak96c,LeeSS07,Levin07}
(see also Ref. \onlinecite{Bonderson08}),
chiral $p$-wave superconductors \cite{Read00,Ivanov01},
a solvable model of spins on the honeycomb lattice
\cite{Kitaev06a}, and interfaces between
superconductors and either 3D topological
insulators \cite{Fu08} or spin-polarized
semiconductors with strong spin-obrit coupling \cite{Sau09}.
A special feature of Ising anyons, which makes
them relatively simple and connects them to BCS
superconductivity, is that they can be understood
in a free fermion picture.

A collection of $2n$ Ising anyons has a $2^{n-1}$-dimensional
Hilbert space (assuming fixed boundary condition).
This can be understood in terms of $2n$ Majorana fermion
operators ${\gamma^{}_i}={\gamma_i^\dagger}$,
$i=1,2,\ldots,n$, one associated to
each Ising anyon, satisfying the anticommutation rules
\begin{equation}
\label{eqn:clifford}
\{{\gamma^{}_i},{\gamma^{}_j}\}=2\delta_{ij}\,.
\end{equation}
The Hilbert space of $2n$ Ising anyons with fixed boundary condition furnishes a representation
of this Clifford algebra; by restricting to fixed boundary condition, we obtain
a representation
of products of an even number of $\gamma$ matrices, which has minimal dimension $2^{n-1}$.
When the $i^\text{th}$ and $(i+1)^\text{th}$ anyons
are exchanged in a counter-clockwise manner,
a state of the system is transformed according to the action of
\begin{equation}
\label{eqn:Ising-braid}
\rho({\sigma_i})=e^{i\pi/8}\,e^{-\pi{\gamma^{}_i}\gamma^{}_{i+1}/4}\,.
\end{equation}
(There is a variant of Ising anyons, associated with
SU(2)$_2$ Chern-Simons theory,
for which the phase factor $e^{i\pi/8}$
is replaced by $e^{-i\pi/8}$. In the fractional quantum
Hall effect, Ising anyons are tensored with
Abelian anyons to form more complicated models
with more particle species; the phase factor depends
on the model.) A key property, essential for applications
to quantum computing, is that {\it a pair} of Ising anyons
forms a two-state system. The two states
correspond to the two eigenvalues $\pm 1$
of ${\gamma^{}_i}\gamma^{}_{j}$. No local degree of freedom
can be associated with each anyon; if we insisted on doing so,
we would have to say that each Ising anyon has $\sqrt{2}$
internal states. In superconducting contexts,
the ${\gamma^{}_i}$s are the Bogoliubov-de Gennes operators
for zero-energy modes (or, simply, `zero modes')
in vortex cores; the vortices
themselves are Ising anyons if they possess a single
such zero mode ${\gamma^{}_i}$.
Although the Hilbert
space is non-local in the sense that it cannot be decomposed
into the tensor product of local Hilbert spaces associated
with each anyon, the system is perfectly compatible
with locality and arises in local lattice models and
quantum field theories.

In three or more dimensions,
the configuration space of $n$ point-like particles
is simply-connected if the particles are distinguishable.
If the particles are indistinguishable, it
is multiply-connected,
${\pi_1}({\cal C}_n^{3D})=S_n$.
The generators of the permutation group
satisfy the relations (\ref{eq:braidrelation1})
and one more, ${\sigma_i^2}=1$. As a result
of this last relation, the permutation group
is finite. The one-dimensional representations
of $S_n$ correspond to bosons and fermions.
One might have hoped that higher-dimensional
representations of $S_n$ would give rise to
interesting 3D analogues of non-Abelian anyons.
However, this is not the case, as shown in
Ref. \onlinecite{Doplicher71a,Doplicher71b}: any higher-dimensional
representation of $S_n$ which is compatible with
locality can be decomposed into the tensor product
of local Hilbert spaces associated
with each particle. For instance, suppose we
had $2n$ spin-$1/2$ particles but ignored
their spin values. Then we would have $2^{2n}$
states which would transform into each other
under permutations. Clearly, if we discovered such a system,
we would simply conclude that we were missing
some quantum number and set about trying to
measure it. This would simply lead us back
to bosons and fermions with additional
quantum numbers. (The color quantum number of quarks
was conjectured by essentially this kind of reasoning.)
The quantum information contained in these
$2^{2n}$ states would not have any special protection.

The preceding considerations point to
the following tension. The Clifford algebra
(\ref{eqn:clifford}) of Majorana fermion zero modes
is not special to two dimensions. One could imagine
a three (or higher) dimensional system with topological defects supporting such zero modes. But the Hilbert space of these
topological defects would be $2^{n-1}$-dimensional, which
manifestly cannot be decomposed into the tensor product
of local Hilbert spaces associated
with each particle, seemingly in contradiction with the results of 
Refs. \onlinecite{Doplicher71a,Doplicher71b} on
higher-dimensional representations of the permutation group
described above. However, as long as
no one had a three or higher dimensional system
in hand with topological defects supporting
Majorana fermion zero modes,
one could, perhaps, sweep this worry under the rug.
Recently, however, Teo and Kane \cite{Teo10}
have shown that a 3D system which is simultaneously
a superconductor and a topological insulator \cite{Moore07,Fu07,Roy09,Qi08}
(which, in many but not all examples, is
arranged by forming superconductor-topological insulator
heterostructures) supports Majorana zero modes at
point-like topological defects.

To make matters worse, Teo and Kane \cite{Teo10}
further showed that exchanging these
defects enacts unitary operations on this
$2^{n-1}$-dimensional Hilbert space which are
essentially equal to (\ref{eqn:Ising-braid}).
But we know that these unitary matrices
form a representation of the braid group,
which is not the relevant group in 3D.
One would naively expect that the relevant group is
the permutation group, but $S_n$ has no such
representation (and even if it did, its use in this
context would contradict locality, according to Ref. \onlinecite{Doplicher71a,Doplicher71b}
and arguments in Ref. \onlinecite{Read03}).
So this begs the question: what is the group ${\cal T}_{2n}$
for which Teo and Kane's unitary
transformations form a representation?

With the answer to this question in hand,
we could address questions such as the following.
We know that a 3D incarnation of Ising anyons
is one possible representation of ${\cal T}_{2n}$;
is a 3D version of other anyons another
representation of ${\cal T}_{2n}$?

Attempts to generalize the braiding of anyons
to higher dimensions sometimes start with
extended objects, whose configuration space
may have fundamental group which is
richer than the permutation group.
Obviously, if one has line-like defects
in 3D which are all oriented in the same direction,
then one is essentially back to the 2D situation
governed by the braid group. This is too
trivial, but it is not clear what kind of extended
objects in higher dimensions would be the best starting point.
What is clear, however, is that Teo and Kane's topological
defects must really be some sort of extended objects.
This is clear from the above-noted contradiction
with the permutation group. It also follows from
the `order parameter' fields which must deform
as the defects are moved, as we will discuss.

In this paper, we show that Teo and Kane's
defects are properly viewed as point-like
defects connected pair-wise by ribbons.
We call the resulting $2n$-particle configuration
space $K_{2n}$. We compute
its fundamental group ${\pi_1}(K_{2n})$, which we denote by
${\cal T}_{2n}$ and find that
${\cal T}_{2n}=\Z \times {\cal T}^r_{2n}$.
Here, ${\cal T}^r_{2n}$ is the `ribbon permutation group',
defined by ${\cal T}^r_{2n} \equiv \Z_2
\times E((\mathbb{Z}_2)^{2n}\rtimes S_{2n})$.
The group $E((\mathbb{Z}_2)^{2n}\rtimes S_{2n})$
is a non-split extension of the permutation group $S_{2n}$ by $\Z_2^{2n-1}$
which is defined as follows: it is the
subgroup of $(\mathbb{Z}_2)^{2n} \rtimes S_{2n}$
composed of those elements for which the total parity of
the element in $(\mathbb{Z}_2)^{2n}$ added to
the parity of the permutation is even.
The `ribbon permutation group' for $2n$ particles,
 by ${\cal T}^r_{2n}$, is the fundamental group of the
 reduced space of $2n$-particle configurations.

Our analysis relies on the topological
classification of gapped free fermion Hamiltonians
\cite{Ryu08,Kitaev09} -- band insulators
and superconductors -- which
is the setting in which Teo and Kane's 3D defects
and their motions are defined.
The starting point for this classification is reducing the problem
from classifying gapped Hamiltonians defined on a lattice to
classifying Dirac equations with a spatially varying mass term.
One can motivate the reduction to a Dirac equation as
Teo and Kane do: they
start from a lattice Hamiltonian and assume
that the parameters in the Hamiltonian vary smoothly in space, so that
the Hamiltonian can be described as a function of both the momentum $k$
and the position $r$.  Near the minimum of the band gap, the
Hamiltonian can be expanded in a Dirac equation, with a position-dependent
mass term.  In fact, Kitaev\cite{Kitaev09} has shown that the reduction
to the Dirac equation with a spatially varying mass term can be derived
much more generally:
gapped lattice Hamiltonians, even if the parameters
in the Hamiltonian do not vary smoothly in space, are {\it stably equivalent}
to Dirac Hamiltonians with a spatially varying mass term.  Here, equivalence
of two Hamiltonians means that one can be smoothly deformed into the other
while preserving locality of interactions and the spectral gap, while
stable equivalence means that one can add additional ``trivial" degrees of freedom (additional
sites which have vanishing hopping matrix elements) to the original lattice
Hamiltonian to obtain a Hamiltonian which is equivalent to
a lattice discretization of the Dirac Hamiltonian.

Since this classification of Dirac Hamiltonians
is essential for the definition of $K_{2n}$, we give
a self-contained review, following Kitaev's
analysis \cite{Kitaev09}. Our exposition parallels
the discussion of Bott periodicity in Milnor's book
\cite{Milnor63}. The basic idea is that each additional
discrete symmetry which squares to $-1$ which we
impose on the system is encapsulated by an
anti-symmetric matrix which defines a complex structure
on $\mathbb{R}^N$, where $N/2$ is the number of
bands (or, equivalently, $N$ is the number of bands
of Majorana fermions). For any given system,
these are chosen and fixed. This leads to
a progression of symmetric spaces $\text{O}(N)\rightarrow
\text{O}(N)/\text{U}(N/2) \rightarrow \text{U}(N/2)/\text{Sp}(N/4)
\rightarrow \ldots$ as the number of such symmetries is increased.
Following Kitaev \cite{Kitaev09}, we view the Hamiltonian
as a final anti-symmetric matrix which must be chosen (and, thus,
put almost on the same footing as the symmetries); it is defined by a choice
of an arbitrary point in the next symmetric space in the progression.
The space of such Hamiltonians is topologically-equivalent
to that symmetric space.
However, as the spatial dimension is increased, $\gamma$-matrices
squaring to $+1$ must be chosen in order to expand
the Hamiltonian in the form of the Dirac equation
in the vicinity of a minimum of the band gap. These halve the dimension
of subspaces of $\mathbb{R}^N$ by separating it
into their $+1$ and $-1$ eigenspaces and thereby
lead to the opposite progression of symmetric spaces. Thus,
taking into account both the symmetries of the system and
the spatial dimension, we conclude that the space of gapped
Hamiltonians with no symmetries in $d=3$ is topologically
equivalent to $\text{U}(N)/\text{O}(N)$. (However, by the preceding
considerations, the same symmetric space also, for instance, classifies
systems with time-reversal symmetry in $d=4$.)
All such Hamiltonians can be continuously deformed into each
other without closing the gap, $\pi_{0}(\text{U}(N)/\text{O}(N))=0$.
However, there are topologically-stable point-like defects
classified by $\pi_{2}(\text{U}(N)/\text{O}(N))=\mathbb{Z}_2$.
These are the defects whose multi-defect configuration space
we study in order to see what happens when they are exchanged.

The second key ingredient in our analysis
is 1950's-vintage homotopy theory, which we use to compute
${\pi_1}(K_{2n})$.  We apply the Pontryagin-Thom
construction to show that $K_{2n}$, which
includes not only the particle locations but also
the full field configuration  around the particles
(i.e. the way in which the gapped free fermion
Hamiltonian of the system explores $\text{U}(N)/\text{O}(N)$),
is topologically-equivalent to a much simpler
space, namely point-like defects connected
pair-wise by ribbons. In order to then
calculate ${\pi_1}(K_{2n})$, we rely on the long
exact sequence of homotopy groups
\begin{equation}
\label{eqn:long-exact-sequence}
\ldots \rightarrow \pi_{i}(E)\rightarrow
{\pi_i}(B)\rightarrow\pi_{i-1}(F)\rightarrow\pi_{i-1}(E)
\rightarrow ...
\end{equation}
associated to a fibration defined by
$F \rightarrow E\rightarrow B$.
(In an exact sequence, the kernel of each map
is equal to the image of the previous map.)
This exact sequence may be familiar to some readers
from Mermin's review of the topological theory of
defects \cite{Mermin79}, where a symmetry associated
with the group $G$ is spontaneously broken to $H$,
thereby leading to topological
defects classified by homotopy groups ${\pi_n}(G/H)$.
These can be computed by (\ref{eqn:long-exact-sequence})
with $E=G$, $F=H$, $B=G/H$, e.g.
if $\pi_{1}(G)=\pi_{0}(G)=0$,
then $\pi_{1}(G/H)=\pi_{0}(H)$.

The ribbon permutation group is a rather weak enhancement
of the permutation group and, indeed, we conclude
that Teo and Kane's unitary operations are
{\it not} a representation of the ribbon permutation
group. However, they are a {\it projective} representation
of the ribbon permutation group. In a
{\it projective} representation, the group
multiplication rule is only respected up to
a phase, a possibility allowed in quantum mechanics.
A representation $\rho$ (sometimes called a linear
representation) of some group $G$ is
a map from the group to the group of linear transformations
of some vector space such that
the group multiplication law is reproduced:
\begin{equation}
\rho(gh)=\rho(g)\cdot\rho(h)
\end{equation}
if $g,h\in G$. Particle statistics arising as a projective
representation of some group
realizes a proposal of Wilczek's \cite{Wilczek98},
albeit for the ribbon permutation group rather than
the permutation group itself. This difference
allows us to sidestep a criticism of Read \cite{Read03}
based on locality, which Teo and Kane's
projective representation respects.
The group $(\mathbb{Z}_2)^{2n-1}$ is generated
by $2n-1$ generators ${x_1}$, ${x_2}$,
\ldots, $x_{2n-1}$ satisfying
\begin{eqnarray}
{x_i^2} &=& 1\cr
{x_i} {x_j} &=& {x_j} {x_i}
\label{eq:Z_2-def}
\end{eqnarray}
However, the projective representation of
$(\mathbb{Z}_2)^{2n-1}$, which gives a subgroup of
Teo and Kane's transformations, is an ordinary
linear representation of a $\Z_2$-central extension,
called the extra special group $E^1_{2n-1}$:
\begin{eqnarray}
{x_i^2} &=& 1\cr
{x_i} {x_j} &=& {x_j} {x_i} \hskip 0.5 cm
\mbox{for } |i-j|\geq 2\cr
{x_i} x_{i+1} &=& z \,x_{i+1} {x_i}\cr
{z^2}&=&1
\label{eq:extra-special}
\end{eqnarray}
Here, $z$ generates the central extension, which we may
take to be $z=-1$. The operations generated
by the ${x_i}$s were dubbed `braidless operations'
by Teo and Kane \cite{Teo10} because they could
be enacted without moving the defects. While these
operations form an Abelian subgroup of ${\cal T}_{2n}$,
their representation on the Majorana zero mode
Hilbert space is {\it not} Abelian -- two such operations
which twist the same defect {\it anti-commute} (e.g. $x_i$ and
$x_{i+1}$).

The remaining sections of this paper will
be as follows. In Section \ref{sec:strong-coupling},
we rederive Teo and Kane's zero modes and unitary transformations
by simple pictorial and counting arguments in a `strong-coupling'
limit of their model. In Section \ref{sec:free-fermion},
we review the topological classification of free-fermion
Hamiltonians, including topological insulators and
superconductors. From this classification, we obtain
the classifying space relevant to Teo and Kane's model
and, in turn, the topological classification of
defects and their configuration space.
In Section \ref{sec:tethered}, we use a toy
model to motivate a simple picture for the defects
used by Teo and Kane and give a heuristic
construction of the ribbon permutation group.
In Section \ref{sec:Kane_space}, we give a full
homotopy theory calculation.
In Section \ref{sec:projective}, we compare
the ribbon permutation group to
Teo and Kane's unitary transformations and conclude that
the latter form a projective, rather than a linear,
representation of the former. Finally, in Section
\ref{sec:discussion}, we review and discuss our results.
Several appendices contain technical details.

\section{Strong-coupling limit of the Teo-Kane Model}
\label{sec:strong-coupling}

In this section, we present a lattice model
in $d$ dimensions which has,
as its continuum limit in $d=3$, the model discussed by
Teo and Kane \cite{Teo10}. In the limit that the mass terms
in this model are large (which can be viewed as
a `strong-coupling' limit), a simple picture of
topological defects (`hedgehogs') emerges.
We show by a counting argument that hedgehogs
possess Majorana zero modes which evolve as the
hedgehogs are adiabatically moved. This adiabatic evolution
is the 3D non-Abelian statistics which it is the main purpose
of this paper to explain.

The strong coupling limit which we describe is the
simplest way to derive the existence of Majorana zero
modes and the unitary transformations of their Hilbert space
which results from exchanging them. This section does
not require the reader to be {\it au courant} with the
topological classification of insulators and superconductors
\cite{Ryu08,Kitaev09}. (In the next section, we will
review that classification in order to make our exposition
self-contained.)

We use a hypercubic lattice in $d$-dimensions, with a single Majorana degree of freedom at each site.
That is,
for $d=1$, we use a chain, in $d=2$ we use a square lattice,
in $d=3$ we use
a cubic lattice, and so on.
We first construct a lattice model whose continuum limit is the Dirac equation with $2^d$-dimensional
$\gamma$-matrices to reproduce the Dirac equation considered by Teo
and Kane; we then show how to perturb this model to open a mass gap.
We begin by considering only nearest neighbor couplings.  The Hamiltonian $H$ is an anti-symmetric
Hermitian matrix.  In $d=1$, we can take the linear chain to give a lattice model with the Dirac equation as its continuum
limit.  That is, $H_{j,j+1}=i$ and $H_{j+1,j}=-i$.  To describe this state in
pictures, we draw these bonds as oriented lines,
as shown in Fig.~(\ref{figDirac}a), with the orientation indicating the
sign of the bond.  The continuum limit of this Hamiltonian is described by a Dirac equation with $2$-dimensional $\gamma$ matrices.
While this system can be described by a unit cell of a single site, we instead choose to describe it by a unit
cell of two sites for convenience when adding mass terms later.
In $d=2$, we can take a $\pi$-flux
state to obtain the Dirac equation in the continuum limit.
A convenient gauge to take to describe the $\pi$-flux state is shown in Fig.~(\ref{figDirac}b), with all the vertical bonds
having the same orientation, and the orientation of the horizontal bonds alternating from row to row.  The continuum limit
here has $4$-dimensional $\gamma$ matrices and we use a $4$-site unit cell.

\begin{figure}[t]
\centering
\includegraphics[width=3.5in]{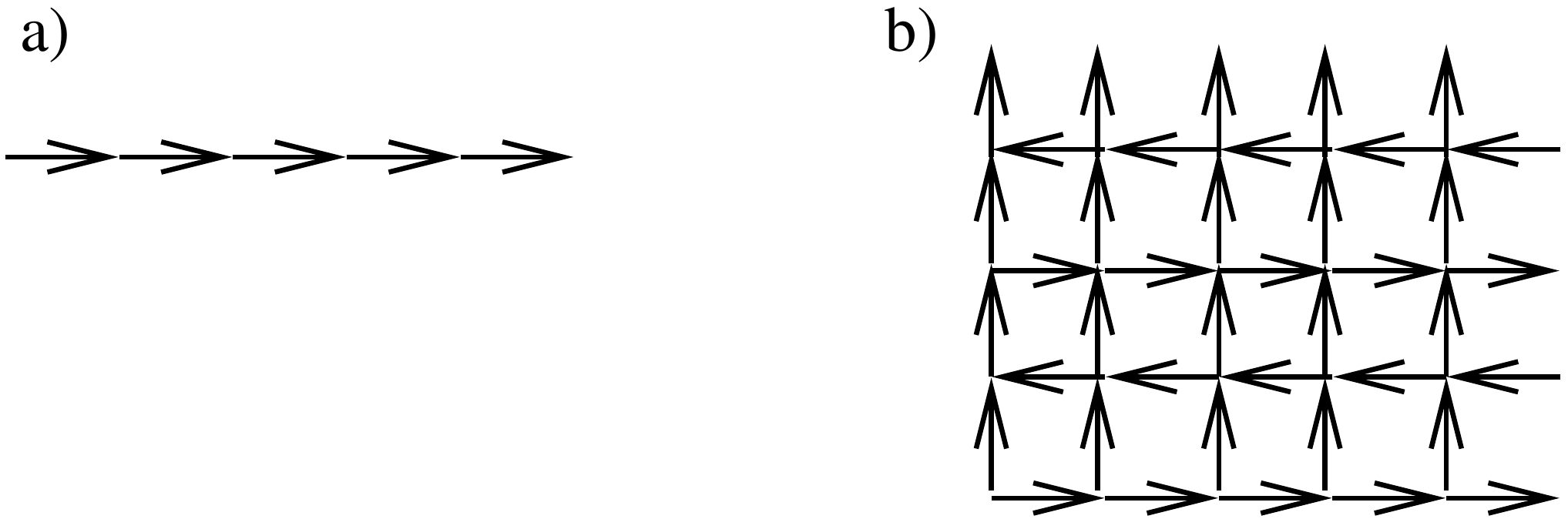}
\caption{(a) A lattice model giving the Dirac equation in $d=1$.
(b) A lattice model in $d=2$.}
\label{figDirac}
\end{figure}

In general, in $d$ dimensions, we can obtain a Dirac equation with $2^d$-dimensional $\gamma$ matrices
by the following iterative
procedure.  Let the ``vertical" direction refer to the direction of the $d$-th basis vector.
Having constructed the lattice Hamiltonian in $d-1$ dimensions, we stack these Hamiltonians vertically on top of each other,
with alternating signs in each layer.  Then, we take all the vertical bonds to be oriented in the
same direction.
This Hamiltonian is invariant under translation in the vertical direction by distance $2$.  Thus,
if $H_{d-1}$ is the Hamiltonian in $d-1$ dimensions, the Hamiltonian $H_d$ is given by
\be
H_d=\begin{pmatrix} H_{d-1} & 2\sin(k/2) I  \\ 2\sin(k/2) I& -H_{d-1} \end{pmatrix},
\ee
where $I$ is the identity matrix and $k$ is the momentum in the vertical direction.
Near $k=0$, this is
\be
\label{Hdcontinuum}
H_d \approx H_{d-1} \otimes \sigma_z + k \otimes \sigma_x.
\ee

This iterative construction corresponds to an iterative construction of $\gamma$-matrices.  Having constructed $d-1$
different $2^{d-1}$-dimensional $\gamma$-matrices $\gamma^{}_1,...,\gamma^{}_{d-1}$, we construct $d$ different $2^d$-dimensional $\gamma$-matrices,
$\tilde \gamma^{}_1,...,\tilde \gamma^{}_{d}$, by $\tilde \gamma^{}_i=\gamma^{}_i \otimes \sigma_z$ for $i=1,...,d-1$, and $\tilde \gamma^{}_d=I\otimes \sigma_x$.

In one dimension, dimerization of bonds corresponds to alternately strengthening and weakening the bonds as shown in
Fig.~(\ref{figDimer}).  In two dimensions, we can dimerize in either the horizontal or vertical directions.
In $d$-dimensions, we have $d$ different directions to dimerize.  Dimerizing in the ``vertical" direction gives, instead
of (\ref{Hdcontinuum}), the result
\be
\label{eqn:dimerizations}
H_d \approx H_{d-1} \otimes \sigma_z + k\otimes \sigma_x + m_d \otimes \sigma_y,
\ee
where $m_d$ is the dimerization strength.  This corresponds to an iterative construction of mass matrices, $M_i$, as follows.
In one dimension, we have $M_1=i\sigma_y$.  Given $d-1$ different mass matrices in $d-1$ dimensions, $M_i$, we construct $\tilde M_i$ in
$d$-dimensions by $\tilde M_i=M_i \otimes \sigma_z$, for $i=1...d-1$, and $\tilde M_d=iI\otimes \sigma_y$.

\begin{figure}
\centering
\includegraphics[width=3.25in]{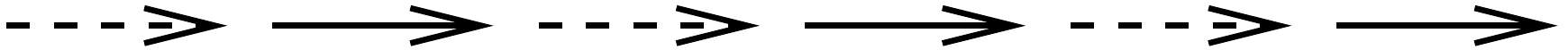}
\caption{Dimerization in $d=1$.}
\label{figDimer}
\end{figure}

If the dimerization is non-zero, and constant, we can increase the dimerization strength without closing the gap until a strong
coupling limit is reached.  In one dimension, by increasing the dimerization strength, we eventually
reach a fully dimerized configuration, in which each site has one non-vanishing bond
connected to it.
In two or more dimensions, the dimerization
can be a combination of dimerization in different directions.  However, if the dimerization is completely in one direction, for example
the vertical direction, we increase the dimerization strength until the vertical bonds are fully dimerized.  Simultaneously, we reduce the
strength of the other bonds to zero without closing the gap.
This is again a fully dimerized state, the columnar state, with each site having one non-vanishing bond.
Any configuration with uniform, small dimerization can be deformed into this pattern without closing the gap by rotating the
direction of dimerization, increasing the strength of dimerization, and then setting the bonds in the other directions to zero.

It is important to understand that the ability to reach such a strong coupling
limits depends on the perturbation of the Dirac equation that we consider;
for dimerization, it is possible to reach a strong coupling limit, while if
we had instead chosen to open a mass gap by adding, for example, diagonal
bonds with imaginary coupling to the two-dimensional Dirac equations, we
would open a mass gap by perturbing the Hamiltonian with the term $i\gamma_1\gamma_2$, and such a perturbation cannot be continued to the strong coupling
limit due to topological obstruction.

Further, if the dimerization is non-uniform then it may not be possible to reach a fully dimerized state without
having defect sites.  Consider the configurations in
Fig.~(\ref{figHedgehog}a) in $d=1$ and in Fig.~(\ref{figHedgehog}b) in $d=2$.  These are the strong coupling limits of the
hedgehog configuration,
and each contains a zero mode, a single unpaired site.
This is one of the central results of the strong-coupling
limit: {\it topological defects have unpaired sites which,
in turn, support Majorana zero modes}.

Such strong-coupling hedgehog configurations can be constructed by the following iterative process in any dimension $d$.
Let $x_d$ correspond to the coordinate in the vertical direction.  For $x_d\geq 0$, stack $d-1$-dimensional hedgehog configurations.
Along the half-line given by $x_d>0$ and $x_i=0$ for $1\leq i \leq d-1$, arrange vertical bonds, oriented to connect
the site with $x_d=2k-1$ to
that with $x_d=2k$, for $k\geq 1$.  Along the lower half plane, given by $x_d<0$, arrange vertical bonds oriented to connect
a site with $x_d=-(2k-1)$ to that with $-2k$, for $k\geq 1$.  This procedure gives the $d=2$ hedgehog in
Fig.~(\ref{figHedgehog}b) from the $d=1$ hedgehog in
Fig.~(\ref{figHedgehog}a), and gives a strong coupling limit of the
Teo-Kane hedgehog in $d=3$.  That is, the Teo-Kane
hedgehog can be deformed into this configuration,  without closing the gap.

\begin{figure}
\centering
\includegraphics[width=3.25in]{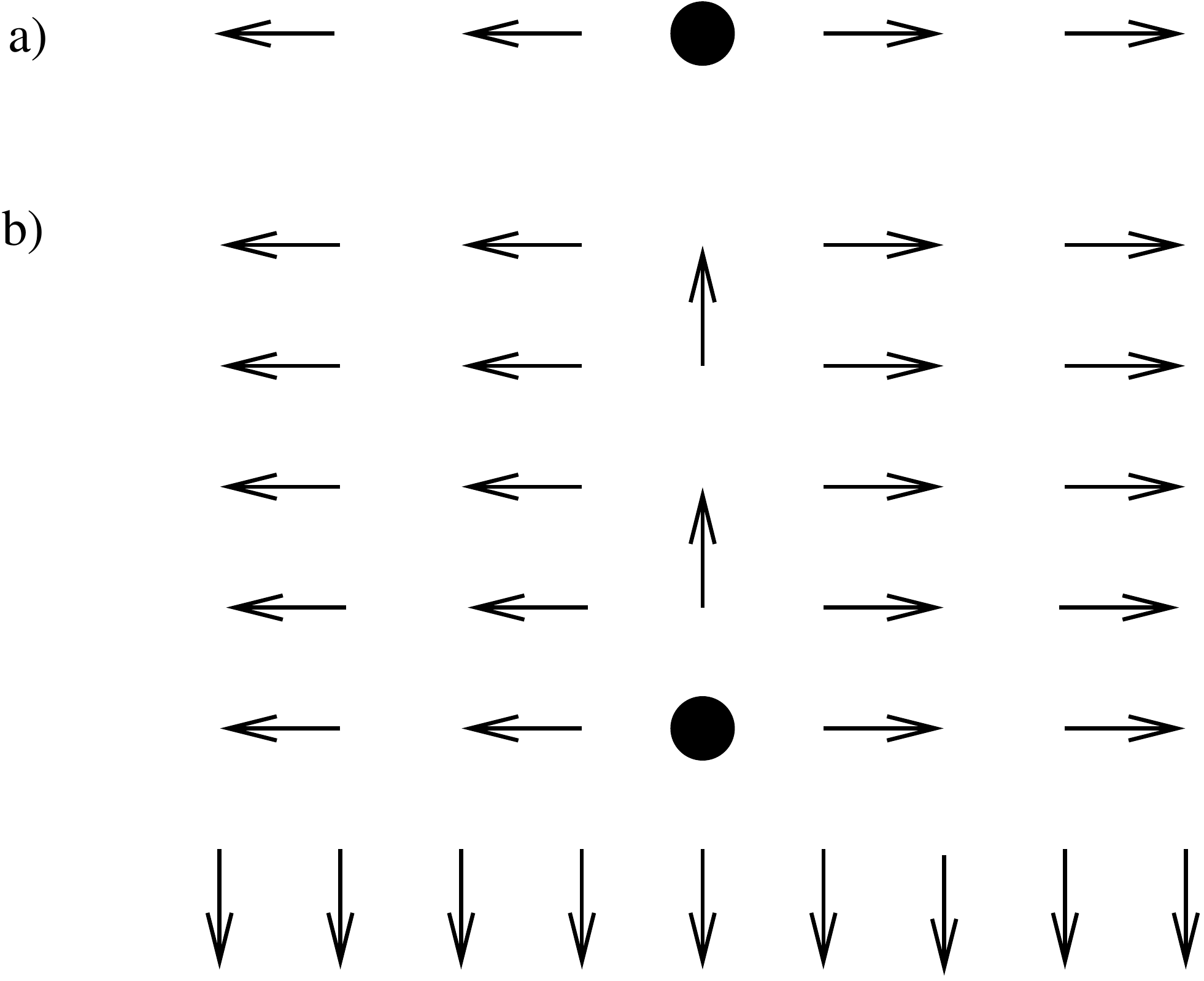}
\caption{(a) A one-dimensional hedgehog.
(b) A two-dimensional hedgehog.}
\label{figHedgehog}
\end{figure}

So long as we consider only nearest-neighbor bonds,
there is an integer index $\nu$ describing different dimerization
patterns in the strong-coupling limit.  This index, which
is present in any dimension, arises from the sublattice symmetry
of the system, and is closely-related to the U(1) symmetry
of dimer models of spin systems\cite{Rokhsar88}.  Label the two
sublattices by $A$ and $B$.  Consider any set of sites, such that every site in that
set has exactly one bond connected to it.  (Recall that,
in the strong coupling
limit, every bond has strength $0$ or $1$ and every site
has exactly one bond connected to it, except for defect sites.)
Then, the number of bonds going from $A$ sites in this set to $B$ sites outside the set
is exactly equal to the number of bonds going from $B$ sites in this set to $A$ sites outside the set.  On the other hand,
if there are defect sites in the set, then this rule is broken.
Consider the region defined by the dashed line
in Fig.~(\ref{figU1}a).  We define the ``flux" crossing the dashed line to be the number of bonds crossing that boundary which leave starting on an $A$ site, minus the number
which leave starting on a $B$ site. The flux is non-zero in this case, but is unchanging as we increase the size of the region.
This flux is the index $\nu$. By the argument given above for
the existence of zero modes, $\nu$ computed for any region
is equal to the number of Majorana zero modes contained
within the region.

The index $\nu$ can be defined beyond the strong-coupling
limit. Consider, for the sake of concreteness, $d=3$.
There are 3 possible dimerizations, one for each
dimension, as we concluded in Eq. \ref{eqn:dimerizations}.
In weak-coupling, the square of the gap is equal to the sum of
the squares of the dimerizations. Thus, if we assume
a fixed gap, we can model these dimerizations by a unit vector.
The integer index discussed above is simply the total winding
number of this unit vector on the boundary of any region.

However, once diagonal bonds are allowed,
the integer index $\nu$ no longer counts zero modes.
Instead, there is a $\mathbb{Z}_2$ index,
equal to $\nu(\text{mod} 2)$ which
counts zero modes modulo 2. To see this in
the strong-coupling limit,
consider the configuration in Fig.~(\ref{figU1}b).
This is a configuration with $\nu=2$
but no Majorana zero modes. However, a $\nu=1$
configuration must still have a zero mode and, thus,
any configuration with odd $\nu$ must have at least one zero mode.

\begin{figure}
\centering
\includegraphics[width=3.25in]{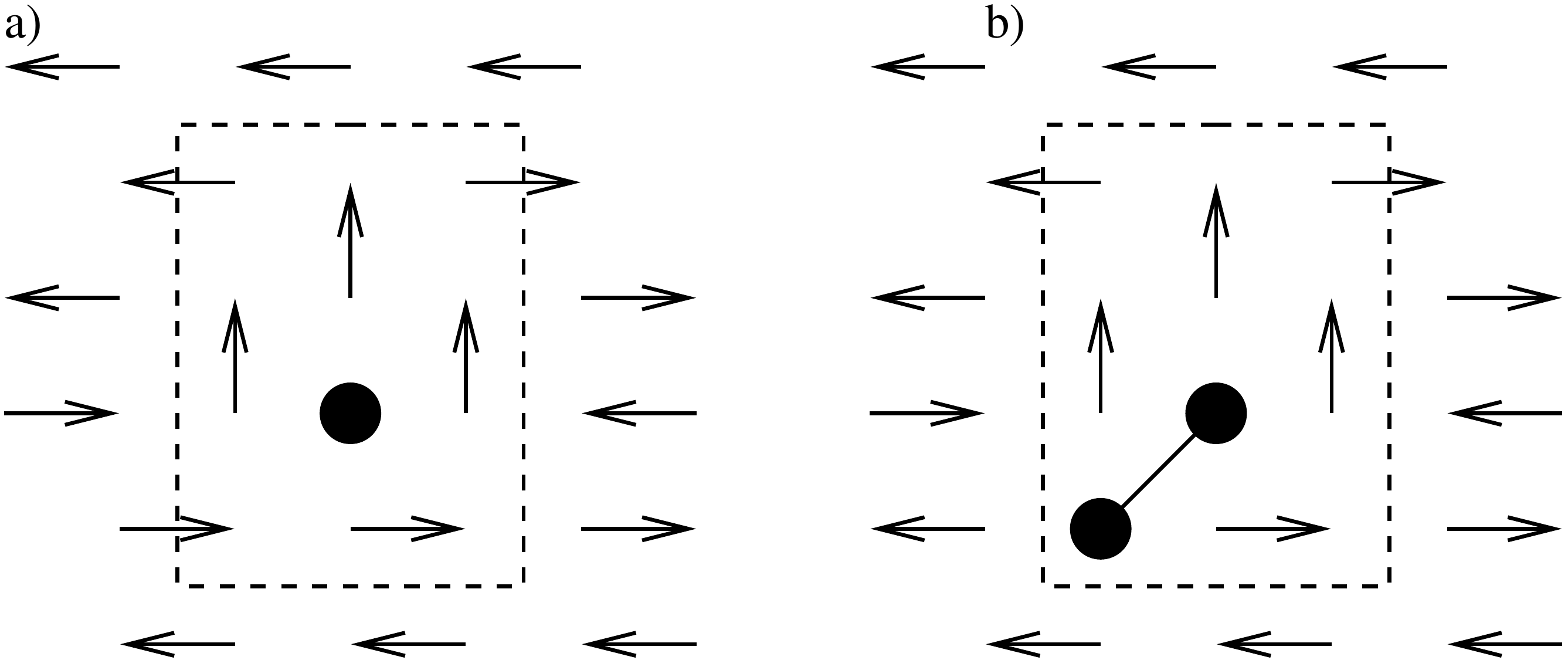}
\caption{(a) Defect acting as source of $U(1)$-flux.  Bonds are oriented from $A$ to $B$ sublattice.  There is a net
flux of one leaving the region defined by the dashed line.
(b) Configuration with diagonal bond added, indicated by the undirected line
connecting the two circles; either orientation of this line, corresponding to different choices of the sign of the
term in the Hamiltonian, would lead to the same result.  There is a net of flux of two leaving the region
defined by the dashed line.}
\label{figU1}
\end{figure}

In Fig.~(\ref{figU1}), we have chosen to orient the
bonds from A to B sublattice
to make it easier to compute $\nu$.  However,
the $\nu$ and its residue modulo 2,
defined above are independent of the orientation of
the bonds (which indicate the sign of terms in the Hamiltonian) and depend only on which sites are connected by bonds (which indicate which terms in the Hamiltonian are non-vanishing).

The $\nu(\text{mod} 2)$ with diagonal bonds
is the same as Kitaev's ``Majorana number"\cite{Kitaev06a}.
We can use this to show the existence of
zero modes in the Teo-Kane hedgehog even outside the
strong-coupling limit.
Consider a hedgehog configuration.
Outside some large distance $R$ from
the center of the hedgehog, deform to the
strong coupling limit without closing the gap.  Then, outside a distance $R$, we can count $\nu(\text{mod} 2)$ by counting bonds leaving
the region and we find a nonvanishing result relative to a reference configuration: if there are an even number of sites in the
region then there are an odd number of bonds leaving in a hedgehog configuration, and if there are an odd number of sites then there
are an even number of bonds leaving.
However, since this implies a nonvanishing
Majorana number, there must be a zero mode inside the region, regardless of what the Hamiltonian inside is.
We note that this is a highly non-trivial result
in the weak-coupling limit, where
the addition of weak diagonal bonds, all oriented
the same direction, to the configuration of Fig.~(\ref{figDirac}b)
corresponds to adding the term
$i\gamma^{}_1 \gamma^{}_2$ to the Hamiltonian in $d=2$.
By the argument given above, even this Hamiltonian
has a zero mode in the presence of a defect
with non-zero $\nu(\text{mod} 2)$.

Given any two zero modes, corresponding to defect sites in the strong coupling limit, we can identify a string of sites
connecting them.  If we have a pair of defect sites on opposite sublattices, corresponding to opposite hedgehogs, then one particular string
corresponds to the north pole of the order parameter, as in Fig.~(\ref{figstring}a).  However, we
can simply choose {\it any} arbitrary string.
Let $\gamma^{}_i,\gamma^{}_j$ be the Majorana operators at the two defect sites.  The operation $\gamma^{}_i\rightarrow -\gamma^{}_i,\gamma^{}_j\rightarrow
-\gamma^{}_j$ can be implemented as follows.  We begin with an adiabatic operation on one of the defect sites and the nearest $2$
sites on the line.  The Hamiltonian on those three sites is an anti-symmetric, Hermitian matrix.  That is, it corresponds to
an oriented plane in three dimensions.  We can adiabatically perform orthogonal rotations of this plane.  Thus, by rotating by $\pi$
in the plane corresponding to the defect site and the first site on the string,
we can change the sign of the mode on the defect and the orientation of the
bond, as shown in Fig.~(\ref{figstring}b).  This rotation is an adiabatic transformation
of the three site Hamiltonian
\be
\begin{pmatrix}
0 & 0 & i\sin(\theta) \\
0 & 0 & i\cos(\theta) \\
-i\sin(\theta) & -i\cos(\theta) & 0
\end{pmatrix}
\ee
along the path $\theta=0\rightarrow \pi$.
We then perform rotations on consecutive triples of sites along the defect line, which changes
the orientation of pairs of neighboring bonds, arriving
at the configuration in Fig.~(\ref{figstring}c).  Finally, we rotate by $\pi$
in the plane containing the other defect site and the last site.
This returns the system to the original configuration,
having effected the desired operation.

Since we only consider adiabatic transformation, we can only perform orthogonal rotations with unit determinant.  Thus, any transformation which swaps two defects and returns the
bonds to their original configuration, must change the sign of one of
the zero modes: $\gamma^{}_i \rightarrow \gamma^{}_j, \gamma^{}_j \rightarrow -\gamma^{}_i$.  Indeed, any orthogonal
transformation with determinant equal to minus one would change the sign of the fermion parity in the system, as the
fermion parity operator is equal to the product of the $\gamma_i$ operators.

\begin{figure}
\centering
\includegraphics[width=3.25in]{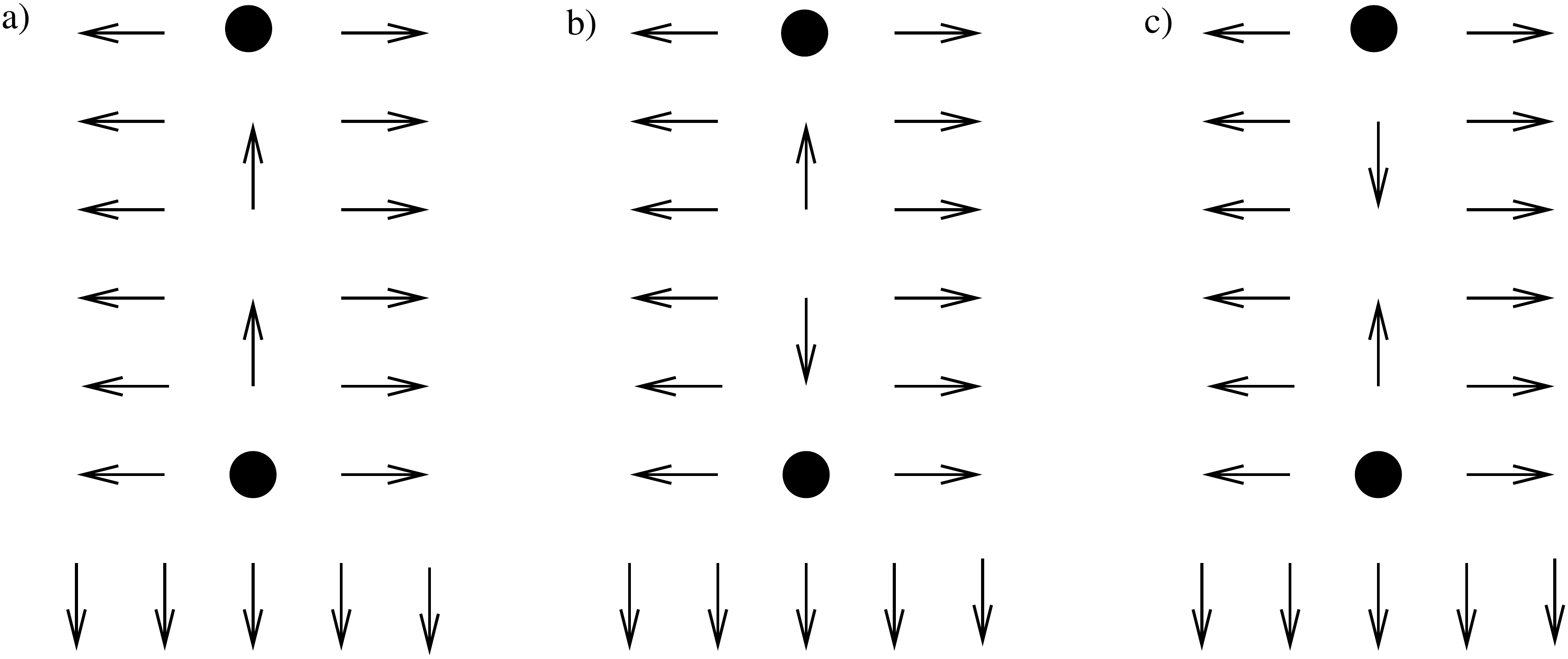}
\caption{(a) Pairs of defects connected by a string.
(b) First rotation applied to the configuration in (a)  Open circle replaces filled circle to indicate sign change of the Majorana mode on the site.  (c)After rotating along the string. (d) Rotating the last site and restoring the string to its original configuration}
\label{figstring}
\end{figure}

We used the ability to change the orientation of a pair of bonds in this
construction.
The fact that one can only change the
orientation of bonds in pairs, and not the
orientation of a single bond,
is related to a global $Z_2$
invariant: the Hamiltonian is an anti-symmetric matrix and the sign
of its Pfaffian cannot be changed without closing the gap.  Changing the direction of a single bond changes the sign of this Pfaffian and so is not possible.

The above discussion left open the question
of which zero changes its sign, i.e. is the
effect of the exchange  $\gamma^{}_i \rightarrow \gamma^{}_j, \gamma^{}_j \rightarrow -\gamma^{}_i$ or  $\gamma^{}_i \rightarrow -\gamma^{}_j, \gamma^{}_j \rightarrow \gamma^{}_i$? The answer is that it depends on how the bonds are returned to their original configuration
after the exchange is completed (which is a clue that
the defects must be understood as extended objects,
not point-like ones). For the bonds to be
restored, one of the defects must be rotated by
$2\pi$; the corresponding zero mode acquires a minus sign.
We will discuss this in greater detail in a later section.
The salient point here
is that the effect of an exchange is a unitary transformation
generated by the operator $e^{\pm\pi \gamma^{}_i \gamma^{}_j /4}$.
This is reminiscent of the representation of braid group
generators for non-Abelian quasiparticles in the quantum
Hall effect \cite{Nayak96c} and vortices in chiral
$p$-wave superconductors \cite{Ivanov01},
namely the braid group representation realized
by Ising anyons \cite{Nayak08}.
But, of course, in 3D the braid group is not
relevant, and the permutation group, which is associated with
point-like particles in $d>2$, does not have non-trivial
higher-dimensional representations consistent with
locality \cite{Doplicher71a,Doplicher71b}. As noted in the introduction,
this begs the question:
what group are the unitary matrices
$e^{\pm\pi \gamma^{}_i \gamma^{}_j /4}$ representing?

\section{Topological Classification of Gapped Free
Fermion Hamiltonians}
\label{sec:free-fermion}

\subsection{Setup of the Problem}

In this section, we will briefly review the topological
classification of translationally-invariant or slowly
spatially-varying free-fermion Hamiltonians following Kitaev's analysis
in Ref. \onlinecite{Kitaev09}. (For a different perspective,
see Schnyder {\it et al.}'s approach in Ref. \onlinecite{Ryu08}.)
The 3D Hamiltonian of the previous section is a specific example which
fits within the general scheme and, by implication,
the 3D non-Abelian statistics which we derived at
the end of the previous section also holds for an entire
class of models into which it can be deformed without
closing the gap. Our discussion will
follow the logic of Milnor's treatment of Bott periodicity
in Ref. \onlinecite{Milnor63}.

Consider a system of $N$ flavors of electrons
${c_j}({\bf k})$ in $d$ dimensions. The flavor index
$j$ accounts for spin as well as the possibility of multiple
bands. Since we will not be assuming charge conservation,
it is convenient to express the complex fermion operators
${c_j}({\bf k})$ in terms of real fermionic operators
(Majorana fermions),
${c_j}({\bf k})=(a_{2j-1}({\bf k})+ia_{2j}({\bf k}))/2$
(the index $j$ now runs from $1$ to $2N$).
The momentum ${\bf k}$ takes values in
the Brillouin zone, which has the topology of the
$d$-dimensional torus $T^d$. The Hamiltonian
may be written in the form
\begin{equation}
\label{eqn:basic}
H = \sum_{i,j,{\bf p}} iA_{ij}({\bf p}){a_i}({\bf p}){a_j}(-{\bf p})
\end{equation}
where, by Fermi statistics, $A_{ij}({\bf p})=-A_{ji}(-{\bf p})$.
Let us suppose that the Hamiltonian
(\ref{eqn:basic}) has an energy gap $2\Delta$,
by which we mean that its eigenvalues $E_{\alpha}(p)$
($\alpha$ is an index labeling the eigenvalues of
$H$) satisfy $|E_{\alpha}(p)|\geq \Delta$.
The basic question which we address in this section
is the following. What topological obstructions
prevent us from continuously deforming one
such gapped Hamiltonian into another?

Such an analysis can apply, as we will see,
not only to free fermion Hamiltonians, but also
to those interacting fermion Hamiltonians which,
deep within ordered phases, are well-approximated
by free-fermion Hamiltonians. (This can include
rather non-trivial phases such as Ising anyons,
but not Fibonacci anyons.) In such settings, the
fermions may be emergent fermionic quasiparticles;
if the interactions between these quasiparticles
are irrelevant in the renormalization-group sense,
then an analysis of free-fermion Hamiltonians
can shed light on the phase diagrams of
such systems. Thus, the analysis of free fermion
Hamiltonians is equivalent to the analysis of
{\it interacting fermion ground states} whose
low-energy quasiparticle excitations are free fermions.

Let us begin by considering a few concrete examples,
in order of increasing complexity.

\subsection{Zero-Dimensional Systems}

First, we analyze
a zero-dimensional system which we will not assume
to have any special symmetry. The Hamiltonian
(\ref{eqn:basic}) takes the simpler form:
\begin{equation}
\label{eqn:zero-dim}
H = \sum_{i,j} iA_{ij}{a_i}{a_j}
\end{equation}
where $A_{ij}$ is a $2N\times 2N$ antisymmetric matrix,
$A_{ij}=-A_{ji}$. Any real antisymmetric matrix can be written
in the form
\begin{equation}
A =  O^T \left( \begin{array}{rrrrr}
0 & -{\lambda_1} &   &  & \\
{\lambda_1} & 0  &    &  &  \\
   &    & 0  & -{\lambda_2}&  \\
   &    & {\lambda_2}  &  0&  \\
   &    &     &    &  \ddots\end{array} \right) O
\end{equation}
where $O$ is an orthogonal matrix and
the $\lambda_i$s are positive. The eigenvalues
of $A$ come in pairs $\pm i\lambda_i$;
thus, the {\it absence of charge conservation} can
also be viewed as the {\it presence of a particle-hole symmetry}.
By assumption, ${\lambda_i}\geq\Delta$ for all $i$.
Clearly, we can continuously deform $A_{ij}$
without closing the gap so that ${\lambda_i}=\Delta$ for all $i$.
(This is usually called `spectrum flattening'.)
Then, we can write:
\begin{equation}
A = \Delta\,\cdot \,{O^T} J O
\end{equation}
where
\begin{equation}
\label{eqn:canonical-J}
J =  \left( \begin{array}{rrrrr}
0 & -1 &   &  & \\
1 & 0  &    &  &  \\
   &    & 0  & -1&  \\
   &    & 1  &  0&  \\
   &    &     &    &  \ddots\end{array} \right)
\end{equation}
The possible choices of $A_{ij}$
correspond to the possible choices of $O\in\text{O}(2N)$,
modulo $O$ which commute with the matrix $J$.
But the set of $O\in\text{O}(2N)$ satisfying ${O^T} J O~=~J$
is U($N$)$\subset$O($2N$).
Thus, the space of all
possible zero-dimensional free fermionic Hamiltonians
with $N$ single-particle energy levels
is topologically-equivalent to the symmetric space O($2N$)/U($N$).

This can be restated in more geometrical terms
as follows. Let us here and henceforth
take units in which $\Delta=1$.
Then the eigenvalues of $A$ are $\pm i$.
If we view the $2N\times 2N$ matrix $A$
as a linear transformation on $\mathbb{R}^{2N}$, then
it defines a complex structure.
Consequently, we can view $\mathbb{R}^{2N}$ as
$\mathbb{C}^{N}$ since multiplication of
$\vec{v}\in \mathbb{R}^{2N}$ by a
complex scalar can be defined as
$(a+ib) \vec{v} \equiv a\vec{v} + b A\vec{v}$.
The set of complex structures on $\mathbb{R}^{2N}$
is given by performing an arbitrary O($2N$) rotation
on a fixed complex structure, modulo the rotations
of $\mathbb{C}^{N}$ which respect the complex structure,
namely U($N$). Thus, once again, we conclude that
the desired space of Hamiltonians is topologically-equivalent
to O($2N$)/U($N$).

What are the consequences of this equivalence?
Consider the simplest case, $N=1$. Then, the
space of zero-dimensional Hamiltonians is topologically-equivalent to
O($2$)/U($1$)$=\mathbb{Z}_2$: there are two classes
of Hamiltonians, those in which the single fermionic level
is unoccupied in the ground state,
$c^\dagger c = (1+i{a_1}{a_2})/2=0$,
and those in which it is occupied. For larger $N$,
O($2N$)/U($N$) is a more complicated space, but it
still has two connected components,
${\pi_0}(\text{O}(2N)/\text{U}(N))=\mathbb{Z}_2$,
so that there are two classes of free fermion Hamiltonians,
corresponding to even or odd numbers of occupied
fermionic levels in the ground state.

Suppose, now, that we restrict ourselves to time-reversal
invariant systems and, furthermore, to those time-reversal
invariant systems which satisfy ${T^2}=-1$, where
$T$ is the anti-unitary operator generating time-reversal.
Then, following Ref. \onlinecite{Kitaev09}, we
write $T {a_i} T^{-1} = ({J_1})_{ij}{a_j}$. The matrix
${J_1}$ is antisymmetric and satisfies ${J_1^2}=-1$.
$T$-invariance of the Hamiltonian requires
 \begin{equation}
{J_1} A = -A {J_1}
\end{equation}
Since ${J_1}$ is antisymmetric and satisfies ${J_1^2}=-1$,
its eigenvalues are $\pm i$.
Therefore, ${J_1}$ defines a complex structure
on $\mathbb{R}^{2N}$ which may, consequently,
be viewed as $\mathbb{C}^{N}$.
Now consider $A$, which is also
antisymmetric and satisfies ${A^2}=-1$, in addition
to anticommuting with $J_1$. It defines
a quaternionic structure on $\mathbb{C}^{N}$
which may, consequently, be viewed as
$\mathbb{H}^{N/2}$.
Multplication of $\vec{v}\in \mathbb{R}^{2N}$
by a quaternion can be defined as:
$(a+bi+cj+dk)\vec{v}\equiv a\vec{v}+b{J_1}\vec{v}
+cA\vec{v}+d {J_1}A\vec{v}$.
The possible choices of $A$
can be obtained from a fixed one
by performing rotations of $\mathbb{C}^{N}$,
modulo those rotations which respect the quaternionic structure,
namely Sp($N/2$).
Thus, the set of time-reversal-invariant zero-dimensional free
fermionic Hamiltonians with ${T^2}=-1$
is topologically-equivalent to U($N$)/Sp($N/2$).
Since ${\pi_0}(\text{U}(N)/\text{Sp}(N/2))=0$,
any such Hamiltonian can be continuously
deformed into any other. This can be understood
as a result of Kramers
doubling: there must be an even number
of fermions in the ground state so the division
into two classes of the previous case does not exist here.

\subsection{2D Systems: $T$-breaking
superconductors}

Now, let us consider systems in more than
zero dimensions. Once again, we will assume
that charge is not conserved, and we will also
assume that time-reversal symmetry is not
preserved. For the sake of
concreteness, let us consider a single band of
spin-polarized electrons on a two-dimensional lattice.
Let us suppose that the electrons condense into
a (fully spin-polarized) $p_x$-wave superconductor.
For fixed superconducting order parameter,
the low-energy theory is a free fermion Hamiltonian
for gapless fermionic excitations at the nodal
points $\pm\vec{k}_{F}\equiv (0,\pm p_F)$. We now ask the question,
what other order parameters could develop
which would fully gap the fermions? For fixed values
of these order parameters, we have a free fermion
Hamiltonian. Thus, these different possible order
parameters correspond to different
possible gapped free fermion Hamiltonians.

The low-energy Hamiltonian of a fully
spin-polarized $p_x$-wave superconductor
can be written in the form:
\begin{equation}
H = {\psi^\dagger}\left( i{v_\Delta}{\partial_x}{\tau_x}
+ i{v_F}{\partial_y}{\tau_z}\right)\psi
\end{equation}
where $v_F$, ${v_\Delta}$ are, respectively, the
Fermi velocity and slope of the gap near the node.
The Pauli matrices $\tau$ act in the particle-hole
space:
\begin{eqnarray}
\psi(k)
\equiv \left(
\begin{array}{c}
c_{\vec{k}_{F}+\vec{k}}\\
c^{\dagger}_{-\vec{k}_{F}+\vec{k}}
\end{array}
\right)
\end{eqnarray}
This Hamiltonian is invariant under the U(1):
$\psi\rightarrow e^{i\theta} \psi$ which corresponds
to conservation of momentum in the $p_y$ direction
(not to charge conservation). Since we will be considering
perturbations which do not respect this symmetry,
it is convenient to introduce Majorana fermions
${\chi^{}_1}$, ${\chi^{}_2}$ according to
$\psi={\chi^{}_1}+i{\chi^{}_2}$. Then
\begin{equation}
H = i{\chi^{}_a}\left( {v_\Delta}{\partial_x}{\tau_x}
+ {v_F}{\partial_y}{\tau_z}\right){\chi^{}_a}
\end{equation}
with $a=1,2$. Note that we have suppressed the
particle-hole index on which the Pauli matrices
$\tau$ act. Since $\chi^{}_1$, $\chi^{}_2$ are each
a 2-component real spinor, this model has
4 real Majorana fields.

We now consider the possible mass terms which
we could add to make this Hamiltonian fully gapped:
\begin{equation}
\label{eqn:Dirac+mass}
H = i{\chi^{}_a}\left( {v_\Delta}{\partial_x}{\tau_x}
+ {v_F}{\partial_y}{\tau_z}\right){\chi^{}_a} + i{\chi^{}_a}M_{ab}{\chi^{}_b}
\end{equation}
If we consider the possible order parameters
which could develop in this system, it is clear
that there are only two choices: an imaginary
superconducting order parameter $ip_y$ (which
breaks time-reversal symmetry) and
charge density-wave order (CDW). These take the form:
\begin{equation}
M^{ip_y}_{ab} = \Delta_{ip_y}\,i{\tau^y} \delta_{ab}
\end{equation}
and
\begin{equation}
\label{eqn:CDW-mass-eg}
M^{CDW}_{ab} = \rho^{}_{2k_F}
{\tau^y} \left(\cos\theta\,  \mu^z_{ab}+
\sin\theta\, \mu^x_{ab}\right)
\end{equation}
where $\mu^{x,z}$ are Pauli matrices and $\theta$
is an arbitrary angle.
For an analysis of the possible mass terms
in the more complex situation
of graphene-like systems, see, for instance,
Ref. \onlinecite{Ryu09}.

Let us consider the space of mass terms
with a fixed energy gap $\Delta$ which is
the same for all 4 of the Majorana fermions in the model
(i.e. a flattened mass spectrum).
An arbitrary gapped Hamiltonian can be continuously
deformed to one which satisfies this condition.
Then we can have
$\Delta_{ip_y}=\pm\Delta$,
$\rho^{}_{2k_F}=0$
or $\rho^{}_{2k_F}=\Delta$, $\Delta_{ip_y}=0$
(in the latter case, arbitrary $\theta$ is allowed).
If both order
parameters are present, then the energy gap
is not the same for all fermions.
It's not that there's anything wrong with
such a Hamiltonian -- indeed, one can imagine a
system developing both kinds of order.
Rather, it is that such a Hamiltonian
can be continuously deformed to one with
either $\Delta_{ip_y}=0$ or $\rho^{}_{2k_F}=0$
without closing the gap. For instance, if
$\Delta_{ip_y}>\rho^{}_{2k_F}$,
then the Hamiltonian can be continuously deformed to
one with $\rho^{}_{2k_F}=0$. (However if we try to deform
it to a Hamiltonian with $\Delta_{ip_y}=0$, the gap will
close at $\Delta_{ip_y}=\rho^{}_{2k_F}$.)
Hence, we conclude that the space of possible mass terms is
topologically-equivalent to the disjoint union
U($1$)$\cup\mathbb{Z}_2$: a single one-parameter family
and two disjoint points.

Since ${\pi_0}(\text{U}(1)\cup\mathbb{Z}_2) = \mathbb{Z}_3$,
there are three distinct classes of quadratic Hamiltonians
for $4$ flavors of Majorana fermions in $2D$.
The one-parameter family of CDW-ordered
Hamiltonians counts as a single class since
they can be continuously deformed into each other.
The parameter $\theta$ is the phase of the CDW,
which determines whether the density is maximum
at the sites, the midpoints of the bonds, or somewhere
in between. It is important to keep in mind, however,
that, although there is no topological obstruction to
continuously deforming one $\theta$ into another,
there may be an energetic penalty which makes it costly.
For instance, the coupling of the system to the lattice may prefer
some particular value of $\theta$.
The classification discussed here accounts only for topological
obstructions; the possibility of energetic barriers must
be analyzed by different methods.

We can restate the preceding analysis in the following,
more abstract language. This formulation will
make it clear that we haven't overlooked a
possible mass term and will generalize to more
complicated free fermion models.
Let us write
${\gamma^{}_1}={\tau_x}\delta_{ab}$,
${\gamma^{}_2}={\tau_z}\delta_{ab}$.
Then
\begin{equation}
\{{\gamma^{}_i},{\gamma^{}_j}\}=2\delta_{ij}
\end{equation}
The Dirac Hamiltonian for $N=4$
Majorana fermion fields takes the form
\begin{equation}
\label{eqn:Dirac-eqn-generic}
H = i\chi({\gamma_i}{\partial_i} + M)\chi
\end{equation}
The matrix $M$ plays the role that the
matrix $A$ did in the zero-dimensional case.
As in that case, we assume a flattened spectrum
which here means that each Majorana fermion
field has the same gap and that this gap is
equal to $1$. (It does {\it not} mean that the energy
is independent of the momentum ${\bf k}$.)
In order to satisfy this, we must require that
\begin{equation}
\{{\gamma^{}_i},M\}=0 {\hskip 0.3 cm} \text{and}
{\hskip 0.3 cm} {M^2}=-1
\end{equation}

Note that it is customary to write the Dirac Hamiltonian
in a slightly different form,
\begin{equation}
\label{eqn:Dirac-to-conventional}
H = \overline{\psi}(i{\gamma_i}{\partial_i} + m)\psi
\end{equation}
which can be massaged into the form of (\ref{eqn:Dirac-eqn-generic})
using $ \overline{\psi}=\psi^\dagger \gamma_0$:
\begin{eqnarray}
H &=& {\psi^\dagger}(i{\gamma_0}{\gamma_i}{\partial_i} +
m{\gamma_0})\psi\cr
&=&  {\psi^\dagger}(i{\alpha_i}{\partial_i} +
m\beta)\psi\cr
&=&i{\psi^\dagger}({\alpha_i}{\partial_i} -
im\beta)\psi
\end{eqnarray}
where ${\alpha_i}={\gamma_0}{\gamma_i}$ and
$\beta={\gamma_0}$. Thus, if we write
${\gamma_i}\equiv {\alpha_i}$ and $M\equiv
-im\beta$ and consider Majorana fermions
(or decompose Dirac fermions into Majoranas),
we recover (\ref{eqn:Dirac-eqn-generic}).
We have used the form (\ref{eqn:Dirac-eqn-generic})
so that it is analogous to (\ref{eqn:zero-dim}), with
$({\gamma_i}{\partial_i} + M)$ replacing $A_{ij}$
and the $i$ pulled out front. Then, the matrix $M$
determines the gaps of the various modes
in the same way as $A$ does in the zero-dimensional
case. Similarly, assuming a `flattened' spectrum
leads to the condition ${M^2}=-1$.

How many ways can we
choose such an $M$? Since ${\gamma^2_2}=1$,
its eigenvalues are $\pm 1$. Hence, viewed as a
linear map from $\mathbb{R}^4$ to itself, this matrix
divides $\mathbb{R}^4$ into two 2D subspaces
$\mathbb{R}^4={X_+}\oplus{X_-}$
with eigenvalue $\pm 1$ under
${\gamma^{}_2}$, respectively.
For ${\gamma^{}_2}={\tau_z}\delta_{ab}$,
this is trivial:
\begin{equation}
{X_+}= \text{span}\left\{
\left(\scriptstyle{\begin{array}{c}
1\\ 0 \end{array}} \right) \otimes  \left(\scriptstyle{\begin{array}{c}
1\\ 0 \end{array}}\right),
 \left(\scriptstyle{\begin{array}{c}
1\\ 0 \end{array}} \right) \otimes  \left(\scriptstyle{\begin{array}{c}
0\\ 1 \end{array}}\right)\right\}
\end{equation}
where $\tau_z$ acts on the
first spinor and the second spinor is indexed by $a=1,2$,
i.e. is acted on by the Pauli matrices $\mu^{x,z}$
in (\ref{eqn:CDW-mass-eg}).
This construction generalizes straightforwardly
to arbitrary numbers $N$ of Majorana fermions, which is
why we use it now.

Now ${\gamma^{}_1} M$ commutes with
${\gamma^{}_2}$ and satisfies ${({\gamma^{}_1} M)^2}=1$.
Thus, it maps ${X_+}$ to itself and defines
subspaces ${X_+^1}, {X_+^2}$
with eigenvalue $\pm 1$ under ${\gamma^{}_1} M$
(and equivalently for ${X_-}$). ${X_+}$ can decomposed
into ${X_+^1}\oplus{X_+^2}={X_+}$.
Choosing $M$ is thus equivalent to choosing
a linear subspace ${X_+^1}$ of ${X_+}$.

This can be divided into three cases.
If ${\gamma^{}_1} M$
has one positive eigenvalue and one negative one
when acting on ${X_+}$ then the space of possible choices
of ${\gamma^{}_1} M$ is equal to the space of 1D linear subspaces
of $\mathbb{R}^2$, which is simply U(1). If, on the other hand,
${\gamma^{}_1} M$ has two positive eigenvalues, then
there is a unique choice, which is simply
$M={\gamma^{}_1}{\gamma^{}_2}$. If ${\gamma^{}_1} M$
has two negative eigenvalues, then
there is again a unique choice,
$M=-{\gamma^{}_1}{\gamma^{}_2}$.
Therefore, the space of possible $M$s is topologically
equivalent to $\text{U}(1)\cup\mathbb{Z}_2$.

Now, suppose that we have $2N$ Majorana fermions.
Then $\gamma^{}_2$ defines $N$-dimensional
eigenspaces
${X_+},{X_-}$ such that
$\mathbb{R}^{2N}={X_+}\oplus{X_-}$
and ${\gamma^{}_1} M$ defines eigenspaces
of ${X_+}$: ${X_+^1}\oplus{X_+^2}={X_+}$.
If ${\gamma^{}_1} M$ has $k$ positive eigenvalues
and $N-k$ negative ones, then the space
of possible choices of ${\gamma^{}_1}M$
is O(N)/O(k)$\times$O(N-k), i.e we can
take the restriction of ${\gamma^{}_1}M$
to ${X_+}$ to be of the form
\begin{equation}
{\gamma^{}_1}M = O^T  \left( \begin{array}{cccccc}
1 &  &   &  & & \\
 & \ddots  &    &  &  &\\
   &    & 1  & &  &\\
   &    &   &  -1&  &\\
   &    &     &    &\ddots &\\
      &    &     &    &  & -1  \end{array} \right) O
\end{equation}
with $k$ diagonal entries equal to $+1$
and $N-k$ entries equal to $-1$. Thus, the space of Hamiltonians for
$N$ flavors of free Majorana fermions
is topologically equivalent to
\begin{equation}
\label{eqn:BO-def}
{\cal M}_{2N} = \bigcup_{k=0}^{N} \text{O}(N)/(\text{O}(k)\times
\text{O}(N-k))
\end{equation}
However, since ${\pi_0}(\text{O}(N)/(\text{O}(k)\times\text{O}(N-k)))
=0$, independent of $k$ (note that $0$ is the group with a single
element, not the empty set $\emptyset$),
${\pi_0}({\cal M}_{2N})=\mathbb{Z}_{N+1}$.

In the model analyzed above, we had
only a single spin-polarized band of electrons.
By increasing the number of bands and allowing
both spins, we can increase the number of
flavors of Majorana fermions. In principle, the number of
bands in a solid is infinity. Usually, we can introduce
a cutoff and restrict attention to a few bands
near the Fermi energy. However, for a purely
topological classification, we can ignore energetics
and consider all bands on equal footing.
Then we can take $N\rightarrow\infty$,
so that ${\pi_0}({\cal M}_\infty)=\mathbb{Z}$.
This classification permits us to
deform Hamiltonians into each
other so long as there is no topological obstruction,
with no regard to how energetically costly
the deformation may be. Thus, the $2N=4$ classification
which we discussed above can perhaps be viewed as
a `hybrid' classification which looks for topological
obstructions in a class of models with a fixed set of
bands close to the Fermi energy.

But even this point of view is not really
natural. The discussion above took as
its starting point an expansion about
a $p_x$ superconductor; the $p_x$ superconducting
order parameter was assumed to be large and
fixed while the $ip_y$ and CDW order parameters
were assumed to be small. In other words, we
assumed that the system was at a point in parameter
space at which the gap, though non-zero, was small
at two points in the Brillouin zone (the intersection
points of the nodal line in the $p_x$ superconducting
order parameter with the Fermi surface). This allowed
us to expand the Hamiltonian about these points
in the Brillouin zone and write it in Dirac form.
And this may, indeed, be reasonable in a system
in which $p_x$ superconducting order is strong
(i.e. highly energetically-favored) and other orders
are weak. However, a topological classification should
allow us to take the system into regimes in which
$p_x$ superconductivity is small and other orders
are large. Suppose, for instance, that we took
our model of spin-polarized electrons (which
we assume, for simplicity, to be at half-filling
on the square lattice) and went into a regime in
which there was a large
$d_{{x^2}-{y^2}}$-density-wave (or `staggered
flux') order parameter \cite{Nayak00b}
$\langle c^\dagger_{{\bf k}+{\bf Q}} c_{\bf k} \rangle
= i\Phi(\cos{k_x}a-\cos{k_y}a)$, where $a$ is the lattice
constant and $\Phi$ is the magnitude of the order parameter.
With nearest-neighbor hopping only,
the energy spectrum is $E^2_{\bf k}=
(2t)^2 (\cos{k_x}a+\cos{k_y}a)^2
+ \Phi^2 (\cos{k_x}a-\cos{k_y}a)^2$.
Thus, the gap vanishes at 4 points,
$(\pm\pi/2,\pm\pi/2)$ and $(\mp\pi/2,\pm\pi/2)$.
The Hamiltonian can be linearized in the
vicinity of these points:
\begin{multline}
H = {\psi^\dagger_1}\left( i{v_\Delta}{\partial_x}{\tau_x}
+ i{v_F}{\partial_y}{\tau_z}\right){\psi_1}\\
+ {\psi^\dagger_2}\left( i{v_\Delta}{\partial_y}{\tau_x}
+ i{v_F}{\partial_x}{\tau_z}\right){\psi_2}
\end{multline}
where $v_F$, ${v_\Delta}$ are, respectively, the
Fermi velocity and slope of the gap near the nodes;
the subscripts 1,2 refer to the two sets of nodes
$(\pm\pi/2,\pm\pi/2)$ and $(\mp\pi/2,\pm\pi/2)$;
and $\psi_{A}$, $A=1,2$ are defined by:
\begin{eqnarray}
\psi_{1,2}(k)
\equiv \left(
\begin{array}{c}
c_{(\pi/2,\pm\pi/2)+\vec{k}}\\
c_{(-\pi/2,\mp\pi/2)+\vec{k}}
\end{array}
\right)
\end{eqnarray}
If we introduce Majorana fermions
$\psi_{A}=\chi^{}_{A1}+i\chi^{}_{A2}$, then
we can write this Hamiltonian with possible
mass terms as:
\begin{multline}
H = i\chi^{}_{1a}\left( {v_\Delta}{\partial_x}{\tau_x}
+ {v_F}{\partial_y}{\tau_z}\right)\chi^{}_{1a}\\
+ i\chi^{}_{2a}\left( {v_\Delta}{\partial_y}{\tau_x}
+ {v_F}{\partial_x}{\tau_z}\right)\chi^{}_{2a}\\
+ i\chi^{}_{Aa}\, M_{Aa,Bb} \,  \chi^{}_{Bb}
\end{multline}
We have suppressed the spinor indices
(e.g. $\chi_{11}$ is a two-component spinor);
with these indices included, $M_{Aa,Bb}$
is an $8\times 8$ matrix. However, in order for
the gap to be the same for all flavors, the
mass matrix must anticommute with $\tau_{x,z}$.
Thus, $M_{Aa,Bb}={\tau_y} {\tilde M}_{Aa,Bb}$.
The matrix ${\tilde M}_{Aa,Bb}$ can have
$0,1,2,3$, or $4$ eigenvalues equal to $+1$
(with the rest being $-1$). The spaces
of such mass terms are, respectively,
$0$, $\text{O}(4)/(\text{O}(1)\times\text{O}(3))$,
$\text{O}(4)/(\text{O}(2)\times\text{O}(2))$,
$\text{O}(4)/(\text{O}(3)\times\text{O}(1))$,
and $0$. Mass terms with $0$ or $4$ eigenvalues
equal to $+1$ correspond physically to
$\pm id_{xy}$-density wave order,
$\langle c^\dagger_{{\bf k}+{\bf Q}} c_{\bf k} \rangle
= \pm\sin{k_x}a\,\sin{k_y}a$.
Mass terms with $2$ eigenvalues
equal to $+1$ correspond physically to
superconductivity, to $Q'=(\pi,0)$ CDW order,
and to linear combinations of the two.
Mass terms with $1$ or $3$ eigenvalues
equal to $+1$ correspond to
(physically unlikely) hybrid orders with, for instance,
superconductivity at $(\pm\pi/2,\pm\pi/2)$ and
$\pm id_{xy}$-density wave order at $(\pm\pi/2,\mp\pi/2)$.
Clearly, this is the $2N=8$ case of the general
classification discussed above. Thus, the same
underlying physical degrees of freedom -- a single
band of spin-polarized electrons on a
square lattice -- can correspond to either
$2N=4$ or $2N=8$, depending on where
the system is in parameter space. One
can imagine regions of parameter space where
the gap is small at an arbitrary number $N$
of points. Thus, if we restrict ourselves to systems
with a single band, then different regions of the parameter
space (with different numbers of points at which the gap is
small) will have very different topologies.
Although such a classification may be a necessary evil
in some contexts, it is far preferable, given the choice,
to allow topology to work unfettered by energetics.
Then, we can consider a large number $n$ of bands.
Suppose that the gap becomes small at $r$ points
in the Brillouin zone in each band. Then, the low-energy
Hamiltonian takes the Dirac form for $2N=2rn$
Majorana fermion fields.
As we will see below, if $N$ is sufficiently large,
the topology of the space of possible mass terms
will be independent of $N$. Consequently,
for $n$ sufficiently large, the topology of the space
of possible mass terms will be independent of $r$.
In other words, we are in the happy situation
in which the topology of the space of Hamiltonians
will be the same in the vicinity of any gap closing.
But any gapped Hamiltonian can be continuously
deformed so that the gap becomes small at some
points in the Brillouin zone. Thus, the problem of classifying
gapped free fermion Hamiltonians in $d$-dimensions
is equivalent to the problem of classifying possible mass terms
in a generic $d$-dimensional Dirac Hamiltonian
so long as the number of bands is sufficiently large \cite{Kitaev09}.
This statement can be made more precise and
put on more solid mathematical footing using ideas
which we discuss in Appendix \ref{sec:dimension}.

\subsection{Classification of Topological Defects}

The topological classification described above
holds not only for classes of translationally-invariant
Hamiltonians such as (\ref{eqn:Dirac-eqn-generic}),
but also for topological defects within a class.
Suppose, for instance, that we consider
(\ref{eqn:Dirac-eqn-generic}) with a mass
which varies slowly as the origin is encircled
at a great distance. We can ask whether such a
Hamiltonian can be continuously deformed into
a uniform one. In a system in which the mass term
is understood as arising as a result of some
kind of underlying ordering such as superconductivity or
CDW order, we are simply talking about topological
defects in an ordered media, but with the caveat that
the order parameter is allowed to explore a very large
space which may include many physically
distinct or unnatural orders, subject only to the condition that
the gap not close.

Let us, for the sake of concreteness,
assume that we have a mass term with $N/2$
positive eigenvalues when restricted to the $+1$ eigenspace of
$\gamma^{}_2$. (For $N$
large, the answer obviously cannot depend on the
number of positive eigenvalues $k$ so long
as $k$ scales with $N$. Thus, we will denote the
space ${\cal M}_{2N}$ defined in Eq. \ref{eqn:BO-def}
by $\mathbb{Z}\times\text{O}(N)/(\text{O}(N/2)\times\text{O}(N/2))$
where the integers in $\mathbb{Z}$ correspond to the number
of positive eigenvalues of the mass term when restricted
to the $+1$ eigenspace of $\gamma^{}_2$.)
Then $M(r=\infty,\theta)$ defines a loop in
$\text{O}(N)/(\text{O}(N/2)\times\text{O}(N/2))$
which cannot be continuously unwound if it
corresponds to a nontrivial element of
${\pi_1}(\text{O}(N)/(\text{O}(N/2)\times\text{O}(N/2)))$.

To compute ${\pi_1}(\text{O}(N)/(\text{O}(N/2)\times
\text{O}(N/2)))$, we parametrize
$\text{O}(N)/(\text{O}(N/2)\times\text{O}(N/2))$
by symmetric matrices $K$ which satisfy ${K^2}=1$
and $\text{tr}(K)=0$. (Such matrices decompose
$\mathbb{R}^{N}$ into their
$+1$ and $-1$ eigenspaces:
$\mathbb{R}^{N}={V_+}\oplus{V_-}$.
$K$ can be written in the form:
$K={O^T} {K_0} O$,
where $K_0$ has $N/2$ diagonal entries
equal to $+1$ and $N/2$ equal to $-1$, i.e
${K_0}=\text{diag}(1,\ldots,1,-1,\ldots,-1)$.)
Note that any such $K$ is itself an orthogonal matrix,
i.e. an element of O$(N)$;
thus $\text{O}(N)/(\text{O}(N/2)\times\text{O}(N/2))$
can be viewed as a submanifold of O$(N)$
in a canonical way. Consider a curve $L(\lambda)$
in $\text{O}(N)/(\text{O}(N/2)\times\text{O}(N/2))$
with $L(0)=K$ and $L(\pi)=-K$.
We will parametrize it as $L(\lambda)=K\, e^{\lambda A}$,
where $A$ is in the Lie algebra of $\text{O}(N)$.
In order for this loop to remain in
$\text{O}(N)/(\text{O}(N/2)\times\text{O}(N/2))$,
we need ${(K\, e^{\lambda A})^2}=1$.
Since ${(K\, e^{\lambda A})^2}=K\, e^{\lambda A} K\, e^{i\lambda A}
= e^{i\lambda K A K} \, e^{\lambda A}$, this condition
implies that $KA=-AK$.
In order to have $L(\pi)=-K$, we need ${A^2}=-1$.
Such a curve is, in fact, a minimal
geodesic from $K$ to $-K$.
Each such geodesic can be represented by its
midpoint $L(\pi/2)=KA$, so the space of such geodesics is
equivalent to the space of matrices $A$ satisfying
${A^2}=-1$ and $KA=-AK$. As discussed in
Ref. \onlinecite{Milnor63}, the space of minimal geodesics
is a good enough approximation to the entire space of loops
(essentially because an arbitrary loop can
be deformed to a geodesic) that we can compute
$\pi_1$ from the space of minimal geodesics
just as well as from the space of loops. Thus,
the loop space of $\text{O}(N/2)/(\text{O}(k)\times\text{O}(N/2-k))$
is homotopically equivalent to the space of matrices $A$ satisfying
${A^2}=-1$ and $KA=-AK$.
Since it anticommutes with $K$, $KA$
maps the $+1$ eigenspace of $K$ to
the $-1$ eigenspace. It is clearly a length-preserving map
since ${(KA)^2}=1$
and, since the $\pm 1$ eigenspaces of $K$ are isomorphic to
$\mathbb{R}^{N/2}$, $KA$ defines an element of
O$(N/2)$. Thus a loop in
$\text{O}(N)/(\text{O}(N/2)\times\text{O}(N/2))$
corresponds to an element of O$(N/2)$
or, in other words:
\begin{equation}
\label{eqn:Bott-step-1}
{\pi_1}(\text{O}(N)/(\text{O}(N/2)\times\text{O}(N/2)))=
{\pi_0}(\text{O}(N/2)).
\end{equation}
The latter group is simply $\mathbb{Z}_2$
since $\text{O}(N/2)$ has two connected components:
(1) pure rotations and (2) rotations combined with a reflection.

It might come as a surprise that we find a $\mathbb{Z}_2$
classification for point-like defects in two dimensions.
Indeed, if we require that the superconducting order parameter
has fixed amplitude at infinity, then vortices
of arbitrary winding number are stable and we have a
$\mathbb{Z}$ classification. However,
in the classification discussed here, we require a weaker
condition be satisfied: that the fermionic gap remain constant.
Consequently, a vortex configuration
of even winding number can be unwound without
closing the free fermion gap by, for instance, `rotating'
superconductivity into CDW order.

\subsection{3D Systems with No Symmetry}

With these examples under our belts, we now turn to
the case which is of primary interest in this paper:
free fermion systems in three dimensions
without time-reversal or charge-conservation symmetry.
We consider the Dirac Hamiltonian in $3D$
for an $2N$-component Majorana fermion field
$\chi$:
\begin{equation}
\label{eqn:3D-Dirac+mass}
H = i\chi\left( \partial_{1}{\gamma^{}_1}
+ \partial_{2}{\gamma^{}_2} +
\partial_{3}{\gamma^{}_3}\right)\chi + i{\chi}M{\chi}
\end{equation}
In the previous section, we discussed
a lattice model which realizes (\ref{eqn:3D-Dirac+mass})
in its continuum limit with
$2N=8$. Different mass terms
correspond to different quadratic perturbations
of this model which open a gap (which can be viewed
as order parameters which we are turning on at the mean-field level).
We could classify such terms by considering,
from a physical perspective, all such ways of opening a gap.
However, we will instead determine the topology
of the space of mass terms (and, thereby, the space
of gapped free fermion Hamiltonians) by
the same mathematical methods by which we analyzed the
$2D$ case.

Since ${\gamma_1^2}=1$ and has vanishing
trace, this matrix decomposes $\mathbb{R}^{2N}$
into its $\pm 1$ eigenspaces:
$\mathbb{R}^{2N}={X_+} \oplus {X_-}$.
Now ${({\gamma^{}_2}{\gamma^{}_3})^2}=-1$
and $[{\gamma^{}_1},{\gamma^{}_2}{\gamma^{}_3}]=0$.
Therefore, ${\gamma^{}_2}{\gamma^{}_3}$ is a
complex structure on ${X_+}$ (and also on ${X_-}$),
i.e. we can define multiplication of vectors
$\vec{v}\in{X_+}$ by complex scalars
according to $(a+bi)\vec{v}\equiv
a\vec{v} + {\gamma^{}_2}{\gamma^{}_3}\vec{v}$.
(Consequently, we can view ${X_+}$ as
$\mathbb{C}^{N/2}$.)
Now, consider a possible mass term $M$,
with ${M^2}=-1$. ${({\gamma^{}_2}M)^2}=1$
and $[{\gamma^{}_1},{\gamma^{}_2}M]=0$.
Let $Y$ be the subspace of ${X_+}$
with eigenvalue $+1$ under ${\gamma^{}_2}M$.
Since $\{{\gamma^{}_2}M,{\gamma^{}_2}{\gamma^{}_3}\}=0$,
${\gamma^{}_2}{\gamma^{}_3}Y$ is the subspace of ${X_+}$
with eigenvalue $-1$ under ${\gamma^{}_2}M$.
In other words, ${X_+}=Y\oplus{\gamma^{}_2}{\gamma^{}_3}Y$,
i.e. $Y$ is a real subspace of ${X_+}$.
Hence, the space of choices of $M$ is the space
of real subspaces $Y\subset{X_+}$ (or,
equivalently, of real subspaces
$\mathbb{R}^{N/2}\subset\mathbb{C}^{N/2}$).
Given any fixed real subspace $Y\subset{X_+}$,
we can obtain all others by performing $\text{U}(N/2)$
rotations of ${X_+}$, but two such rotations
give the same real subspace if they differ only
by an $\text{O}(N/2)$ rotation of $Y$. Thus,
the space of gapped  Hamiltonians
for $2N$ free Majorana fermion fields
in $3D$ with no symmetry is topologically-equivalent
to $\text{U}(N/2)/\text{O}(N/2)$. In the remaining
sections of this paper, we will be discussing
topological defects in such systems and their motions.

\subsection{General Classification and Bott Periodicity}
\label{sec:bott}

Before doing so, we pause for a minute
to consider the classification in other dimensions
and in the presence of symmetries such as
time-reversal and charge conservation.
We have seen that systems with no symmetry
in $d=0,2,3$ are classified by the spaces
$\text{O}(2N)/\text{U}(N)$,
$\mathbb{Z}\times
\frac{\text{O}(N)}{\text{O}(N/2)\times\text{O}(N/2)}$,
and $\text{U}(N/2)/\text{O}(N/2)$. By similar methods,
it can be shown that the $d=1$ case is
classified by $\text{O}(N)$. As we have seen,
increasing the spatial dimension increases the
number of $\gamma$ matrices by one.
The problem of choosing $\gamma_{1},\ldots,\gamma_{d}$
satisfying $\{{\gamma_i},{\gamma_j} \}=2\delta_{ij}$
and $M$ which anti-commutes with the $\gamma_i$s
and squares to $-1$ leads us to subspaces of
$\mathbb{R}^{2N}$ of smaller and smaller
dimension, isometries between these spaces,
or complex of quaternionic structures on these spaces.
This leads the progression of spaces
in the top row of Table \ref{tbl:classifying}.

At the same time, we have seen
that a time-reversal-invariant system in $d=0$
is classified by $\text{U}(N)/\text{Sp}(N/2)$.
Suppose that we add a discrete anti-unitary symmetry
$S i S^{-1} =-i$ defined by
\begin{equation}
\label{eqn:symmetries}
S {a_i} S^{-1} = ({J})_{ij}{a_j}
\end{equation}
which squares to ${J^2}=-1$.
It must anti-commute with the mass term
\begin{equation}
\label{eqn:anti-comm}
\{J,M\}=0
\end{equation}
in order to ensure invariance under the symmetry,
so choosing a $J$
amounts to adding a complex structure,
which leads to the {\it opposite} progression of
classifying spaces.
Consider, as an example of the preceding statements,
a time-reversal invariant system in $d=3$.
Then time-reversal symmetry $T$ is an
example of a symmetry generator $J$
discussed in the previous paragraph.
We define a real subspace $Y\subset{X_+}$,
in a similar manner as above, but now
as the subspace of ${X_+}$
with eigenvalue $+1$ under ${\gamma^{}_2}T$,
rather than under ${\gamma^{}_2}M$. Once again,
${X_+}=Y\oplus{\gamma^{}_2}{\gamma^{}_3}Y$.
Now, $\{{\gamma^{}_3}M,{\gamma^{}_2}T\}=0$,
and ${({\gamma^{}_3}M)^2}=1$, so the
$+1$ eigenspace of ${\gamma^{}_3}M$
is a linear subspace of $Y$.
The set of all such linear subspaces is
$\mathbb{Z}\times
\frac{\text{O}(N/2)}{\text{O}(N/4)\times\text{O}(N/4)}$.
But this is the same classifying space as for a system
with no symmetry in $d=2$ (apart from a reduction
of $N$ by a factor of $2$). Thus, we are led
to the list of classifying spaces for gapped
free fermion Hamiltonians
in Table \ref{tbl:classifying}.

\begin{table*}
\begin{tabular}{c | c c c c c c c}
dim.: &  0    &   1 & 2 & 3 & 4 & \ldots\\
\hline\hline
SU($2$), $T$, $Q$& $\mathbb{Z}\times
 \frac{\text{O}(N)}{\text{O}(N/2)\times\text{O}(N/2)}$
 & $\text{U}(N/2)/\text{O}(N/2)$ & $\text{Sp}(N/4)/\text{U}(N/4)$ &
 $\text{Sp}(N/8)$ & $\mathbb{Z}\times\frac{\text{Sp}(N/8)}
 {\text{Sp}(N/16)\times\text{Sp}(N/16)}$ \ldots\\
SU($2$), $T$, $Q$, $\chi$&  $\text{O}(N/4)$& $\mathbb{Z}\times
 \frac{\text{O}(N/4)}{\text{O}(N/8)\times\text{O}(N/8)}$
 & $\text{U}(N/8)/\text{O}(N/8)$ & $\text{Sp}(N/16)/\text{U}(N/16)$ &
 $\text{Sp}(N/32)$ &\ldots\\
no symm.& $\text{O}(2N)/\text{U}(N)$ &  $\text{O}(N)$& $\mathbb{Z}\times
 \frac{\text{O}(N)}{\text{O}(N/2)\times\text{O}(N/2)}$
 & $\text{U}(N/2)/\text{O}(N/2)$ & $\text{Sp}(N/4)/\text{U}(N/4)$ &
 \ldots\\
$T$ only & $\text{U}(N)/\text{Sp}(N/2)$  & $\text{O}(N)/\text{U}(N/2)$
&$\text{O}(N/2)$ & $\mathbb{Z}\times
 \frac{\text{O}(N/2)}{\text{O}(N/4)\times\text{O}(N/4)}$ & \ldots\\
$T$ and $Q$ &
$\mathbb{Z}\times\frac{\text{Sp}(N/2)}{\text{Sp}(N/4)\times\text{Sp}(N/4)}$
& $\text{U}(N/2)/\text{Sp}(N/4)$ & $\text{O}(N/2)/\text{U}(N/4)$
& $\text{O}(N/4)$\\
$T$, $Q$, $\chi$ & $\text{Sp}(N/4)$ &
$\mathbb{Z}\times\frac{\text{Sp}(N/4)}{\text{Sp}(N/8)\times
\text{Sp}(N/8)}$
& $\text{U}(N/4)/\text{Sp}(N/8)$ & $\text{O}(N/4)/\text{U}(N/8)$\\
&\vdots & & & &$\ddots$
\end{tabular}
\caption{The period-$8$ (in both dimension and number of
symmetries) table of classifying spaces for free fermion Hamiltonians
for $N$ complex $=2N$ real (Majorana) fermion fields in dimensions
$d=0,1,2,3,\ldots$ with no symmetries; time-reversal symmetry ($T$) only; time-reversal and charge conservation
symmetries ($T$ and $Q$); time-reversal, charge conservation,
and sublattice symmetries ($T$, $Q$, and $\chi$);
and the latter two cases with SU($2$) symmetry.
As a result of the period-$8$ nature of the table,
the top two rows could equally well be the
bottom two rows of the table.
Moving $p$ steps to the right and $p$ steps down leads to
the same classifying space (but for $1/2^p$ as many fermion
fields), which is a reflection of Bott periodicity, as explained
in the text. The number of disconnected components of any such
classifying space -- i.e. the number of different phases in that
symmetry class and dimension -- is given by the corresponding
$\pi_0$, which may be found in Eq. \ref{eqn:stable-pi-0}. Higher homotopy groups, which classify defects, can be computed using Eq. \ref{eqn:Bott periodicity}. Table \ref{tbl:unitary-classifying}, given in Appendix
\ref{sec:QnotT}, is the analogous table
for charge-conserving Hamiltonians without time-reversal symmetry.}
\label{tbl:classifying}
\end{table*}

In this table, $Q$ refers to charge-conservation symmetry.
Charge conservation is due to the invariance of
the Hamiltonian of a system under the U(1) symmetry
${c_i}\rightarrow e^{i\theta}{c_i}$. In terms of
Majorana fermions $a_i$ defined according to
${c_j}=(a_{2j-1}+ia_{2j})/2$, the symmetry takes the
form $a_{2j-1}~\rightarrow~\cos\theta a_{2j-1}
+ \sin\theta a_{2j}$, $a_{2j}~\rightarrow~-\sin\theta a_{2j-1}
+ \cos\theta a_{2j}$. However, if a free fermion Hamiltonian
is invariant under the discrete
symmetry ${c_i}\rightarrow i{c_i}$ or, equivalently,
$a_{2j-1} \rightarrow a_{2j}$, $a_{2j} \rightarrow - a_{2j-1}$,
then it is automatically invariant under the full U(1) as well,
and conserves charge \cite{Kitaev09}. Thus, we can treat charge
conservation as a discrete symmetry $Q$
which is unitary, squares to $-1$, and commutes
with the Hamiltonian (i.e. with the $\gamma$ matrices and $M$).
Since $Q$ transforms ${c_i}\rightarrow i{c_i}$, it
anti-commutes with $T$.
Note further that if a system has time-reversal
symmetry, then the product of time-reversal $T$ and
charge conservation $Q$ is a discrete anti-unitary symmetry,
$QT$ which anti-commutes with the Hamiltonian and with $T$
and squares to $-1$. Then $QT$ is defined by a choice of matrix $J$,
analogous to $T$, as in Eq. \ref{eqn:symmetries}.
If the system is not time-reversal-invariant,
then charge conservation is a unitary symmetry.
It is easier then to work with complex fermions,
and the classification of such systems falls into an entirely
different sequence, as discussed in Appendix \ref{sec:QnotT}.)

If a system is both time-reversal symmetric
and charge-conserving, i.e. if it is a time-reversal
invariant insulator, then it may have an additional
symmetry which guarantees that the eigenvalues
of the Hamiltonian come in $\pm E$ pairs, just as in a superconductor.
An example of such a symmetry is the sublattice
symmetry of Hamiltonians on a bipartite lattice in
which fermions can hop directly from the $A$ sublattice
to the $B$ sublattice but cannot hop directly between sites on the
same sublattice. In such a case, the system is
invariant under a unitary symmetry
$\chi$ defined as follows. If we block diagonalize
$\chi$ so that one block acts on sites in the $A$ sublattice
and the other on sites in the $B$ sublattice, then
we can write $\chi=\text{diag}(k,-k)$, i.e.
${a_i}({\bf x})\rightarrow -k_{ij}{a_j}({\bf x})$ for
${\bf x}\in A$ and
${a_i}({\bf x})\rightarrow k_{ij}{a_j}({\bf x})$ for ${\bf x}\in B$.
This symmetry transforms the Hamiltonian to minus itself
if ${k^2}=1$ or, in other words, if ${\chi^2}=1$. Then
$\chi Q$ is a unitary symmetry which squares to $-1$
and anti-commutes with the Hamiltonian, $T$,
and $QT$. Hence $\chi Q$, too, is defined by a choice of matrix $J$,
as in Eq. \ref{eqn:symmetries}. We will call such a
symmetry a sublattice symmetry $\chi$ and a system satisfying
this symmetry a `bipartite' system, but the symmetry
may have a different microscopic origin.

In an electron system, time-reversal ordinarily squares
to $-1$, because the transformation law is 
${c_\uparrow}\rightarrow {c_\downarrow}$,
${c_\downarrow}\rightarrow -{c_\uparrow}$,
as we have thus far assumed in taking ${J^2}=-1$.
However, it is possible to have a system of
fully spin-polarized electrons which has an
anti-unitary symmetry $T$ which squares to $+1$.
(One might object to calling this symmetry time-reversal
because it doesn't reverse the electron spins,
but $T$ is a natural label because it is a symmetry
which is just as good for the present purposes.)
Then, since ${J^2}=1$, a choice of $J$ is similar to
a choice of a $\gamma$ matrix. In general,
symmetries (\ref{eqn:symmetries})
which square to $+1$ have the same effect
on the topology of the space of free fermion Hamiltonians
as adding dimensions since each such $J$
defines a subspace of half the dimension within the
eigenspaces of the $\gamma$ matrices.
This is true for systems with ${T^2}=1$.

SU(2) spin-rotation-invariant and time-reversal-invariant
insulators (systems with $T$ and $Q$) effectively
fall in this category. The Hamiltonian for such
a system can be written in the form $H=h\otimes {I_2}$
where the second factor is the $2\times 2$ identity
matrix acting on the spin index. Then time-reversal
can be written in the form $T= t \otimes i{\sigma_y}$,
where ${t^2}=1$, and $Q$ can be written in
the form $Q= q \otimes {I_2}$,
so that $QT =qt \otimes i{\sigma_y}$, where ${(qt)^2}=1$.
Thus, since the matrix $i\sigma_y$ squares to $-1$,
the symmetries $T$ and $QT$ have effectively become
symmetries which square to $+1$. They now move the
system through the progression of classifying spaces
in the same direction as increasing the dimension,
i.e. in the opposite direction to symmetries which
square to $-1$. Thus, SU(2) spin-rotation-invariant and
time-reversal-invariant insulators in $d$ dimensions
are classified by the same space as
systems with no symmetry in $d+2$ dimensions.
However, in a system which, in addition, has
sublattice symmetry $\chi=x\otimes {I_2}$,
we have $(qx)^{-1}$. Thus, sublattice symmetry is
still a symmetry which squares to $-1$.
Since the two symmetries which square to $+1$
($T$ and $QT$) have the same effect as increasing
the dimension while the symmetry which squares
to $-1$ has the same effect as decreasing the dimension,
SU(2) spin-rotation-invariant and time-reversal-invariant insulators
with sublattice symmetry in $d$ dimensions
are classified by the same space as
systems with no symmetry in $d+1$ dimensions
(but with $N$ replaced by $N/4$).
Similar considerations apply to superconductors with
SU(2) spin-rotational symmetry.

\begin{table*}
 \begin{tabular}[t]{|c|c|c|c|c|}
 \hline
Symmetry classes&Physical realizations&$d=1$&$d=2$&$d=3$
 \\\hline
\hline D&SC&{\color{blue}$p$-wave SC}&{\color{blue}$(p+ip)$-SC}&0
\\\hline
DIII&TRI SC&{\color{red} ${\rm Z_2}$}&{\color{blue} $(p+ip)(p-ip)$-SC}&He$^3$-B
\\\hline AII&TRI
ins.&0&HgTe Quantum well&${\rm Bi_{1-x}Sb_x}$, ${\rm Bi_2Se_3}$, etc.
\\\hline
CII&Bipartite TRI ins.&Carbon nanotube&0&{\color{red} ${\rm Z_2}$}
\\\hline
C&Singlet SC&0&{\color{blue} $(d+id)$-SC}&0
\\\hline
CI&Singlet TRI
SC&0&0&{\color{red} Z}
\\\hline
AI& TRI ins. w/o SOC&0&0&0
\\\hline
BDI&Bipartite TRI
ins. w/o SOC&Carbon nanotube&0&0
\\\hline
 \end{tabular}
  \caption{Topological periodic table in physical dimensions $1,2,3$. The first column contains 8 of the 10 symmetry classes in the Cartan notation
 adopted by Schnyder {\it et al.}\cite{Ryu08}, following Zirnbauer
 \cite{Zirnbauer96,Altland97}.
 The second column contains the requirements for physical systems which can
 realize the corresponding symmetry classes. ``w/o" stands for ``without". SC stands for superconductivity, TRI for time-reversal invariant, and SOC for spin-orbit coupling. The three columns $d=1,2,3$ list topological states in the spatial dimensions $1,2,3$ respectively. $0$ means the topological classification is trivial. The red labels {\color{red} Z} and {\color{red} ${\rm Z_2}$} stand for topological states
classified by these groups but for which states corresponding
to non-trivial elements of $\mathbb{Z}$ or $\mathbb{Z}_2$
have not been realized in realistic materials.
The blue text stands for topological states for which
a well-defined physical model has been proposed but convincing experimental candidate has not been found yet. (See text for more discussions on the realistic materials.) }
  \label{tbl:periodic}
\end{table*}

In order to discuss topological defects in
the systems discussed here,
it is useful to return to the arguments which led to
(\ref{eqn:Bott-step-1}).
By showing that the space of loops in
$\text{O}(N/2)/(\text{O}(N/4)\times\text{O}(N/4))$ is
well-approximated by $\text{O}(N/4)$,
we not only showed that
${\pi_1}(\text{O}(N/2)/(\text{O}(N/4)\times\text{O}(N/4)))=
{\pi_0}(\text{O}(N/4))$ but, in fact, that
${\pi_k}(\text{O}(N/2)/(\text{O}(N/4)\times\text{O}(N/4)))=
\pi_{k-1}(\text{O}(N/4))$ (see Ref. \onlinecite{Milnor63}).
Continuing in the same way, we can approximate
the loop space of $\text{O}(N/4)$ (i.e. the space of loops in
$\text{O}(N/4)$) by minimal
geodesics from $\mathbb{I}$ to $-\mathbb{I}$:
$L'(\lambda)=e^{\lambda A_1}$ where ${A_1^2}=-1$.
The mid-point of such a geodesic, $L'(\pi/2)=A_1$
again defines a complex structure ${A_1}={O^T}JO$,
where $J$ is given by (\ref{eqn:canonical-J}) so that
$\pi_{k}(\text{O}(N/4))=\pi_{k-1}(\text{O}(N/4)/\text{U}(N/8))$.
In a similar way, minimal geodesics in
$\text{O}(N/4)/\text{U}(N/8)$ from $A_1$ to $-A_1$
can be parametrized by their mid-points $A_2$,
which square to $-1$ and anti-commute
with $A_1$, thereby defining a quaternionic
structure, so that the loop space of
$\text{O}(N/4)/\text{U}(N/8)$ is equivalent to
$\text{U}(N/8)/\text{Sp}(N/8)$ and, hence
$\pi_{k}(\text{O}(N/4)/\text{U}(N/8))=
\pi_{k-1}(\text{U}(N/8)/\text{Sp}(N/8))$.
Thus, we see that {\it the passage from
one of the classifying spaces to its loop space
is the same as the imposition of a symmetry
such as time-reversal to a system classified by that space}:
both involve the choice of successive anticommuting
complex structures.
Continuing in this fashion (see Ref. \onlinecite{Milnor63}),
we recover {\it Bott periodicity}:
\begin{multline}
\label{eqn:Bott periodicity}
{\hskip -0.5 cm} {\pi_k}(\text{O}(16N))=\\
\pi_{k-1}(\text{O}(16N)/\text{U}(8N))
= \pi_{k-2}(\text{U}(8N)/\text{Sp}(4N)) \\
= \pi_{k-3}(\mathbb{Z}\times
\text{Sp}(4N)/(\text{Sp}(2N)\times\text{Sp}(2N)))
= \pi_{k-4}(\text{Sp}(2N))\\
= \pi_{k-5}(\text{Sp}(2N)/\text{U}(2N))
= \pi_{k-6}(\text{U}(2N)/\text{O}(2N))\\
 = \pi_{k-7}(\mathbb{Z}\times
 \text{O}(2N)/(\text{O}(N)\times\text{O}(N)))\\
 = \pi_{k-8}(\text{O}(N))
\end{multline}
The approximations made at each step require
that $N$ be in the {\it stable
limit}, in which the desired homotopy groups are
independent of $N$. For instance,
${\pi_k}(\text{O}(N))$ is
independent of $N$ for $N>k/2$.

It is straightforward to compute $\pi_0$ for each of these
groups:
\begin{eqnarray}
\label{eqn:stable-pi-0}
{\pi_0}(\text{O}(N))&=&\mathbb{Z}_2\cr
\pi_{0}(\text{O}(2N)/\text{U}(N))&=&\mathbb{Z}_2\cr
\pi_{0}(\text{U}(2N)/\text{Sp}(N))&=&0\cr
\pi_{0}(\mathbb{Z}\times
\text{Sp}(2N)/\text{Sp}(N)\times\text{Sp}(N))&=&\mathbb{Z}\cr
\pi_{0}(\text{Sp}(N))&=&0\cr
\pi_{0}(\text{Sp}(N)/\text{U}(N))&=&0\cr
\pi_{0}(\text{U}(N)/\text{O}(N))&=&0\cr
\pi_{0}(\mathbb{Z}\times
 \text{O}(2N)/(\text{O}(N)\times\text{O}(N)))
 &=&\mathbb{Z}
\end{eqnarray}
Combining (\ref{eqn:stable-pi-0}) with (\ref{eqn:Bott periodicity}),
we can compute any of the stable homotopy groups
of the above $8$ classifying spaces. As discussed above,
the space of gapped free fermion Hamiltonians in
$d$-dimensions in a given {\it symmetry class}
(determined by the number modulo $8$ of symmetries squaring to $-1$
minus the number of those squaring to $+1$) is
homotopically-equivalent to one of these classifying spaces.
Thus, using (\ref{eqn:stable-pi-0}) with (\ref{eqn:Bott periodicity})
to compute the stable homotopy groups of these classifying spaces
leads to a complete classification of topological states and
topological defects in all dimensions and
symmetry classes, as we now discuss.

Gapped Hamiltonians with a given symmetry and dimension are classified by
$\pi_0$ of the corresponding classifying space in Table \ref{tbl:classifying}. Due to Bott periodicity, the table is periodic along both directions of dimension and symmetry, so that there are 8 distinct symmetry classes.
Ryu {\it et. al}\cite{Ryu08} denoted these classes using the
Cartan classification of symmetric spaces, following
the corresponding classification of disordered systems and
random matrix theory \cite{Zirnbauer96,Altland97} which
was applied to the (potentially-gapless) surface states of
these systems. In this notation, systems with no symmetry
are in class D, those with $T$ only are in DIII, and those
with $T$ and $Q$ are in AII. The other 5 symmetry classes,
C, CI, CII, AI, and BDI arise, arise in systems which
have spin-rotational symmetry or a sublattice symmetry.
There are actually 2 more symmetry classes (denoted by A and AIII in the random matrix theory) which lie on a separate $2\times 2$ periodic table, which is less relevant to the present work and will be discussed in the Appendix \ref{sec:QnotT}. In Table \ref{tbl:periodic} we have listed examples of
topologically-nontrivial states in physical dimensions 1,2,3
in all 8 symmetry classes. To help with
the physical understanding of these symmetry classes, we have also listed the physical requirements for the realization of each symmetry class.
In each dimension, there are two symmetry classes
in which the topological states are classified by integer invariants and
two symmetry classes in which the different states are
distinguished by $\mathbb{Z}_2$ invariants. In all the
cases in which a real material or a well-defined physical model system is known 
with non-trivial $\mathbb{Z}$ or $\mathbb{Z}_2$ invariant,
we have listed a typical example in the table.
In some of the symmetry classes, non-trivial examples
have not been realized yet, in which case we leave the topological classification
$\mathbb{Z}$ or $\mathbb{Z}_2$ in the corresponding position in the table.

%New paragraph inserted. XLnote0418
In one dimension, generic superconductors (class D) are classified by $Z_2$, of which the nontrivial example is a $p$-wave superconductor with a single Majorana zero mode on the edge. The time-reversal invariant superconductors (class DIII) are also classified by $Z_2$. The nontrivial example is a superconductor in which spin up electrons pair into a $p$-wave superconductor and spin down electron form another $p$-wave superconductor which is exactly the time-reversal of the spin-up one. Such a superconductor has two Majorana zero modes on the edge which form a Kramers pair and are topologically protected. The two integer classes are bipartite time-reversal invariant insulators with (CII) and without (BDI) spin-orbit coupling. An example of the
BDI class is a graphene ribbon, or equivalently a carbon nanotube with a zigzag edge. \cite{Fujita96,Nakada96}. The low-energy band structure of graphene and carbon nanotubes is well-described by a tight-binding model with nearest-neighbor hopping on a honeycomb lattice, which is bipartite. The integer-valued topological quantum number corresponds to the number of zero modes on the edge, which depends on the orientation of the nanotube. Because carbon has negligible spin-orbit coupling, to a good approximation it can be viewed as a system in the BDI class, but it can also be considered as a system in class CII when spin-orbit coupling is taken into account. In two dimensions, generic superconductors (class D) are classified by an integer, corresponding to the number of chiral Majorana edge states on the edge. The first nontrivial example was the $p+ip$ wave superconductor, shown by Read and Green\cite{Read00} to have one chiral Majorana edge state.
Non-trivial superconductors in symmetry class D
are examples of {\it topological superconductors}.
Some topological superconductors can be consistent with spin rotation symmetry; singlet superconductors (class C) are also classified by integer, with the simplest physical example a $d+id$ wave superconductor. Similar to the 1D case, the time-reversal invariant superconductors (class DIII) are classified by $Z_2$, of which the nontrivial example is a superconductor with $p+ip$ pairing of spin-up electrons and $p-ip$ pairing of spin-down electrons.\cite{Roy06,Qi09,Ryu08} The other symmetry class in
2D with a $Z_2$ classification is composed of time-reversal invariant insulators (class AII), also known as quantum spin Hall insulators\cite{Kane05A,Kane05B,Bernevig06a}. The quantum spin Hall insulator phase has been theoretically predicted\cite{Bernevig06b} and experimentally realized\cite{Koenig07} in HgTe quantum wells. In three dimensions, time-reversal invariant insulators (class AII) are also classified by $Z_2$. \cite{Fu07,Moore07,Roy09} The $Z_2$ topological invariant corresponds to a topological magneto-electric response with quantized coefficient $\theta=0,\pi$\cite{Qi08}. Several nontrivial topological insulators in this class have been theoretically predicted and experimentally realized, including ${\rm Bi_{1-x}Sb_x}$ alloy\cite{Fu07b,Hsieh08} and the family of ${\rm Bi_2Se_3}$, ${\rm Bi_2Te_3}$, ${\rm Sb_2Te_3}$\cite{Zhang09,Xia09,Chen09}. In 3D, time-reversal invariant superconductors (class DIII) are classified by an integer, corresponding to the number of massless Majorana cones on the surface.\cite{Ryu08} A nontrivial example with topological quantum number $N=1$ turns out to be the B phase of He$^3$.\cite{Qi09,Roy08b,Ryu08} The other classes with nontrivial topological classification in 3D are singlet time-reversal invariant superconductors (CI), classified by an integer; and bipartite time-reversal invariant insulators (CII), classified by $Z_2$. Some models have been proposed\cite{Schnyder09} but no realistic material proposal or experimental realization has been found in these two classes. We would like to note that different physical systems can correspond to the same symmetry class. For example, bipartite superconductors are also classified by the BDI class.

%New paragraph end. XLnote0418
The two remaining symmetry classes (unitary (A) and chiral unitary (AIII)) corresponds to systems with charge conservation symmetry but without
time-reversal symmetry, which forms a separate $2\times 2$ periodic table. For the sake of completeness, we carry out
the preceding analysis for these two classes in Appendix \ref{sec:QnotT}.

Topological defects in these states are classified by higher homotopy groups of the classifying spaces. Following the convention of Ref. \cite{Kitaev09}, we name the classifying spaces by $R_q,q=0,1,2,...,7$, with $R_1=O(N),~R_2=O(2N)/U(N),~...~R_7=U(N)/O(N),~R_0=\mathbb{Z}\times
 \text{O}(2N)/(\text{O}(N)\times\text{O}(N))$ in the order of Eq. (\ref{eqn:stable-pi-0}). The symmetries in Table \ref{tbl:periodic} can be labeled by $p=0,1,2,...,7$, so that in $d$ dimensions and $p$-th symmetry class, the classifying space is $R_{2+p-d}$. A topological defect with dimension $D$ ($D<d$) is classified by
\begin{eqnarray}
\pi_{d-D-1}\left(R_{2+p-d}\right)=\pi_0\left(R_{p-D+1}\right)\label{eq:defectclassification}
\end{eqnarray}
which is determined by the zero-th homotopy groups in Eq. (\ref{eqn:stable-pi-0}). The important conclusion we obtain from this formula is that the classification of topological defect is determined by the dimension of the defect $D$ and the symmetry class $p$, and is {\em independent} of the spatial dimension $d$. For $p$ and $D$, we obtain the same $8\times 8$ periodic table for the classification of topological defects. \cite{noteaboutKane}

In the following
we will focus on the 3D system with no symmetry, and discuss the generic cases in Sec. \ref{sec:discussion}.
Since a 3D system with no symmetry is classified
by $R_7=\text{U}(N)/\text{O}(N)$, point-like defects in
such a system are classified by
${\pi_2}(\text{U}(N)/\text{O}(N))$.
In the stable limit,
${\pi_2}(\text{U}(N)/\text{O}(N))=\mathbb{Z}_2$.
However, we note that, for smaller values of
$N$ below the stable limit, the classification is a little different,
e.g. ${\pi_2}(\text{U}(2)/\text{O}(2))=
{\pi_2}(\text{U}(1)\times{S^2})=\mathbb{Z}$.
The `8-band model' in Ref. \onlinecite{Teo10},
which we discussed in the strong-coupling
limit in the previous section,
is an example of this particular `small $N$'
case, which is why the hedgehogs
in that model are classified by a
winding number $\in\mathbb{Z}$.

\section{Exchanging Particles Connected
by Ribbons in 3D}
\label{sec:tethered}

In this section, we will take an O(3) non-linear
$\sigma$ model, i.e. one with target space
$S^2$, as a toy model for our problem.
It is essentially the `8-band' model
discussed in Ref. \onlinecite{Teo10}
and in the strong-coupling limit in Section
\ref{sec:strong-coupling}. As noted above,
it is the $2N=8$ limit of the classification
reviewed above. It will be more familiar to
most readers and easier to visualize
than the full problem which we discuss in
the next section.
We will give a heuristic explanation
of $\pi_1$ of the configurations of
defects and will make a few comments about
where our toy model goes wrong, compared
to the full problem. In the next section,
we will undertake a full and careful calculation
of $\pi_1$ of the configurations of
defects of a model with classifying space
$\text{U}(N)/\text{O}(N)$.

A free fermion Hamitonian with no
symmetry, but limited to 8 bands,
can be expanded about its minimum
energy point in the Brillouin zone as
\begin{equation}
H = i\chi ({\partial_i}{\gamma_i} +
{n_i}{\Gamma_i})\chi
\end{equation}
From the considerations in the previous
section, we learned that the space of mass terms
is ${S^2}$, which is why we have written the mass
term in the above form with a unit vector $\vec{n}\in{S^2}$.
(More precisely, the space of mass terms is
$\text{U}(2)/\text{O}(2)=\text{U}(1)\times{S^2}$,
and the mass term can be written
$e^{\theta{\gamma_1}{\gamma_2}{\gamma_3}}{n_i}{\Gamma_i}
e^{-\theta{\gamma_1}{\gamma_2}{\gamma_3}}$.
However, the extra U(1) plays no role here and can
be ignored.)
In the model of Section \ref{sec:strong-coupling},
the three components of $\vec{n}$ correspond to dimerization
in each of the three directions on the cubic
lattice. In Teo and Kane's
model they correspond to the real and imaginary
parts of the superconducting order parameter
and the sign of the Dirac mass at a band inversion.

\begin{figure}[b]
\centering
\includegraphics[width=3in]{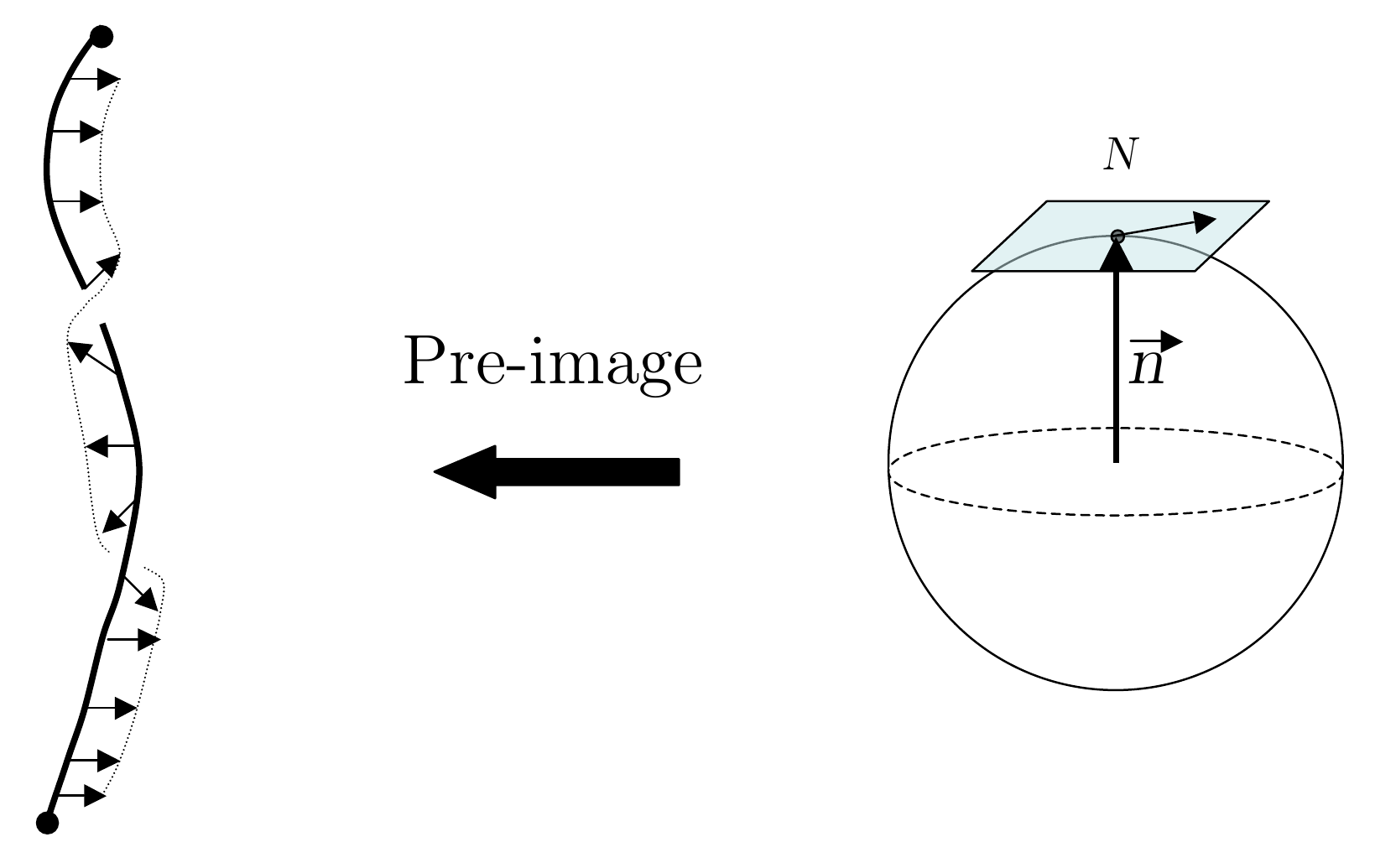}
\caption{The pre-image of the North pole in the $S^2$
target space of $\vec{n}$ is a collection of arcs
connecting hedgehogs. A fixed tangent vector at the North
pole defines a vector field along the arcs, thereby making
them framed arcs, which may be viewed as ribbons.}
\label{fig:north-pole}
\end{figure}

Now consider defects in the $\vec{n}$ field,
which are classified by ${\pi_2}({S^2})=\mathbb{Z}$.
Defects with winding number $\pm 1$ are positive and
negative hedgehogs. For simplicity, we will focus on
these; higher winding number hedgehods can be
built up from these. (In the real model, as opposed to the
toy model, the `hedgehogs' have a $\mathbb{Z}_2$
classification so there are no higher winding number
hedgehogs and, in fact, they do not even have a sign.)
As noted by Teo and Kane \cite{Teo10}, the
$\vec{n}$ field around a hedgehog can be
visualized in a simplified way, following Wilczek
and Zee's discussion of the Hopf term in
a $2+1$-D O(3) non-linear $\sigma$ model\cite{Wilczek83}.
The field $\vec{n}$ can be viewed as a map from
the physical space where the electrons live,
$\mathbb{R}^3$. If we assume
that the total winding number is zero
(equal numbers of + and - hedgehogs)
and that $\vec{n}$ approaches a constant
at $\infty$, we can compactify the physical space $\mathbb{R}^3$
so that it is $S^3$.
The target space of the map is $S^2$. The pre-image
of the north pole $N\in S^2$ is a set of arcs
and loops. The choice of the north pole $N$
is arbitrary, and any other point on the sphere
would be just as good for the following discussion.
Let's ignore the loops for the moment
and focus on the arcs.
Since $\vec{n}$ points in every direction at
a hedgehog, the arcs terminate at hedgehogs.
In fact, each arc connects a $+1$ hedgehog to
a $-1$ hedgehog.
We now pick an arbitrary unit vector in the tangent
space of the sphere at $N$. This vector can be
pulled back to $S^3$ to define a vector field
along the arcs which is clearly normal
to the arcs. This is a {\it framing}, which allows us to define,
for instance, a self-winding number for an arc. Intuitively,
we can think of a framing as a thickening of
an arc into a ribbon.
Thus, the field $\vec{n}$ allows
us to to define a set of framed arcs connecting
the hedgehogs -- in other words, a set
of ribbons connecting the hedgehogs. As the
normal vector twists around an arc, the ribbon twists,
as shown in Fig. \ref{fig:north-pole} (Although we will draw the ribbons
as bands in the physical space, their width should
not be taken seriously; they should really be viewed as
arcs with a normal vector field.)

\begin{figure}
\centering
\includegraphics[width=2in]{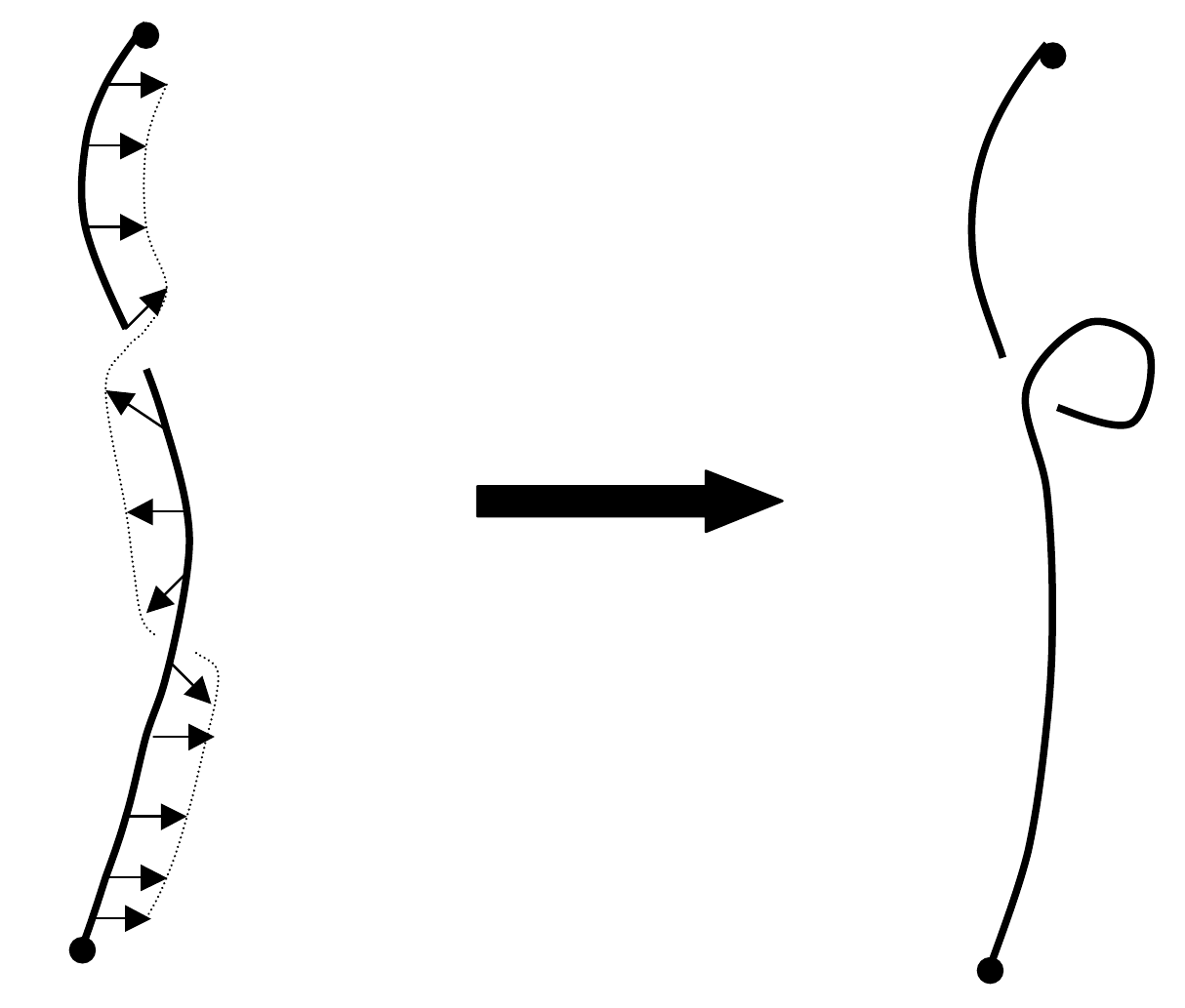}
\caption{We will depict framed arcs or ribbons
as arcs with twists accounted for by drawing
kinks in the arcs, as shown above.}
\label{fig:ribbons->arcs}
\end{figure}

Although these ribbons are strongly reminiscent
of particle trajectories, it is important to keep in
mind that they are not. A collection of ribbons
connecting hedgehogs defines a state of the system
at an instant of time. Ribbons, unlike particle trajectories,
can cross. They can break and reconnect as the system
evolves in time. As hedgehogs are moved, the
ribbons move with them.

A configuration of particles connected pairwise
by ribbons is a seemingly crude approximation
to the full texture defined by $\vec{n}$.
However, according to the Pontryagin-Thom
construction, as we describe in the next Section
(and explain in Appendix \ref{sec:appendix_pontryagin_thom_construction}),
it is just as good as the full texture for topological
purposes. Thus, we focus on the space of
particles connected pairwise by ribbons.

We now consider a collection of such particles
and ribbons. For a topological discussion,
all that we are interested in about
the ribbons is how many times they twist, so we will
not draw the framing vector but will, instead, be careful to
put kinks into the arcs in order to keep track
of twists in the ribbon, as depicted in Fig. \ref{fig:ribbons->arcs}.
The fundamental group of their configuration space is
the set of transformation which return the particles
and ribbons to their initial configurations,
with two such transformations identified if they
can be continuously deformed into each other.
Consider an exchange of two $+1$ hedgehogs, as depicted
in Fig. \ref{fig:exchange1}.
Although this brings the particles back to their
initial positions (up to a permutation, which
is equivalent to their initial configuration
since the particles are identical), it does not
bring the ribbons back to their initial configuration.
Therefore, we need to do a further motion
of the ribbons. By cutting and rejoining them
as shown in Fig. \ref{fig:exchange2}a, a procedure
which we call `recoupling', we now have the ribbons
connecting the same particles as in the initial configuration.
But the ribbon on the left has a twist in it. So we rotate
that particle by $-2\pi$ in order to undo the twist,
as in Fig. \ref{fig:exchange2}b.

\begin{figure}
\centering
\includegraphics[width=1.5in]{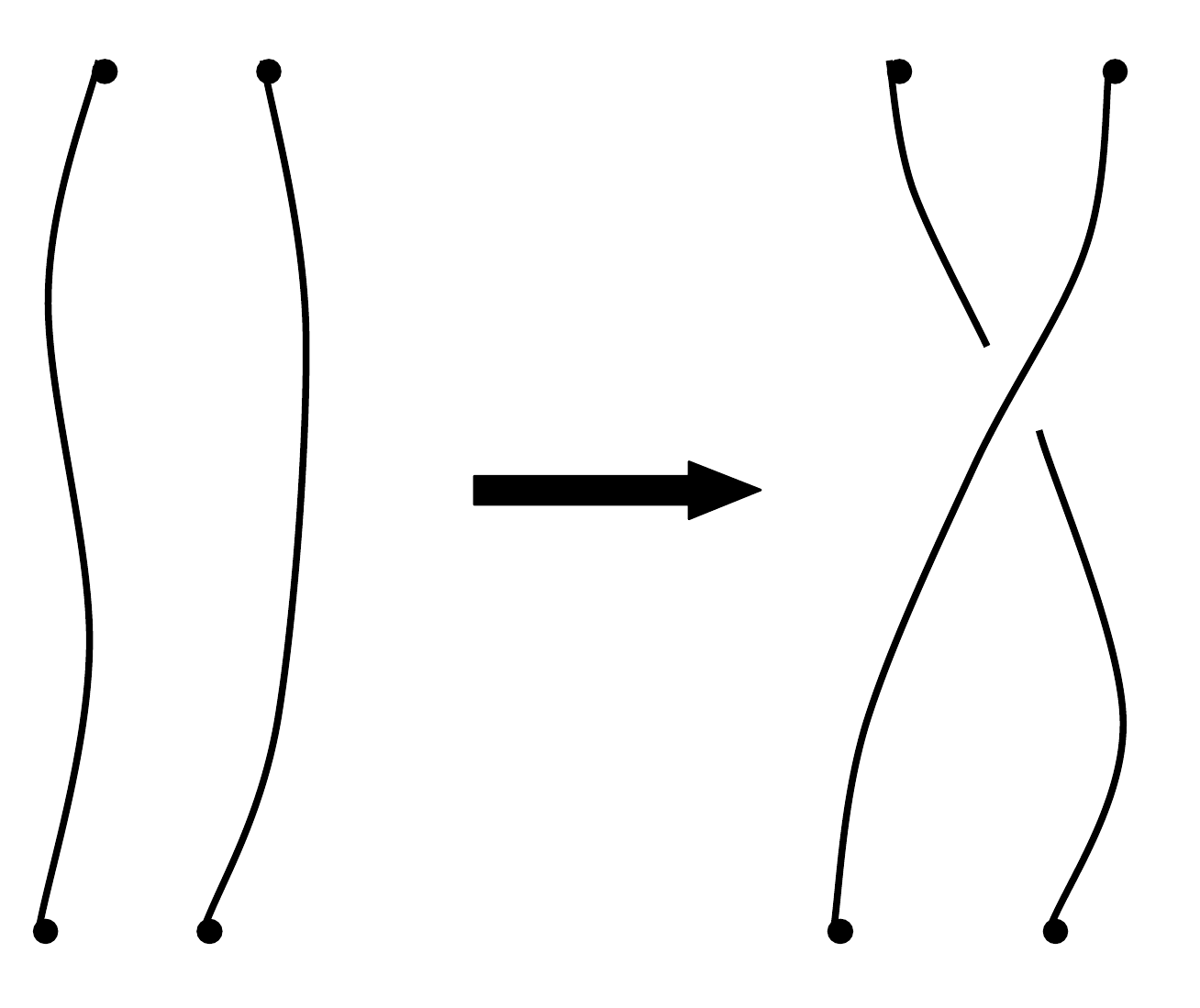}
\caption{When two defects are exchanged, the
$\vec{n}$-field around them is modified. This is encapsulated
by the dragging of the framed arcs as the defects are
moved.}
\label{fig:exchange1}
\end{figure}

\begin{figure}
\centering
\includegraphics[width=2.5in]{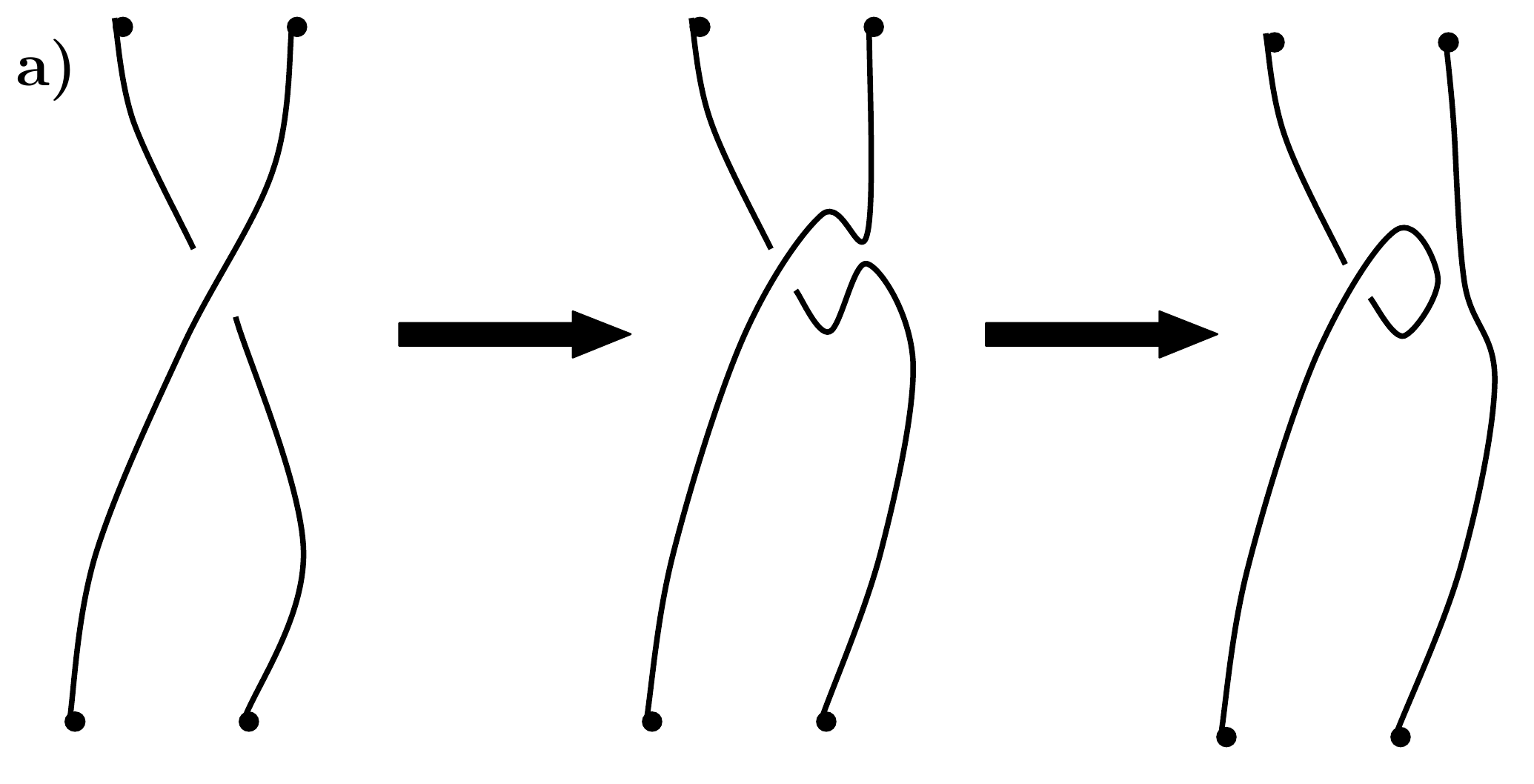}
\includegraphics[width=1.75in]{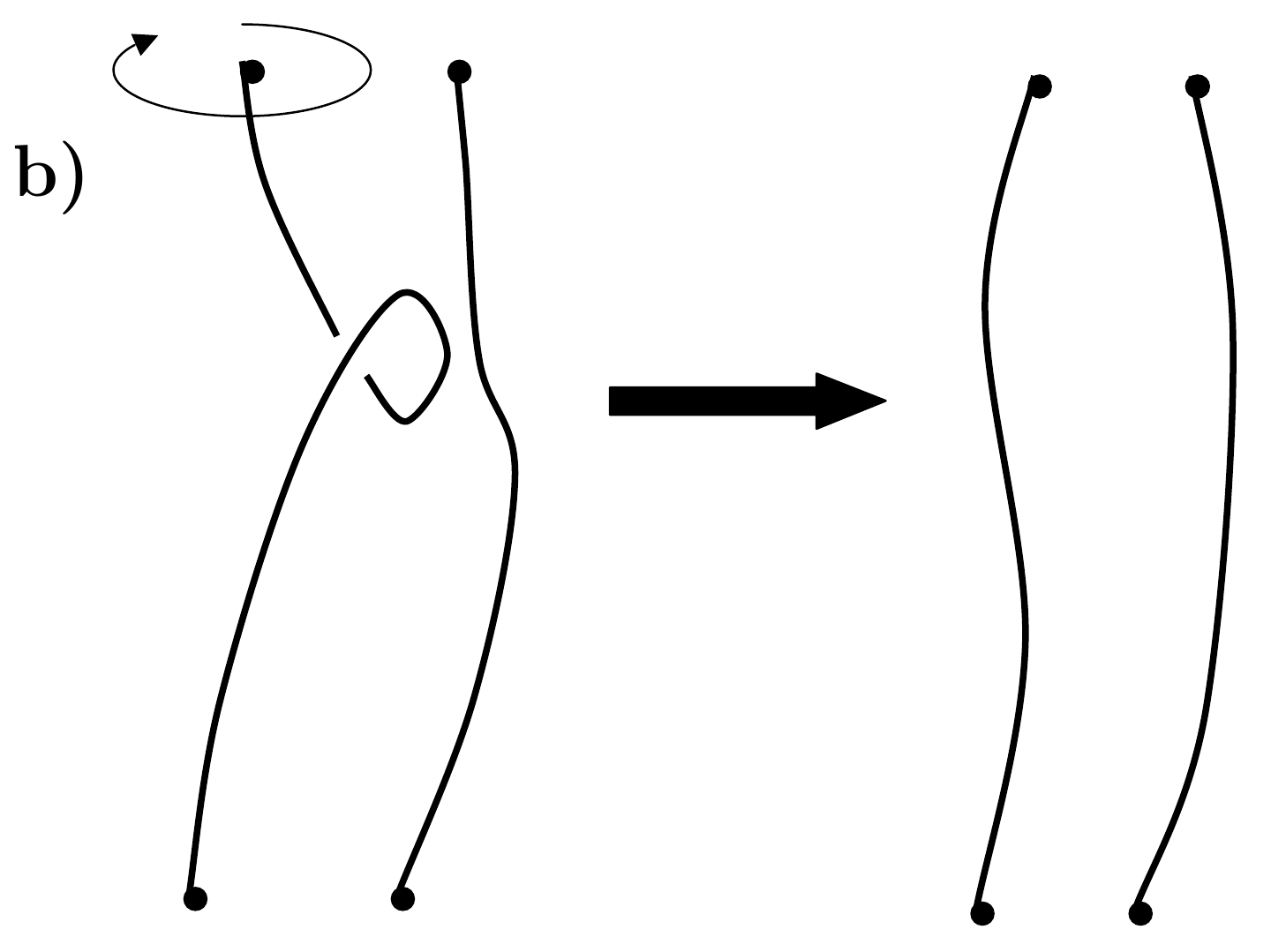}
\caption{(a) In order to restore the framed arcs so that
they are connecting the same defects, it is necessary
to perform a recoupling by which they are reconnected.
In order to keep track of the induced twist, it is easiest
to perform the recoupling away from the overcrossing.
(b) The particle on the left must be rotated by
$-2\pi$ in order to undo a twist in the framed arc
to which it is attached.}
\label{fig:exchange2}
\end{figure}

Let us use $t_i$ to denote such a transformation, defined by
the sequence in Figs. \ref{fig:exchange1},
\ref{fig:exchange2}a, and \ref{fig:exchange2}b.
The $t_i$s do not satisfy the multiplication rules of the
permutation group. In particular, ${t_i}\neq t_i^{-1}$.
The two transformations ${t_i}$ and $t_i^{-1}$
are not distinguished by whether the exchange
is clockwise or counter-clockwise -- this is immaterial
since a clockwise exchange can be
deformed into counter-clockwise one -- but rather
by which ribbon is left with a twist which must be
undone by rotating one of the particles.

\begin{figure}
\centering
\includegraphics[width=3.25in]{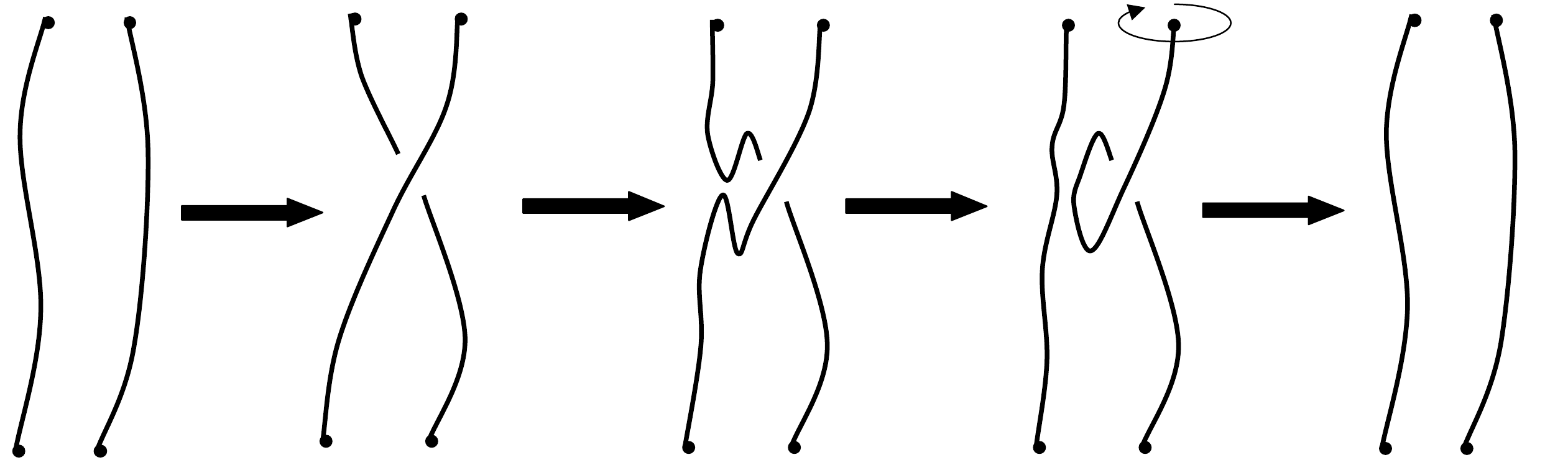}
\caption{The sequence of moves which defines
$t_i^{-1}$. (Here, the $i^{\rm th}$ particle is
at the top left and the $(i+1)^{\rm th}$ is at the
top right.) This may be contrasted with
the sequence in Figs. \ref{fig:exchange1},
\ref{fig:exchange2}a, and \ref{fig:exchange2}b,
which defines ${\sigma_i^1}$.}
\label{fig:exchange4}
\end{figure}

To see that the operations $t_i$, defined by
the sequence in Figs. \ref{fig:exchange1},
\ref{fig:exchange2}a, and \ref{fig:exchange2}b,
and $t_i^{-1}$, defined by the sequence in \ref{fig:exchange4},
are, in fact, inverses, it is useful to note that when they are performed
sequentially, they involve two $2\pi$ twists of the same hedgehog.
In \ref{fig:exchange2}b, it is the hedgehog on the left
which is twisted; this hedgehog moves to the right in the first step of
\ref{fig:exchange4} and is twisted again in the fourth step.
One should then note that a double twist
in a ribbon can be undone continuously by using
the ribbon to ``lasso'' the defect, a famous fact
related to the existence of spin-$1/2$ and the
fact that ${\pi_1}(SO(3))=\mathbb{Z}_2$.
This is depicted in Fig. \ref{fig:lassomove} in Appendix
\ref{sec:Postnikov}. It will be helpful for our late
discussion to keep in mind that $t_i$ not only
permutes a pair of particles but also rotates one
of them; any transformation built up by multiplying
$t_i$s will enact as many $2\pi$ twists as pairwise
permutations modulo two.

Thus far we have only discussed the $+1$ hedgehogs.
We can perform the similar transformations which
exchange $-1$ hedgehogs. We will not repeat the above
discussion for $-1$ hedgehogs since the discussion
would be so similar; furthermore, in the $N\rightarrow\infty$
model which is our main interest, defects do not carry
a sign, so they can all be permuted with each other.

We have concluded that ${t_i}\neq t_i^{-1}$
and, therefore, the group of transformations which bring
the hedgehogs and ribbons back to their initial configuration
is not the permutation group. This leaves open the question:
what is ${t_i}^2$? The answer is that $t_i^2$ can be continuously
deformed into a transformation which doesn't involve
moving any of the particles -- Teo and Kane's
`braidless operations'. Consider the transformation
$x_i$ depicted in Fig. \ref{fig:twist-transfer}. Defect $i$
is rotated by $2\pi$, the twist is transferred
from one ribbon to the other,
and defect $i+1$ is rotated by $-2\pi$.
Since a $4\pi$ rotation can be unwound, as depicted
in Fig. \ref{fig:lassomove}, ${x_i^2=1}$.

Intuitively, one expects that ${x_i}={t_i}^2$
since neither ${x_i}$ nor ${t_i}^2$
permutes the particles and both of them involve
$2\pi$ rotations of both particles $i$ and $i+1$.
To show that this is, in fact, the correct, we need
to show that the history in Fig. \ref{fig:twist-transfer}
can be deformed into the sequence of Figs.
\ref{fig:exchange1}, \ref{fig:exchange2}a,
\ref{fig:exchange2}b repeated twice.
If the history in Fig. \ref{fig:twist-transfer} is
viewed as a `movie' and the sequence of Figs.
\ref{fig:exchange1}, \ref{fig:exchange2}a,
\ref{fig:exchange2}b repeated twice is
viewed as another `movie', then we need a
one-parameter family of movies -- or a `movie
of movies' -- which connects the two movies.
We will give an example of such a `movie
of movies' shortly. With this example in hand, the
reader can verify that  ${x_i}={t_i}^2$ by drawing
the corresponding pictures, but we will not do so here
since this discussion
is superseded, in any case, by the
the next section, where a similar result is shown for
the $N\rightarrow \infty$ problem by more general
methods. We simply accept this identity for now.

\begin{figure}
\centering
\includegraphics[width=3.25in]{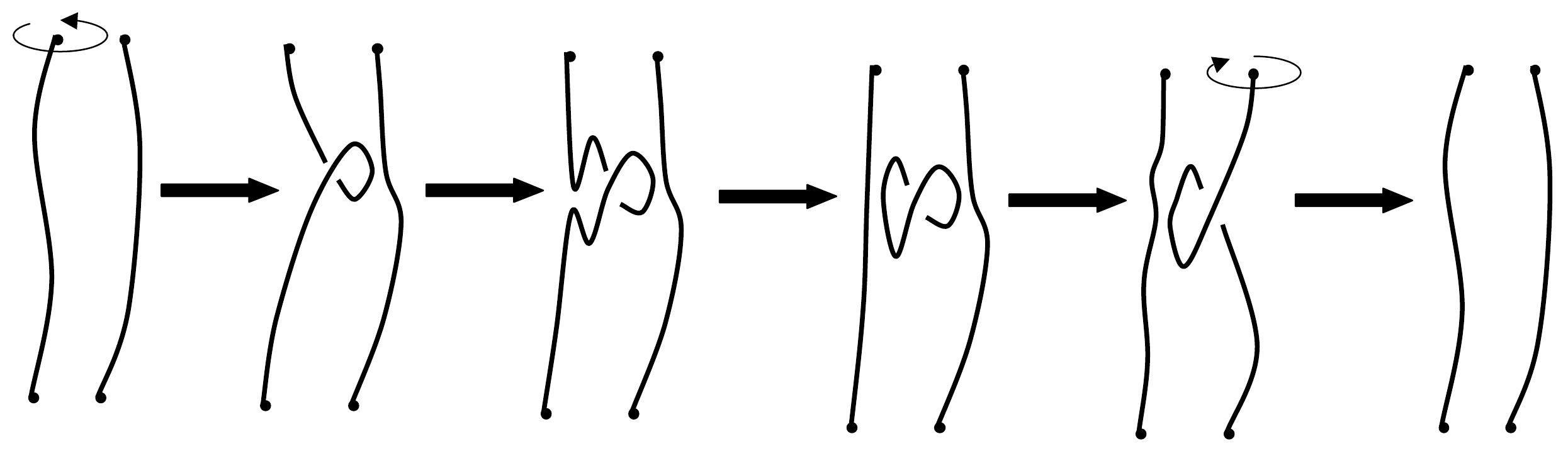}
\caption{The sequence of moves which defines
$x_i$: the defect on the left is rotated by $2\pi$,
the twist is transfered to the ribbon on the right by two
recouplings, and then the defect on the right is
rotated by $-2\pi$. (Here, the $i^{\rm th}$ defect is
at the top left and the $(i+1)^{\rm th}$ is at the
top right.) The defects themselves are not moved
in such a process.}
\label{fig:twist-transfer}
\end{figure}

We now consider the commutation
relation for the $x_i$s. Clearly, for $|i-j|\geq 2$,
${x_i} {x_j} = {x_j} {x_i}$. It is also
intuitive to conclude that
\begin{equation}
{x_i} x_{i+1} = x_{i+1} {x_i}
\end{equation}
since the order in which twists are transferred is
seemingly unimportant.
However, since this is a crucial point, we verify it
by showing in Figure \ref{fig:movie-of-movies}
that the sequence of moves which
defines ${x_i} x_{i+1}$ (a `movie') can be continuously
deformed into the sequence of moves which
defines $x_{i+1}{x_i}$ (another `movie'). Such a deformation
is a `movie of movies'; going from left-to-right in
Fig. \ref{fig:movie-of-movies} corresponds to going
forward in time while going from up to down corresponds
to deforming from one movie to another.

\begin{figure*}
\centering
\includegraphics[width=6in]{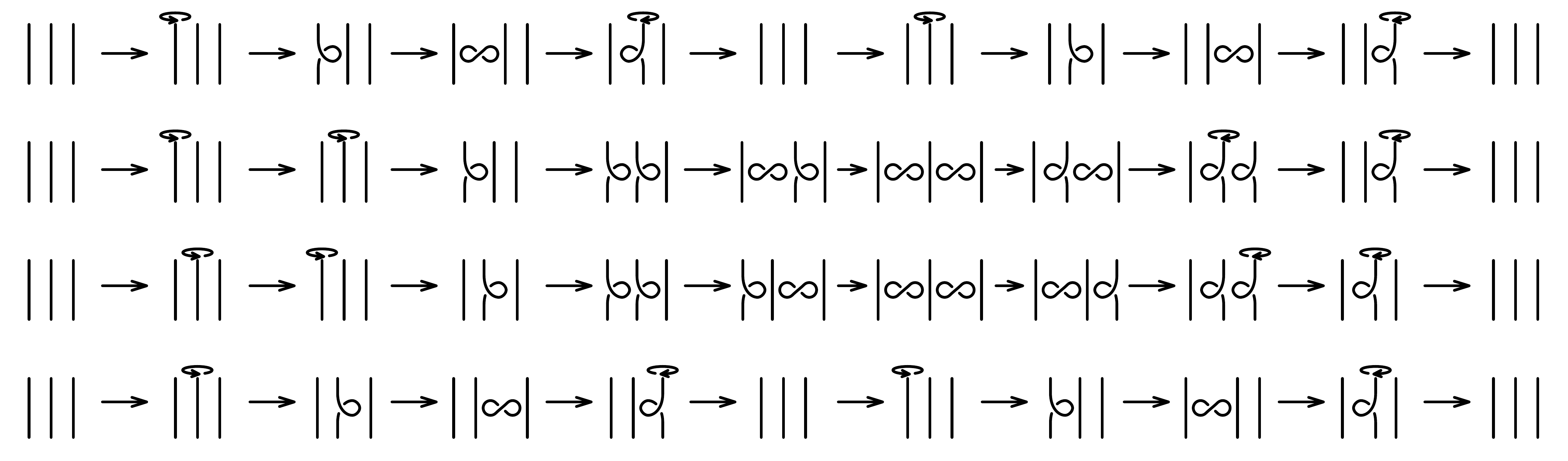}
\caption{The sequence of moves which defines
${x_i}x_{i+1}$ is shown in the top row.
The sequence of moves which defines
$x_{i+1}{x_i}$ is shown in the bottom row.
The rows in between show how the
top row can be continuously deformed
into the bottom one. Such a deformation
of two different sequences is a `movie of movies'
or a two-parameter family of configurations.
Moving to the right increases the time parameter while
moving down increases the deformation parameter
which interpolates between ${x_i}x_{i+1}$ and
$x_{i+1}{x_i}$.}
\label{fig:movie-of-movies}
\end{figure*}

Thus, we see that the equivalence
class of motions of the defects (i.e. $\pi_1$
of their configuration space) has an Abelian
subgroup generated by the $x_i$s. Since ${x_i^2}=1$
and they all commute with each other, this is simply
$n-1$ copies of $\mathbb{Z}_2$,
or, simply, $(\mathbb{Z}_2)^{n-1}$.

In order to fully determine
the group of transformations which bring
the hedgehogs and ribbons back to their initial configuration,
we need to check that the $t_i$s generate the full
set of such transformations -- i.e. that the transformations
described above and those obtained by combining them
exhaust the full set. In order to do this, we need
the commutation relations of the $t_i$s with
each other. Clearly, ${t_i}{t_j}={t_j}{t_i}$ for
$|i-j|\geq 2$ since distant operations which
do not involve the same hedgehogs nor the same
ribbons must commute. On the other hand
operations involving the same hedgehogs or
ribbons might not commute. For instance,
\begin{equation}
\label{eqn:t-x-comm}
{t_i} x_{i+1} = x_{i}x_{i+1}t_i
\end{equation}
To see why this is true, note that
if we perform $x_{i+1}$ first, then defects $i+1$ and
$i+2$ are twisted by $2\pi$. However, $t_i$ then permutes
$i$ and $i+1$ and twists $i$ by $2\pi$. Thus, the left-hand-side
permutes $i$ and $i+1$ and only twists $i+2$. The
$(i+1)^{\rm th}$ hedgehog was twisted by $x_{i+1}$
and then permuted by $t_i$ so that it ended up in the
$i^{\rm th}$ position, where it was twisted again in
the last step in $t_i$; two twists can be continuously deformed
to zero, so this hedgehog is not twisted at all.
The right-hand-side similarly permutes $i$ and $i+1$ and
only twists $i+2$ by $2\pi$. The reader may find it instructive
to flesh out the above reasoning by constructing
a movie of movies.

The multiplication rule which we have just described
(but not fully justified) is that of a semi-direct product,
which is completely natural in this context:
when followed by a permutation, a transfer of twists
ends up acting on the permuted defects.
The twists $x_i$ form the group $(\mathbb{Z}_2)^{n-1}$
which we can represent by $n$-component vectors
all of whose entries are $0$ or $1$ which satisfy the
constraint that the sum of the entries is even.
The entries tell us whether a given hedgehog is
twisted by $2\pi$ or not. In any product of $x_i$s,
an even number of hedgehogs is twisted by $2\pi$.
Now consider, for $n$ odd, the group elements given by
\begin{equation}
\label{sigma-i-odd}
{\sigma_i} = x_{n-1} x_{n-3} \ldots x_{i+3}x_{i+1} x_{i-2}x_{i-4}
\ldots x_{1}\, t_{i}
\end{equation}
for $i$ odd and
\begin{equation}
\label{sigma-i-even}
{\sigma_i} = x_{n-1} x_{n-3} \ldots x_{i+2}x_{i} x_{i-1}x_{i-3}
\ldots x_{1}\, t_{i}
\end{equation}
for $i$ even.
From (\ref{eqn:t-x-comm}), we see that ${{\sigma_i}^2}=1$.
The group element $\sigma_i$ permutes the
$i^{\rm th}$ and $(i+1)^{\rm th}$
hedgehogs and twists all of the hedgehogs.
Thus, the $\sigma_i$s generate
a copy of the permutation group $S_n$. The $\sigma_i$s do not
commute with the $x_i$s, however; instead they act
according to the semi-direct product structure noted above.
On the other hand, the situation is a bit different for
$n$ even. This may be a surprise since one might expect
that $n$ even is the same as $n$ odd but with the last
hedgehog held fixed far away. While this is true, exchanging the
last hedgehog with the others brings in an additional layer of
complexity which is not present for $n$ odd. The construction
above, Eqs. \ref{sigma-i-odd}, \ref{sigma-i-even}, does not
work. One of the hedgehogs will be left untwisted by such
a construction; since subsequent $\sigma_i$s will permute this
untwisted hedgehog with others, we must keep track of the
untwisted hedgehog and, therefore, the $\sigma_i$s will
not generate the permutation group. 
In the even hedgehog number case, the group of
tranformations has a $(\mathbb{Z}_2)^{n-1}$ subgroup,
as in the odd case, but there isn't an $S_n$ subgroup,
unlike in the odd case. To understand the even case,
it is useful to note that in both cases, every transformation
either (a) twists an even number of ribbons,
which is the subgroup $(\mathbb{Z}_2)^{n-1}$;
(b) performs an even permutation, which is the
subgroup $A_n$ of $S_n$; or (c)
twists an odd number of ribbons and performs
an odd permutation. Another way of saying this is
that the group of transformations is the `even part' of
$(\mathbb{Z}_2)^{n}\rtimes S_n$: the subgroup of
$(\mathbb{Z}_2)^{n}\rtimes S_n$ consisting of
those elements whose $(\mathbb{Z}_2)^{n}$ parity
added to their $S_{n}$ parity is even. In the
odd hedgehog number case, this is the
semidirect product $(\mathbb{Z}_2)^{n-1}\rtimes S_n$;
in the even hedgehog number case, it is not.
As we will see in Section \ref{sec:projective}, the
difference between the even and odd hedgehog number
cases is related to the fact that, for an even number
of hedgehogs, the Hilbert space decomposes into
even and odd total fermion number parity sectors.
By contrast, the situation is simpler for an
odd number of hedgehogs, where the parity
of the total fermion number is not well-defined and
the representation is irreducible.

To summarize, we have given some plausible
heuristic arguments that the `statistics' of $+1$ hedgehogs
in a model of $2N=8$ Majorana fermions
is governed by a group
$E((\mathbb{Z}_2)^{n}\rtimes S_{n})$,
the `even part' of $(\mathbb{Z}_2)^{n}\rtimes S_{n}$:
those elements of $(\mathbb{Z}_2)^{n}\rtimes S_{n}$
in which the parity of the sum of the entries of the element
in $(\mathbb{Z}_2)^{n}$ added to the parity of the
permutation in $S_{n}$ is even.
(The same group governs the $-1$
hedgehogs). Rather than devoting more time
here to precisely determining the group for
the toy model, we will move on to the problem
which is our main concern here,
a system of $2N\rightarrow\infty$ Majorana fermions.
This problem is similar, with some important differences.
(1) The target space is no longer $S^2$ but is, instead,
$U(N)/O(N)$. (2) Consequently, the defects do not carry a sign. There
is no preferred pairing into $\pm$ pairs; the defects
are all on equal footing. All $2n$ of them can be exchanged.
(3) The group obtained by computing $\pi_1$
of the space of configurations of $2n$
defects then becomes the direct product
of the `ribbon permutation group'
${\cal T}^r_{2n}$ with a trivial $\Z$,
${\cal T}_{2n}=\Z \times {\cal T}^r_{2n}$.
The ribbon permutation group ${\cal T}^r_{2n}$ is given by
${\cal T}^r_{2n} \equiv \Z_2 \times E((\mathbb{Z}_2)^{2n}\rtimes S_{2n})$, where $E((\mathbb{Z}_2)^{2n}\rtimes S_{2n})$
is the `even part' of $(\mathbb{Z}_2)^{2n}\rtimes S_{2n}$.

\section{Fundamental Group of the Multi-Defect Configuration Space}
\label{sec:Kane_space}
%%%%%%%%%%%%%%%%%%%%%%%%%%%%%%%%%%%%%

In Section \ref{sec:free-fermion} we concluded that
the effective target space for the order parameter
of a system of fermions in 3D with no symmetries is
$U(N)/O(N)$ -- which, as is conventional, we will simply call $U/O$,
dropping the $N$ in the large-$N$ limit. This enables
us to rigorously define the space of topological configurations,
$K_{2n}$, of $2n$ hedgehogs in a ball, and calculate
its fundamental group $\pi_1(K_{2n})$, thereby elucidating
Teo and Kane's \cite{Teo10} hedgehog motions and
unitary transformations.

We now outline the steps involved in this calculation:

\begin{itemize}

\item We approximate the space $U/O$ by a {\it cell complex}
(or CW complex), ${\cal C}$, a topological space constructed by
taking the union of disks of different dimensions and
specifying how the boundary of each higher-dimensional
disk is identified with a subset of the lower-dimensional disks.
This is a rather crude approximation in some respects,
but it is sufficient for a homotopy computation.

\item We divide the problem into (a) the motion of the hedgehogs
and (b) the resulting deformation of the field configuration
between the hedgehogs. This is accomplished by expressing
the configuration space in the following way. Let us call the
configuration space of $2n$ distinct points
in three dimensions $X_{2n}$. (For the sake of mathematical convenience,
we will take our physical system
to be a ball $B^3$ and stipulate that the points must lie
inside a ball $B^3$. Let's denote the space of field configurations
by ${\cal M}_{2n}$. This space is the space of maps to $U/O$
from $B^3$ with $2n$ points (at some standard locations) excised.
The latter space is denoted by
${B^3}\,\backslash\, 2n\mbox{ standard points}$.
Since we will be approximating $U/O$ by ${\cal C}$,
we can take ${\cal M}_{2n}$ to be the space of
maps from ${B^3}\,\backslash\, 2n\mbox{ standard points}$
to ${\cal C}$ with boundary conditions at the $2n$ points
specified below. Then, there is a {\it fibration} of spaces:

\[\begindc{0}[30]
    \obj(1,2)[a]{$\mathcal{M}_{2n}$}
    \obj(2,2)[b]{${K}_{2n}$}
    \obj(2,1)[c]{$X_{2n}$}
    \mor{a}{b}{}
    \mor{b}{c}{}
\enddc\]
%${\cal M}_{2n} \rightarrow K_{2n} \rightarrow X_{2n}$,

\item We introduce another two fibrations which further
divide the problem into more manageable pieces:

\[\begindc{0}[30]
    \obj(1,2)[a]{${R}_{2n}$}
    \obj(2,2)[b]{${K}_{2n}$}
    \obj(2,1)[c]{$Y_{2n}$}
    \mor{a}{b}{}
    \mor{b}{c}{}
\enddc\]

\[\begindc{0}[30]
    \obj(1,2)[a]{$\mathcal{N}_{2n}$}
    \obj(2,2)[b]{${Y}_{2n}$}
    \obj(2,1)[c]{$X_{2n}$}
    \mor{a}{b}{}
    \mor{b}{c}{}
\enddc\]

%$R_{2n}\rightarrow K_{2n} \rightarrow Y_{2n}$ and ${\cal N}_{2n} \rightarrow Y_{2n} \rightarrow X_{2n}$.
The original fibration is kind of a "fiber-product" of the two new fibrations.
Here, $R_{2n}$ is essentially the space of order parameter textures
interpolating between the hedgehogs, and $Y_{2n}$ is the space of
configurations of $2n$ points with infinitesimal spheres
surrounding each point and maps from each of these
spheres to ${\cal C}$. ${\cal N}_{2n}$ is the space
of maps from $2n$ infinitesimal spheres to ${\cal C}$,
with each one of the spheres surrounding a different
one of the $2n$ points (at some standard locations)
excised from $B^3$. We call these order parameter
maps from infinitesimal spheres to ${\cal C}$
``germs''.

\item Having broken the problem down into smaller
pieces by introducing these fibrations, we use the
fact that a fibration
$F \rightarrow E\rightarrow B$
induces a long exact sequence for homotopy groups
\begin{equation*}
\ldots \rightarrow \pi_{i}(E)\rightarrow
{\pi_i}(B)\rightarrow\pi_{i-1}(F)\rightarrow\pi_{i-1}(E)
\rightarrow ...
\end{equation*}
For instance, applying this to the
fibration ${\cal M}_{2n}
\rightarrow K_{2n} \rightarrow
X_{2n}$ leads to the exact sequence
$\ldots \rightarrow \pi_1({\cal M}_{2n})
\rightarrow \pi_1(K_{2n}) \rightarrow
\pi_1(X_{2n})\rightarrow 1$.  It follows that $\pi_1(K_{2n})$ is an extension of the permutation
group $S_{2n}=\pi_1(X_{2n})$.
By itself, the above long exact sequence is not very helpful
for computing any of the homotopy groups involved
unless we can show by independent means that
two of the homotopy groups are trivial. Then the homotopy
groups which lie between the trivial ones in the sequence
are tightly constrained.

\item We directly compute that
${\pi_1}({\cal N}_{2n})=(\mathbb{Z}_2)^{2n}$
and ${\pi_1}(X_{2n})=S_{2n}$. We show that
the homotopy exact sequence then implies that
${\pi_1}(Y_{2n})=(\mathbb{Z}_2)^{2n}\rtimes S_{2n}$.

\item We compute the homotopy groups of $R_{2n}$,
defined by the fibration $R_{2n}\rightarrow
K_{2n} \rightarrow Y_{2n}$.
This computation involves a different
way from the cell structure of thinking about the topology
of a space, called the ``Postnikov tower'', explained in
detail Appendix \ref{sec:Postnikov}. The basic idea is to approximate
a space with spaces with only a few non-trivial homotopy groups.
(This is analogous to the cell structure, which has only a few
non-trivial homology groups.) The simplest examples of such spaces
are Eilenberg-Mac Lane spaces,
which only have a single non-trivial homotopy group.
The Eilenberg-Mac Lane space $K(A,m)$ is defined for a
group $A$ and integer $m$ as the space with homotopy group
${\pi_m}(K(A,m))=A$ and ${\pi_k}(K(A,m))=0$ for all $k\neq m$.
(The group $A$ must be Abelian for $m>1$.)
Such a space exists and is unique up to homotopy.
A space $T$ with only two non-trivial homotopy groups can
be constructed through the fibration
$K(B,n) \rightarrow T \rightarrow K(A,m)$. The space $T$ has ${\pi_m}(T)=A$
and ${\pi_n}(T)=B$, as may be seen from the
corresponding long exact sequence for homotopy groups.
Continuing in this fashion, one
can construct a sequence of such approximations
$M_n$ to a space $M$. They are defined by
${\pi_k}({M_n})={\pi_k}(M)$ for $k\leq n$ and
${\pi_k}({M_n})=0$ for $k>n$. They can be constructed
iteratively from the fibration
$K(A,n) \rightarrow {M_n} \rightarrow M_{n-1}$, where ${\pi_n}(M)=A$.

\item With ${\pi_1}(Y_{2n})$, ${\pi_2}(Y_{2n})$,
${\pi_0}(R_{2n})$ and  ${\pi_1}(R_{2n})$ in hand,
we compute the desired group ${\pi_1}(K_{2n})$
from the homotopy exact sequence.

\end{itemize}

We now go through these steps in detail.
\vskip 0.25 cm

{\bf Approximating U/O by a cell complex}.
Depending on microscopic details,
gradients in the overall phase of the
fermions may be so costly that we wish
to consider only configurations in which this
overall phase is fixed. We will refer to this
as the scenario in which `phase symmetry is broken'.
In this case, the effective target space is $SU/SO$,
the non-phase factor of $U/O \cong U(1)/O(1) \times SU/SO$.
In this case, we simplify matters by replacing $SU/SO$
by $\til{U/O}$, the universal cover of $U/O$.
$\til{U/O}$ is homotopy equivalent to $SU/SO$,
so this substitution is harmless.  This substitution
results in a reduced configuration space $\til{K}_{2n}$
and we will concentrate first on calculating $\pi_1(\til{K}_{2n})$.
In an appendix, we show that this
reduction essentially makes no difference:
$\pi_1(K_{2n}) = \pi_1(\til{K}_{2n})\times \Z$.

We now define a cell complex ${\cal C}$
approximating $\til{U/O}$.
In constructing this cell structure, we are not interested
in the beautiful homogeneous nature of $\til{U/O}$
but rather only its homotopy type.  The homotopy type
of a space tells you everything you will need to know
to study {\em deformation classes} of maps either
into or out of that space.  An important feature of any
homotopy type is the list of homotopy groups
(but these are by no means a complete characterization
in general). For $\til{U/O}$, the
homotopy groups are $\pi_i(\til{U/O}) =
0, \Z_2, \Z_2, 0, \Z, 0, 0, 0, \Z$ for $i = 1, \dots, 9$ and
thereafter $\pi_i(\til{U/O})$ cycles through the last eight groups.
(For $U/O$, the first group would be $\Z$.)

Because $\til{U/O}$ is simply-connected,
but has nontrivial $\pi_2$, it natural in building a
cellular model for its homotopy type to begin with $S^2$.
Since $\pi_2 (S^2) = \Z$ and we only need
a $\Z_2$ for $\pi_2(\UO)$, we should kill off the even
elements by attaching a 3-cell $D^3$ using a
degree-2 map of its boundary 2-sphere to the original $S^2$.
For future reference, take this map to be
$(\theta, \phi) \to (2\theta, \phi)$ in a polar coordinate
system where the north pole $N = (\pi,0)$.
Similarly, a 4-cell is attached to achieve
$\pi_3(\UO) \cong \Z_2$.  The necessity of the
4-cell is proved (Fact 1) below.

The preceding logic leads us to the cell structure:
\begin{equation}
\label{eq:U/O_cell_structure}
{\cal C} = S^2 \bigcup_{\text{degree}=2} D^3 \bigcup_{\text{2Hopf}} D^4 \bigcup \text{cells of dimension} \geq 5
\end{equation}

Since we are only trying to compute the fundamental group $\pi_1(\til{K}_{2n})$ from our various homotopy long exact sequences, we do not have to figure out the higher cells (dimension $\geq 5$) of $\UO$.  We will, however, verify that $\pi_3(\UO)$ is generated by the Hopf map into the base $S^2 \subset \UO$.

To summarize, we will henceforth assume that the order
parameter takes values in the cell complex ${\cal C}$.
Although ${\cal C}$ is a crude approximation for
U/O, it is good enough for the topological calculations
which follow.

{\bf Dividing the problem into the motion of the hedgehog
centers and the deformation of the field configuration.}
Let us assume that our physical system is a ball
of material $B^3$. Let $n \geq 0$ be the number of
hedgehog pairs in the system.
A configuration in $\til{K}_{2n}$ is a texture in the
order parameter, $\Phi(x): {B^3}\rightarrow {\cal C}$,
which satisfies the following
boundary conditions at the boundary of $B^3$
and at the $2n$ hedgehog locations (which
are singularities in the order parameter).
The order parameter has winding number $0$
at the boundary of the ball, $\partial B^3$
and winding number $1$ around each of the
hedgehog centers. (Recall that ${\pi_2}({\cal C})=\mathbb{Z}_2$,
so the winding number can only be $0$ or $1$).

From its definition, $\til{K}_{2n}$ is the total space of a fibration:

\[\begindc{0}[30]
    \obj(1,2)[a]{$\mathcal{M}_{2n}$}
    \obj(2,2)[b]{$\til{K}_{2n}$}
    \obj(2,1)[c]{$X_{2n}$}
    \mor{a}{b}{}
    \mor{b}{c}{}
\enddc\]

\noindent The above diagram suggests that we
should think of the fibration
$\mathcal{M}_{2n}
\rightarrow \til{K}_{2n} \rightarrow X_{2n}$
in the following way: above each point in
$X_{2n}$ there is a fiber $\mathcal{M}_{2n}$;
the total space formed thereby is $\til{K}_{2n}$.
(This is not quite a fiber bundle, since we
do not require that there be local coordinate
charts in which $\til{K}_{2n}$ is simply the
direct product.) Here, $X_{2n}$ is the simply the
configuration space of $2n$ distinct points in $B^3$.
We write this formally as
$X_{2n} = \prod_{i=1}^{2n} B^3 \setminus \text{big diagonal}$.
(The big diagonal consists of $2n$-tuples of points in $B^3$
where at least two entries are identical.)
The space $\mathcal{M}_{2n}$ consists of maps
from $B^3 \setminus 2n \text{ points in a fixed standard position}$
to ${\cal C}$ with the prescribed winding numbers given
in the preceding paragraph.

\begin{figure}[htpb]
%\labellist \small\hair 2pt
%  \pinlabel $\text{Maps}((B^3\setminus 2n\text{ standard points})\to \til{U/O})$ at 170 206
%\endlabellist
\centering
\includegraphics[scale=0.5]{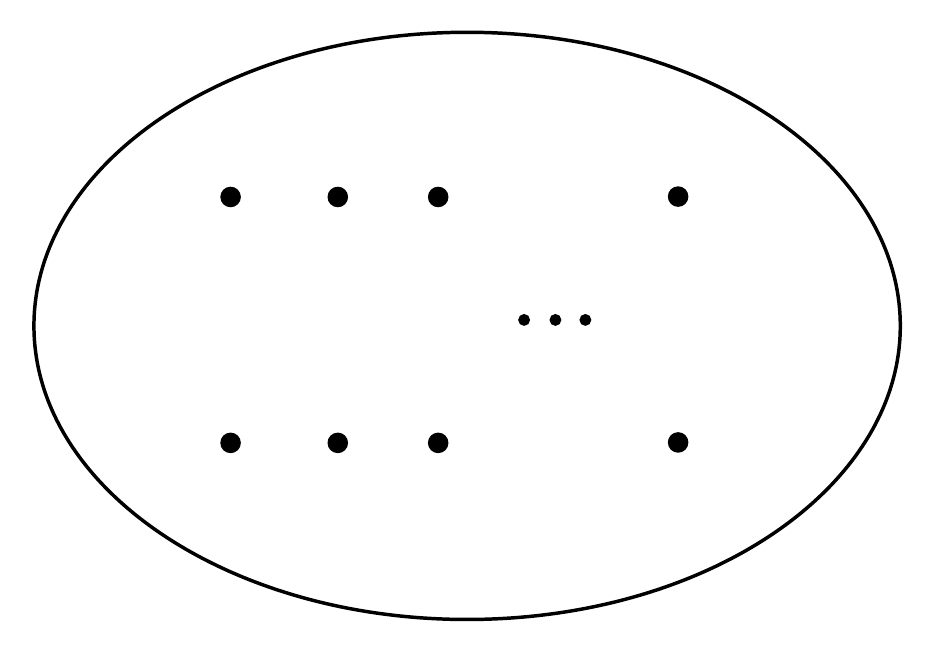}
\caption{$B^3 \setminus 2n$ points in standard position.
The space $\mathcal{M}_{2n}$ consists of maps from
this manifold to ${\cal C}$.}
\label{fig:M2n}
\end{figure}

{\bf Germs of order parameter textures.}
It is helpful to introduce an intermediate step in the fibration.
Define a point in $Y_{2n}$ as a configuration in
$X_{2n}$ together with a ``germ''
of $\Phi(x)$, which we call $\til{\Phi}(x)$,
defined only near $\partial B^3$
and the $2n$ points. The idea behind the germ
$\til{\Phi}(x)$ is to forget about the order parameter
$\Phi(x)$ except for its behavior in an infinitesimal
neighborhood around each hedgehog center
and at the boundary of the system.
$\til{\Phi}(x)$ must satisfy the same boundary conditions as
$\Phi(x)$ itself. We take $\til{\Phi}(x)$ to be
constant on $\partial B^3$ and to have
winding number $1$ around each of the
hedgehog centers. With this definition,
we now have the fibration:

\[\begindc{0}[30]
    \obj(1,2)[a]{$\mathcal{N}_{2n}$}
    \obj(2,2)[b]{$Y_{2n}$}
    \obj(2,1)[c]{$X_{2n}$}
    \mor{a}{b}{}
    \mor{b}{c}{}
\enddc\]

\noindent where $\mathcal{N}_{2n}$ is the space of
(germs of) order parameter textures $\til{\Phi}$ from the
neighborhoods of the $2n$ fixed standard points
and $\partial B^3$ to ${\cal C}$.
We will henceforth replace discussion of germs
with the equivalent and simpler concept of maps
on $\partial B^3 \cup \left( \bigcup_{i=1}^{2n} S_i^2 \right)$
where $S_i^2$ is a small sphere surrounding
the $i$th standard point.  Thus,
\begin{equation}
\mathcal{N}_{2n} \subset
\text{Maps}\biggl(\Bigl(\partial B^3 \cup \bigcup_{i=1}^{2n} S_i^2
\Bigr) \to {\cal C} \biggr).
 \end{equation}

We now define $Q_{2n}$ as the ball $B^3$
with small balls (denoted below by interior$(S_i^2)$)
centered about the hedgehogs
deleted:
\begin{equation}
Q_{2n}= \Bigl(B^3 \setminus \bigcup_{i=1}^{2n} \text{interior}(S_i^2)
\Bigr)
\end{equation}
for fixed standard positions $i = 1, \dots , 2n$.
Then $R_{2n}$ is the space of order parameter textures on
$Q_{2n}$ which satisfy the boundary condition
that the winding number is $0$ on $\partial B^3$
and $1$ on each of the small spheres.
With this definition, we have the fibration:

\[\begindc{0}[30]
    \obj(1,2)[a]{$R_{2n}$}
    \obj(2,2)[b]{$\til{K}_{2n}$}
    \obj(2,1)[c]{$Y_{2n}$}
    \mor{a}{b}{}
    \mor{b}{c}{}
    \label{eqn:R_2n-def}
\enddc\]

\noindent Given the cell structure ${\cal C}$, we can
specify the boundary conditions for the order
parameter precisely. On $\partial B^3$,
the order parameter is equal to the North Pole
in $S^2$. (Recall that ${S^2}\subset{\cal C}$ is
the bottom cell of the structure ${\cal C}$ which
we are using to approximate U/O.)  On each of
the spheres $S_i^2$, the order parameter
$\Phi(x)$ defines a map from ${S_i^2}\rightarrow S^2$
which is the identity map
(where, again $S^2$ is understood as a subset
of the order parameter space ${S^2}\subset{\cal C}$).
This ensures that the order parameter has the
correct winding numbers at the boundaries of $Q_{2n}$.
In essence, what we have done in writing
Eq. \ref{eqn:R_2n-def} is to break up an
order parameter texture containing hedgehogs
into (a) the hedgehogs together with the order parameter
on infinitesimal neighborhoods around them (i.e. `germs')
and (b) order parameter textures in the intervening regions
between the hedgehogs. The space of configurations (a)
is $Y_{2n}$; the space of configurations (b) is $R_{2n}$.

The name $R_{2n}$ is for ``ribbons.''
As we saw in Section \ref{sec:tethered},
if the order parameter manifold were $S^2$,
we could summarize an order parameter texture
by looking at the inverse image of the North Pole
$N \subset S^2$ and a fixed tangent vector at the North Pole.
The inverse images form a collection of ribbons.
Now, the order parameter manifold is actually
(approximated by) ${\cal C}$, but the bottom cell
in ${\cal C}$ is $S^2$. The effect of the 3-cell
is that hedgehogs lose their sign, so  there is no
well-defined ``arrow'' running lengthwise along the ribbons.
The 4-cell allows the ``twist'' or framings of ribbons to be
altered at will by $\pm 2$.

{\bf Long exact sequence for homotopy groups.}
It is very convenient to use fibrations to calculate homotopy groups.
(For those interested in $K$-theory, the last two chapters of
Milnor's {\em Morse Theory} \cite{Milnor63}
are a must read and exhibit these methods with clarity.) As noted above,
fibrations have all the homotopy properties of fiber bundles
but are (often) found arising between function spaces
where it would be a lot of work -- and probably a distraction from important business -- to attempt to verify the existence of
locally trivial coordinate charts. Operationally, fibrations share with fiber bundles the all-important ``homotopy long exact sequence'':

\[\begindc{0}[30]
    \obj(1,2)[a]{$F$}
    \obj(2,2)[b]{$E$}
    \obj(2,1)[c]{$B$}
    \mor{a}{b}{}
    \mor{b}{c}{}
\enddc\]

\noindent we have:
$$\cdots \to \pi_{i+1}(B) \to \pi_i(F) \to \pi_i(E) \to \pi_i(B) \to \pi_{i-1}(F) \to \cdots$$

We now compute $\pi_1(Y_{2n})$ from the exact sequence:
%\begin{equation}\label{eq:pi_1(Y_2n)_computation}
%\begindc{0}[3]
   %\obj(10,70)[a]{$\pi_2(X_{2n})$}
   % \obj(30,70)[b]{$\pi_1(\mathcal{N}_{2n})$}
   % \obj(50,70)[c]{$\pi_1(Y_{2n})$}
    %\obj(70,70)[d]{$\pi_1(X_{2n})$}
    %\obj(90,70)[e]{$\pi_0(\mathcal{N}_{2n})$}
    %\obj(30,65)[f]{\begin{sideways}$\cong$\end{sideways}}
    %\obj(30,60)[g]{$\prod_{i=1}^{2n} \pi_1(\text{Maps}(S_i^2,\UO))$}
    %\obj(30,55)[h]{\begin{sideways}$\cong$\end{sideways}}
    %\obj(30,50)[i]{$\prod_{2n\text{ copies}}\pi_3(\UO)$}
    %\obj(30,45)[j]{\begin{sideways}$\cong$\end{sideways}}
    %\obj(30,40)[k]{$(\Z_2)^{2n}$}
    %\obj(70,65)[l]{\begin{sideways}$\cong$\end{sideways}}
    %\obj(70,60)[m]{$S_{2n}$}
    %\mor{a}{b}{$\partial$}
    %\mor{b}{c}{}
    %\mor{c}{d}{}
    %\mor{d}{e}{$\partial$}
%\enddc
%\end{equation}
\begin{equation}
\label{eq:pi_1(Y_2n)_computation}
\begindc{0}[3]
    \obj(10,70)[a]{$\pi_2(X_{2n})$}
    \obj(28,70)[b]{$\pi_1(\mathcal{N}_{2n})$}
    \obj(45,70)[c]{$\pi_1(Y_{2n})$}
    \obj(62,70)[d]{$\pi_1(X_{2n})$}
    \obj(80,70)[e]{$\pi_0(\mathcal{N}_{2n})$}
    \mor{a}{b}{$\partial$}
    \mor{b}{c}{}
    \mor{c}{d}{}
    \mor{d}{e}{$\partial$}
\enddc
\end{equation}

\noindent
We can compute two of the homotopy
groups in this equation by inspection.
$\pi_1(X_{2n})$ is clearly the symmetric group of point exchange:
\begin{equation}
\label{eqn:pi_1_X}
\pi_1(X_{2n}) = S_{2n} .
\end{equation}
Meanwhile, $\pi_1(\mathcal{N}_{2n})$ amounts to (products of)
loops of maps from the $S_i^2$ to ${\cal C}$
and reduces to $2n$ copies of the third homotopy group
of ${\cal C}$ (and, therefore, to ${\pi_3}(\UO)$).
Thus, $\pi_1(\mathcal{N}_{2n})=(\mathbb{Z}_2)^{2n}$:
\begin{eqnarray}
\label{eqn:pi_1_N}
\pi_1(\mathcal{N}_{2n}) = \prod_{i=1}^{2n} \pi_1(\text{Maps}(S_i^2,\UO)) &=& \prod_{2n\text{ copies}}\pi_3(\UO)\cr &=&
(\Z_2)^{2n} .
\end{eqnarray}

To proceed further, we need to evaluate {\em boundary maps} in the homotopy exact sequence. In Appendix \ref{sec:hopf-map},
we explain boundary maps through the example of the Hopf map.
Consider Eq. \ref{eq:pi_1(Y_2n)_computation}.
$\pi_2(X_{2n})$ is generated by the $2n\choose{2}$
different 2-parameter motions in which a pair of hedgehogs
come close together and explore the 2-sphere of possible relative positions around their center of mass.  This 2-parameter family of motions involves no ``rotation'' of the maps $\til\Phi$
which describe $\pi_1(\mathcal{N}_{2n})$
(i.e the order parameter configuration in the neighborhood
of each hedgehog does not rotate as the hedgehogs are moved),
so the left most $\partial$ map in Eq. \ref{eq:pi_1(Y_2n)_computation}
is zero.  Similarly, a simple exchange of hedgehogs produces no twist
of the order parameter configuration in the neighborhood
of either hedgehog, so the second $\partial$
map of Eq. \ref{eq:pi_1(Y_2n)_computation} is also zero.
Thus, we have a short exact sequence:
\[\begindc{0}[3]
    \obj(10,30)[a]{$1$}
    \obj(23,30)[b]{$\Z_2^{2n}$}
    \obj(41,30)[c]{$\pi_1(Y_{2n})$}
    \obj(60,30)[d]{$S_{2n}$}
    \obj(75,30)[e]{$1$}
    \mor{a}{b}{}
    \mor{b}{c}{$\alpha$}
    \mor{c}{d}{$\beta$}
    \mor{d}{e}{}
\enddc\]
To derive this short exact sequence, we used
the triviality of the boundary maps noted above
and Eqs. \ref{eqn:pi_1_X}, \ref{eqn:pi_1_N}
to simplify Eq. \ref{eq:pi_1(Y_2n)_computation}.

There is a natural group homomorphism $s$:
$$s:S_{2n}\rightarrow\pi_1(Y_{2n})$$
which associates to each permutation a motion
of hedgehogs which permutes the hedgehogs in $Y_{2n}$
but does not rotate the order parameter configurations
$\til\Phi$ near the hedgehogs. Then
$\beta \circ s = id_{S_{2n}}$. In other words,
the sequence is {\em split}:
\[\begindc{0}[3]
    \obj(10,30)[a]{$1$}
    \obj(25,30)[b]{$\Z_2^{2n}$}
    \obj(40,30)[c]{$\pi_1(Y_{2n})$}
    \obj(60,30)[d]{$S_{2n}$}
    \obj(75,30)[e]{$1$}
    \mor{a}{b}{}
    \mor{b}{c}{$\alpha$}
    \mor{d}{e}{}
    \cmor((46,31)(51,33)(57,31)) \pright(51,35){$\beta$}
    \cmor((57,29)(51,27)(46,29)) \pleft(51,25){$s$}[\atleft,\dasharrow]
\enddc\]
Thus, $\pi_1(Y_{2n})$ is a semi-direct product.
To determine $\pi_1(Y_{2n})$ completely, it only remains to
identify how $s(S_{2n})$ acts on the twist factors $\Z_2^{2n}$
under conjugation. It is quite clear that this action is the only
natural one available: $s(p)$ acts on $\Z_2^{2n}$ by applying
the permutation $p$ to the $2n$ coordinates of $Z_2^{2n}$.
So, $\pi_1(Y_{2n}) \cong \Z_2^{2n} \rtimes S_{2n}$ with group law:

\begin{equation}\label{eq:pi_1(Y2n)_group_multiplication}
(v,p) \circ (v',p') = (v + p(v'), p \circ p')
\end{equation}

\noindent where $v \in \Z_2^{2n}$ is a $\Z_2$-vector, $p \in S_{2n}$ a permutation, and $p(v')$ the natural action of $S_{2n}$ on $\Z_2^{2n}$ applied to $v'$. Note that this is precisely the multiplication
rule which we obtained pictorially in Section \ref{sec:tethered}.

{\bf Computing the homotopy groups of $R_{2n}$,
the space of order parameter textures interpolating
between the hedgehogs}.
Of course, computing $\pi_1(Y_{2n})$ only gets us part
of the way home. Our ultimate goal is to compute
$\pi_1({\til K}_{2n})$. Thus, we now turn to
the homotopy long exact sequence:
$$\pi_2(Y_{2n}) \overset{\partial_2}{\longrightarrow} \pi_1(R_{2n}) \rightarrow \pi_1(\til{K}_{2n}) \rightarrow \pi_1(Y_{2n}) \overset{\partial_1}{\longrightarrow} \pi_0(R_{2n})$$
First consider $\partial_1$.  The kernel of $\partial_1$ is represented by loops in $Y_{2n}$ which extend to loops in $R_{2n}$.
A loop $\gamma$ in $Y_{2n}$ is a motion of the hedgehogs
together with rotations (about the spheres $S_i^2$)
of $\til\Phi$ which brings the system back to
its initial configuration. If a loop $\gamma$ is in
kernel of $\partial_1$, then
there is a corresponding loop in the configuration space of
ribbons in $B^3$ (obtained by lifting $\gamma$ to
$\til{K}_{2n}$).

The next steps are to compute ${\pi_0}(R_{2n})$
and ${\pi_1}(R_{2n})$. These computations are
detailed in Appendix \ref{sec:Postnikov}, where we see
that using the cell structure
${\cal C}$ which we introduced for $\UO$ is
tricky as a result of the higher cells.
Thus, we instead introduce the
``Postnikov tower'' for $\UO$ which allows
us to make all the calculations we need.
We find that $\pi_0(R_{2n})=\mathbb{Z}_2$
and ${\pi_1}(R_{2n})=(\mathbb{Z}_2)^{2n}$.

Thus, $\pi_1(\til{K}_{2n})$ sits in the following exact sequence:
\[\begindc{0}[3]
    \obj(10,20)[a]{$\pi_2(Y_{2n})$}
    \obj(28,20)[b]{$\pi_1(R_{2n})$}
    \obj(45,20)[c]{$\pi_1(\til{K}_{2n})$}
    \obj(62,20)[d]{$\pi_1(Y_{2n})$}
    \obj(81,20)[e]{$\pi_0(R_{2n})$}
    \obj(28,15)[b1]{\begin{sideways}$\cong$\end{sideways}}
    \obj(62,15)[c1]{\begin{sideways}$\cong$\end{sideways}}
    \obj(81,15)[c1]{\begin{sideways}$\cong$\end{sideways}}
    \obj(28,10)[b2]{$(\mathbb{Z}_2)^{2n}$}
    \obj(62,10)[d2]{$(\Z_2)^{2n} \rtimes S_{2n}$}
    \obj(81,10)[e2]{$\Z_2$}
    \mor{a}{b}{$\partial_2$}
    \mor{b}{c}{}
    \mor{c}{d}{}
    \mor{d}{e}{$\partial_1$}
 %   \mor(75,20)(87,20){}
\enddc\]
Recall that when we studied $\pi_2(Y_{2n})$, we found
(exactly as in the case of $\pi_2(X_{2n})$) that there are
$2n\choose{2}$ generators corresponding to relative
2-parameter motions of any pair of hedgehogs
around their center of mass.  This can be
used to understand the map
$\partial_2: \pi_2(Y_{2n}) \to \pi_2(R_{2n})$.
The image of any center of mass 2-motion is the ``bag''
containing the corresponding pair of hedgehogs.
Thus, $\text{coker}(\partial_2) \cong \Z_2$; it is
$\Z_2^{2n}$ modulo the even sublattice
(vectors whose coordinate sum is zero in $\Z_2$).
Thus, we have the short exact sequence:
\begin{equation}\label{eq:extension_problem}
\begindc{0}[3]
    \obj(10,20)[a]{$1$}
    \obj(27,20)[b]{$\text{coker}(\partial_2)$}
    \obj(44,20)[c]{$\pi_1(\til{K}_{2n})$}
    \obj(61,20)[d]{$\text{ker}(\partial_1)$}
    \obj(78,20)[e]{$1$}
    \obj(27,15)[c1]{\begin{sideways}$\cong$\end{sideways}}
    \obj(27,10)[b2]{$\Z_2$}
    \mor{a}{b}{}
    \mor{b}{c}{}
    \mor{c}{d}{$\pi$}
    \mor{d}{e}{}
\enddc
\end{equation}

The kernel $\text{ker}(\partial_1)$ consists of the part of $\pi_1(Y)$ associated with even ($2\pi$) twisting.  As shown in Figure \ref{fig:exchange2}, a simple exchange is associated to a total twisting of ribbons by $\pm 2\pi$.

%\begin{figure}[htpb]
%\centering
%\includegraphics[width=3.5in]{ribbontwist}
%\caption{} \label{fig:ribbontwist}
%\end{figure}

Thus, $\partial_1(v,p) = \sum_{i=1}^{2n} v_i + \text{parity}(p) \in \Z_2$.  We use the notation $E(\mathbb{Z}_2^m \rtimes {S_m})$ for $\text{ker}(\partial_1)$.

{\em Note:} If $m=2$, $\Z_2^{m} \rtimes S_{m}$ is the dihedral group $D_4$ and its ``even'' subgroup $\text{ker}(\partial_1) \cong \Z_4$.  This shows that for $m$ even, the induced short exact sequence does not split, and the extension is more complicated:

\begin{equation*}
1 \longrightarrow \Z_2^{2n-1} \longrightarrow E(\mathbb{Z}_2^{2n} \rtimes S_{2n}) \longrightarrow S_{2n}\rightarrow 1
\end{equation*}

There is a final step required to solve the extension problem \ref{eq:extension_problem} and finish the calculation of $\pi_1(\til{K}_{2n})$.  We geometrically construct a homomorphism $s:\text{ker}(\partial_1) = E(\mathbb{Z}_2^{2n} \rtimes S_{2n}) \to \pi_1(\til{K}_{2n})$ which is a left inverse to the projection.

This will show that $\pi_1(\til{K}_{2n})$ is a semidirect product $\Z_2 \rtimes \text{ker}(\partial_1)$, but since $\Z_2$ has no nontrivial automorphism, the semidirect product is actually direct:

\begin{equation}\label{eq:pi_1(K_2n)_direct_product}
\pi_1(\til{K}_{2n}) \cong \Z_2 \times E(\mathbb{Z}_2^{2n} \rtimes S_{2n})\end{equation}

To construct $s$, note that all elements of $\text{ker}(\partial_1)$ can be realized by a loop $\gamma$ of maps into the bottom 2-cell of Eq. \ref{eq:U/O_cell_structure} $S^2$.  Still confining the order parameter (map) to lie in $S^2$, such a loop lifts to an arc $\til{\gamma}$ of ribbons representing an arc in $\til{K}_{2n}$.  We may choose the lift so that as the ribbons move, they never ``pass behind'' the $2n$ hedgehogs.  (For example, we may place the hedgehogs on the sphere of radius $=\frac{1}{2}$ inside the 3-ball $B^3$ (assumed to have radius $=1$) and then keep all ribbons inside $B_{\frac{1}{2}}^3$.  These arcs may be surgered (still as preimages of $N \subset S^2$) so that they return to their original position except for a possible accumulation of normal twisting $t2\pi$.  Since $\gamma \in \text{ker}(\partial_1)$, $t$ must be even.  Now, allowing the order parameter (map) to leave $S^2$ and pass over the 4-cell (of Eq. \ref{eq:U/O_cell_structure}), attached by $2\text{Hopf}:S^3 \to S^2$, we may remove these even twists.  (The 4-cell can introduce small closed ribbons with self-linking $=2$ in a small ball.  These small ribbons can be surgered into other ribbons.)  This lifts a generating set of $\text{ker}(\partial_1)$ into $\pi_1(\til{K}_{2n})$ as a set theoretic cross section (left inverse to $\pi$). But what about relations?  Because the entire loop is constant outside $B_{\frac{1}{2}}^3$, the corresponding homology class in $H_2(Q \times I;\Z_2) \cong \pi_1(R)$ is trivial, so $s$ is actually a group homomorphism.

\section{Representation Theory of the Ribbon Permutation Group}
\label{sec:projective}

In this section, we discuss the mathematics of the group
${\cal T}_{m}$ and its representation. The
purpose of this section is to show that
a direct factor of ${\cal T}_{m}$,
called the even ribbon permutation group,
is a ghostly recollection of the braid
group and the Teo-Kane representation of the even ribbon permutation group is a projectivized version of
the Jones representation of the braid group
at a $4^{th}$ root of unity, i.e. the representation
relevant to Ising anyons.

\subsection{Teo-Kane fundamental groups}

In Section V, we consider only even number of hedgehogs for
physical reasons.  In this section, we will include
the odd case for mathematical completeness.
\iffalse For each integer $m \geq 0$, let $B^3_m$ be the closed complement
of $m$ disjoint $3$-balls $D^3_i, i=1,2,\cdots, m$ in the interior
of $B^3$.  A Teo-Kane configuration is a continuous map $\Phi:
B^3_m \longrightarrow U/O$ such that $[\Phi|_{\partial B^3}]=0\in
\pi_2(U/O)$, and $[\Phi|_{\partial D^3_i}]=1\in
\pi_2(U/O), i=1,2,\cdots, m$,  i.e. the map $\Phi|_{\partial B^3}: S^2\rightarrow U/O$ represents $0$ in
$\pi_2(U/O)\cong \mathbb{Z}_2$, and $\Phi|_{\partial D_i^3}, i=1,2,\cdots, m$ representing $1$.
 Then the Teo-Kane configuration space $K_m$ is
the topological space of all such maps with the compact-open topology, where the $3$-balls $D^3_i, i=1,2,\cdots, m$ are not
fixed.\fi

The Teo-Kane fundamental
group is the fundamental group of the Teo-Kane configuration
space $K_m$.  As computed in Section V, ${\cal T}_{m}=\pi_1(K_m)\cong \mathbb{Z}\times \mathbb{Z}_2 \times
E(\mathbb{Z}_2^m \rtimes {S_m})$, where
the subgroup ${\cal T}^r_{m}=\mathbb{Z}_2 \times
E({\mathbb{Z}_2^m \rtimes S_m})$ is called the ribbon permutation group. Here, $E({\mathbb{Z}_2^m \rtimes S_m})$ is
the subgroup of ${\mathbb{Z}_2^m \rtimes S_m}$ comprised of
elements whose total parity in $\mathbb{Z}_2^m$
added to their parity in $S_m$ is even.
In the following, we will call
the group $G_m=E(\mathbb{Z}_2^m \rtimes {S_m})$ the {\it even} ribbon permutation group because
it consists of the part of $\pi_1(Y_m)$ associated with even ($2\pi$) twisting.
For the representations of the Teo-Kane fundamental
groups, we will focus on the even ribbon permutation groups $G_m$.
No generality is lost if we consider only irreducible representations
projectively because irreducibles of $\mathbb{Z}$ and $\mathbb{Z}_2$ contribute only overall phases.
But for reducible representations, the relative phases from representations of $\mathbb{Z}$
and $\mathbb{Z}_2$ might have physical consequences in interferometer
experiments.

The even ribbon permutation group $G_m$ is an index$=2$ subgroup of
$\mathbb{Z}_2^m \rtimes S_m$.  To have a better understanding of $G_m$, we
recall some facts about the important group $\mathbb{Z}_2^m \rtimes S_m$.
The group $\mathbb{Z}_2^m \rtimes S_m$ is the symmetry group of the
hypercube $\mathbb{Z}_2^m$, therefore it is called the hyperoctahedral group, denoted as
$C_m$.  $C_m$ is also a Coxeter group of type
$B_m$ or $C_m$, so in the mathematical literature it is also
denoted as $B_m$ or $BC_m$.  To avoid confusion with the braid
group ${\cal B}_m$, we choose to use the hyperoctahedral group notation $C_m$.
The group $C_m$ has a faithful representation as signed
permutation matrices in the orthogonal group $O(m)$: matrices with
exactly one non-zero entry $\pm 1$ in each row and column.
Therefore, it can also be realized as a subgroup of the
permutation group $S_{2m}$, called signed permutations: $\sigma:
\{\pm 1, \pm 2, \cdots, \pm m\} \rightarrow \{\pm 1, \pm 2, \cdots, \pm
m\}$ such that $\sigma(-i)=-\sigma(i)$.

We will denote elements in $C_m$ by a pair $(b,g)$, where
$b=(b_i)\in \mathbb{Z}_2^m$ and $g\in S_m$.  Recall the multiplication of
two elements $(b,g)$ and $(c,h)$ is given by $(b,g)\cdot (c,h)=(b+g.c, gh)$, where $g.c$
is the action of $g$ on $c$ by permuting its coordinates.  Let $\{e_i\}$ be the standard
basis elements of $\mathbb{R}^m$.  To save notation, we will also use it
for the basis element of $\mathbb{Z}_2^m$.  As a signed permutation matrix $e_i$ introduces
a $-1$ into the $i^{th}$ coordinate $x_i$.  Let $s_i$ be the transposition
of $S_m$ that interchanges $i$ and $i+1$.  As a signed permutation matrix, it
interchanges the coordinates $x_i, x_{i+1}$.
There is a total parity map $det: C_m\rightarrow \mathbb{Z}_2$ defined as
$det(b,g)=\sum_{i=1}^m b_i +parity(g)\; mod \; 2$.  We denote
the total parity map as $det$ because in the realization of $C_m$
as signed permutation matrices in $O(m)$, the total parity is just the
determinant.  Hence $G_m$, as the kernel of $det$,  can be identified as a subgroup of
$SO(m)$.  The set of elements
$t_i=(e_i,s_i), i=1,\cdots, m-1$ generates $G_m$.  As a signed permutation matrix,
$t_i(x_1,\cdots,x_i,x_{i+1},\cdots, x_m)=(x_1,\cdots,-x_{i+1},x_i,\cdots,
x_m)$.

Given an element $(b,g)\in C_m$, let $\mathbb{Z}_2^{m-1}$ be
identified as the subgroup of $\mathbb{Z}_2^m$ such that $\sum_{i=1}^{m} b_i$ is
even.  Then we have:

\begin{proposition}\label{presentation}

\begin{enumerate}

\item For $m\geq 2$, the even ribbon permutation group $G_m$ has a presentation as an abstract group
\begin{equation*}
\begin{split}
<t_1, \cdots, t_{m-1}| t_i^4=1, i=1,\cdots, m-1,\\
(t_it_{i+1}^{-1})^3=(t_i^{-1}t_{i+1})^3=1, i=1, \cdots, m-2>.
\end{split}
\end{equation*}

\item The exact sequence $$1\rightarrow \mathbb{Z}_2^{m-1} \rightarrow G_m \rightarrow S_m
\rightarrow 1$$ splits if and only if $m$ is odd.

\item When $m$ is even, a normalized $2$-cocycle
$f(g,h):  S_m\times S_m \rightarrow \mathbb{Z}_2^{m-1}$ for the extension of $S_m$ by $\mathbb{Z}_2^{m-1}$
above can be
chosen as $f(g,h)=0$ if $g$ or $h$ is even and $f(g,h)=e_{g(1)}$
if $g$ and $h$ are both odd.

\end{enumerate}
\end{proposition}

We briefly give the idea of the proof of Prop. \ref{presentation}.  For $(1)$, first
we use a presentation of $C_m$ as a Coxeter group of type
$B_m$: $<r_1,\cdots, r_m|r_i^2=1, (r_1r_2)^4=1, (r_ir_{i+1})^3=1,
i=2,\cdots, m-1>$.  Then the Reidemeister-Schreier method \cite{Magnus76} allows
us to deduce the presentation for $G_m$ above.  For $(2)$, when
$m$ is odd, a section $s$ for the splitting can be defined as
$s(g)=(0,g)$ if $g$ is even and $s(g)=((11\cdots 1),g)$ if $g$ is
odd.  When $m$ is even, that the sequence does not split follows from
the argument in \cite{Jones83}.  For $(3)$, we choose a set map
$s(g)=(0,g)$ if $g$ is even and $s(g)=(e_1,g)$ if $g$ is odd.
Then a direct computation of the associated factor set as on page $91$ of \cite{Brown82} gives rise
to our $2$-cocycle.

As a remark, we note that there are another two obvious maps from the
hyperoctahedral group  $C_m$ to $\mathbb{Z}_2$.  One of them is the
sum of bits in $b$ of $(b,g)$.  The kernel of this map is the
Coxeter group of type $D_m$, which is a semi-direct product of $\mathbb{Z}_2^{m-1}$
with $S_m$.  The two groups $D_m$ and $G_m$ have
the same order, and are isomorphic when $m$ is odd (the two splittings induce the same
action of $S_m$ on $\mathbb{Z}_2^{m-1}$), but different
when $m$ is even.  To see the difference, consider the order $2$ automorphism $\phi: C_m\rightarrow C_m$  given by $\phi(x)=det(x) x$.
Its restriction is the identity on $G_m$, but non-trivial on $D_m$.

\subsection{Teo-Kane unitary transformations}

Suppose there are $m$ hedgehogs in $B^3$.  A unitary transformation
$T_{ij}=e^{\frac{\pi}{4}\gamma_i \gamma_j}$ of the
Majorana fermions is associated with the interchange of the ($i$,$j$)-pair
of the
hedgehogs. Interchanging the same pair twice results in the
\lq\lq braidless" operation $T_{ij}^2=\gamma_i \gamma_j$.  The
Majorana fermions $\gamma_i, i=1,2,\cdots, m$ form the Clifford
algebra $Cl_m(\mathbb{C})$.  Therefore, the Teo-Kane unitary
transformations act as automorphisms of the Clifford algebra
$$T_{ij}: \gamma \rightarrow T_{ij} \gamma T_{ij}^{\dagger}.$$
The projective nature of the Teo-Kane representation rears its
head here already as an overall phase cannot be constrained by
such actions.
Do these unitary transformations afford a linear representation
of the Teo-Kane fundamental group $K_m$? If so, what are their images?
The surprising fact is that the Teo-Kane
unitary transformations cannot give rise to a linear representation of the
Teo-Kane fundamental group.  The resulting representation is
intrinsically projective.   We consider only the even ribbon permutation group
$G_m$ from now on.

To define the Teo-Kane representation of $G_m$, we use the presentation of $G_m$
by $t_i, i=1,\cdots, m-1$ in Prop. \ref{presentation}.  The associated unitary matrix for $t_i$ comes from the
Teo-Kane unitary transformation $T_{i,i+1}$.  As was alluded above, there is a phase
ambiguity of the Teo-Kane unitary matrix.  We will discuss this ambiguity more in the
next subsection.  For the discussion below, we choose any matrix realization of the
Teo-Kane unitary transformation with respect to a basis of the Clifford algebra $Cl_m(\mathbb{C})$.
A simple computation using the presentation of $G_m$ in Prop. \ref{presentation}
verifies that the assignment of $T_{i,i+1}$ to $t_i$ indeed
leads to a representation of $G_m$.  Another verification follows from a relation to the Jones
representation in the next subsection.  We can also check directly that
this is indeed the right assignment for
$T^2_{i,i+1}: \gamma_i \rightarrow -\gamma_i,\gamma_{i+1} \rightarrow -\gamma_{i+1}$.
In the Clifford algebra $Cl_m(\mathbb{C})$, $\gamma_i, \gamma_{i+1}$
correspond to the standard basis element $e_i, e_{i+1}$.  The
element $t_i^2$ of $G_m$ is $(e_i+e_{i+1},1)$.  As a signed
permutation matrix, $t_i^2$ sends $(x_1,\cdots, x_i,
x_{i+1},\cdots, x_m)$ to $(x_1,\cdots, -x_{i},-x_{i+1},\cdots, x_m)$,
which agrees with the action of $T^2_{i,i+1}$ on $\gamma_i,
\gamma_{i+1}$.

To see the projectivity, the interchange of the ($i$,$i+1$)-pair hedgehogs corresponds
to the element $t_i=(e_i, s_i) \in G_m$.  Performing the
interchange twice gives rise to the element $t_i^2=(e_i+e_{i+1},1)$, denoted as $x_i$.  Since
$x_i$'s are elements of a subgroup of $G_m$ isomorphic $\mathbb{Z}_2^{m-1}$,
obviously we have
$x_ix_{i+1}=x_{i+1}x_i$.  On the other hand,
$T^2_{i,i+1}T^2_{i+1,i+2}=\gamma_i \gamma_{i+1} \cdot \gamma_{i+1}
\gamma_{i+2}=\gamma_i \gamma_{i+2}$, and $T^2_{i+1,i+2}T^2_{i,i+1}=\gamma_{i+1} \gamma_{i+2} \cdot \gamma_{i}
\gamma_{i+1}=-\gamma_i \gamma_{i+2}$.  Since it is impossible
to encode the $-1$ in the $x_i$'s of $G_m$, the representation has to
be projective.  Note that an overall phase in
$T_{ij}$ will not affect the conclusion.
In the next subsection, we will see this projective representation
comes from a linear representation of the braid group---the Jones
representation at a $4^{th}$ root of unity and the $-1$ is
encoded in the Jones-Wenzl projector $p_3=0$.

To understand the images of the Teo-Kane representation of $G_m$,
we observe that the Teo-Kane unitary transformations $T_{ij}=e^{\frac{\pi}{4}\gamma_i \gamma_j}$ lie inside
the even part $Cl_m^0(\mathbb{C})$ of $Cl_m(\mathbb{C})$.
Therefore, the Teo-Kane representation of $G_m$ is just the spinor
representation projectivized to $G_m \subset SO(m)$.  It follows
that the projective image group of Teo-Kane representation of the even ribbon permutation group is $G_m$ as an
abstract group.  Recall $Cl_m^0(\mathbb{C})$ is reducible into two irreducible sectors if $m$
is even, and irreducible if $m$ is odd.  When $m$ is even, it is important to
know how the relative action of the center $Z(G_m)\cong \mathbb{Z}_2$ of $G_m$ on the two
irreducible sectors. The center of $G_m$ is generated by the element $z=((1\cdots 1), 1)$, whose
corresponding Teo-Kane unitary transformation is $U=\gamma_1 \gamma_2 \cdots
\gamma_m$ up to an overall phase.  As we show in
Appendix \ref{sec:braid-group},
the relative phase of $z$ on the two sectors is always $-1$.

\section{Discussion}
\label{sec:discussion}

We now review and discuss the main results derived in this
paper. Using the topological classification of
free fermion Hamiltonians \cite{Ryu08,Kitaev09},
we considered a system of fermions in 3D which is allowed to
have arbitrary superconducting order parameter
and arbitrary band structure; and is also allowed to
develop any other possible symmetry-breaking
order such as charge density-wave, etc. -- i.e. we do
not require that any symmetries are preserved.
We argued in Section \ref{sec:free-fermion} that the space
of possible gapped ground states
of such a system is topologically equivalent to U($N$)/O($N$)
for $N$ large (for large $N$, the topology of U($N$)/O($N$)
becomes independent of $N$).
By extension, if we can spatially vary the superconducting
order parameter and band structure at will with no
regard to the energy cost, then there will be topologically
stable point-like defects classified by
${\pi_2}(U($N$)/O($N$))=\mathbb{Z}_2$.

This statement begs the question of whether
one actually can vary the order parameter and
band structure in order to create such defects.
In a given system, the energy cost may simply
be too high for the system to wind around
U($N$)/O($N$) in going around such a defect.
(This energy cost, which would include the condensation
energy of various order parameters, is not taken into
account in the free fermion problem.) If we
create such defects, they may be so costly
that it is energetically favorable for them
to simply unwind by closing the gap over large
regions. (The energy cost associated with such an unwinding
depends on the condensation energy of the order parameters
involved, which is not included in the topological classification.)
Thus, U($N$)/O($N$) is not the target space
of an order parameter in the usual sense because
the different points in U($N$)/O($N$) may not correspond to
different ground states with the same energy.
However, in Section \ref{sec:strong-coupling},
we have given at least one concrete model of free fermions with
no symmetries in 3D in which the topological defects
predicted by the general classification are present
and stable. Furthermore, Teo and Kane \cite{Teo10}
have proposed several devices in which
these defects are simply superconducting vortices
at the boundary of a topological insulator.

In order to understand the quantum mechanics of
these defects, it is important to first understand their
quantum statistics. To do this,
we analyzed the multi-defect configuration
space; its fundamental group governs defect
statistics. The configuration space of $2n$ point-like
defects of a system with `order parameter' taking values
in U($N$)/O($N$) is the space which we call $K_{2n}$.
It can be understood as a fibration.
The base space is $X_{2n}$, the configuration
space of $2n$ points (which we know has fundamental
group $S_{2n}$ in dimension three and greater).
Above each point in this base space there is
a fiber ${\cal M}_{2n}$ which is the space of maps from
the ball $B^3$ minu $2n$ fixed points to
U($N$)/O($N$) with winding number $1$ about each
of the $2n$ points. $K_{2n}$ is the total space of
the fibration. In Section \ref{sec:Kane_space},
we found that its fundamental group is
${\pi_1}(K_{2n})=\mathbb{Z}\times\mathbb{Z}_2
\times E({\mathbb{Z}_2^{2n} \rtimes S_{2n}})$, where
$E({\mathbb{Z}_2^{2n} \rtimes S_m})$ is
the subgroup of ${\mathbb{Z}_2^m \rtimes S_{2n}}$ comprised of
elements whose total parity in $\mathbb{Z}_2^{2n}$
added to their parity in $S_{2n}$ is even.

The fundamental group of the configuration
space is the same as the group of equivalence
classes of spacetime histories of a system
with $2n$ point-like defects. Since these different
equivalence classes cannot be continuously
deformed into each other, quantum mechanics
allows us to assign them different unitary
matrices. These different unitaries form
a representation of the fundamental group of the configuration
space of the system. However, we found in Section \ref{sec:projective}
that Teo and Kane's unitary transformations are not
a linear representation of ${\pi_1}(K_{2n})$,
but a {\it projective representation}, which is to say that
they represent ${\pi_1}(K_{2n})$ only up to a phase.
Equivalently, Teo and Kane's unitary transformations are
an ordinary linear representation of a {\it central
extension} of ${\pi_1}(K_{2n})$, as discussed in
Section \ref{sec:projective}.

This surprise lurks in a seemingly innocuous set of
defect motions: those in which defects $i$ and $j$ are rotated
by $2\pi$ and the order parameter field surrounding
them relaxes back to its initial configuration. This has the
following effect on the Majorana fermion zero mode operators
associated with the two defects:
\begin{equation}
{\gamma_i}\rightarrow -{\gamma_i}\, , \hskip 0.5 cm
{\gamma_j}\rightarrow -{\gamma_j}
\label{eqn:gamma-trans}
\end{equation}
One might initially expect that two such motions, one affecting
defects $i$ and $j$ and the other affecting $i$ and $k$,
would commute since they simply multiply the operators
involved by $-$ signs. However, the unitary operator which
generates (\ref{eqn:gamma-trans}) is\cite{footnote1}:
\begin{equation}
U^{ij} = e^{i\theta} \,{\gamma_i} {\gamma_j}\, .
\end{equation}
Thus, the unitary operators $U^{ij}$ and $U^{ik}$ do not commute;
they anti-commute:
\begin{equation}
U^{ij} \, U^{ik}= -U^{ik}\,U^{ij} \,  .
\end{equation}
However, as shown in Figure \ref{fig:movie-of-movies}, the
corresponding classical motions can be continuously deformed
into each other. Thus, a linear representation of the fundamental group
of the classical configuration space would have these two
operators commuting. Instead, the quantum mechanics of this
system involves a projective representation.

Projective quasi-particle statistics were first proposed by
Wilzcek\cite{Wilczek98}, who suggested a projective representation of the
permutation group in which generators $\sigma_j$ and $\sigma_k$ anti-commute
for $|j-k|\geq 2$, rather than commuting.  Read\cite{Read03} criticized
this suggestion as being in conflict with locality.  We can sharpen Read's
criticism as follows.
Suppose that one can perform the operation $\sigma_1$ by acting
on a region of space, called $A$, containing particles $1$ and $2$,
and one can perform $\sigma_3$ by acting on a region called $B$, containing
particles $3$ and $4$, and suppose that regions $A$ and $B$ are disjoint.
Consider the following thought experiment: let Bob perform operation
$\sigma_3$ at time $t=0$ and let Bob then repeat operation $\sigma_3$ at
time $t=1$.  Let Alice prepare a spin in the state $(1/\sqrt{2})(|\uparrow\rangle+|\downarrow\rangle)$ at time $t=-1$, and then let Alice perform the following
sequence of operations.
At time $t=-\epsilon$, for some small $\epsilon$, she
performs the unitary operation $|\uparrow\rangle\langle \uparrow| \otimes \sigma_1+|\downarrow\rangle\langle\downarrow | \otimes I$, where $I$ is the
identity operation, leaving the particles alone.
At time $t=+\epsilon$, she
performs the unitary operation $|\uparrow\rangle\langle \uparrow| \otimes I
+|\downarrow\rangle\langle\downarrow | \otimes \sigma_1$.   Thus, if
the spin is up, she performs $\sigma_1$ at time $t=-\epsilon$, while
if the spin is down, she does it at time $t=+\epsilon$.
Finally, at
time $t=2$, Alice performs the operation $\sigma_1$ again.  One may then
show that, due to the anti-commutation of $\sigma_1$ and $\sigma_3$, that
the spin ends in the state
$(1/\sqrt{2})(|\uparrow\rangle-|\downarrow\rangle)$.  However, if Bob had
not performed any operations, the spin would have ended in the original state
$(1/\sqrt{2})(|\uparrow\rangle+|\downarrow\rangle)$.  Thus, by performing these
interchange operations, Bob succeeds in transmitting information to a space-like
separated region (if $A$ and $B$ are disjoint, and the time scale in the
above thought experiment is sufficiently fast, then Alice and Bob are space-like
separated throughout).

Having seen this criticism, we can also see how Teo and Kane's construction
evades it.  The fundamental objects for Teo and Kane are not particles, but
particles with ribbons attached.  One may verify that, in every case where
operations in Teo and Kane's construction anti-commute, the two operations
do not act on spatially disjoint regions due to the attached ribbons.  That is,
the interchange of particles also requires a rearrangement of the order parameter field.

While this argument explains why a projective representation
does not violate causality, it does not really explain why
a projective representation actually occurs in this system.
Perhaps one clue is the fact that the hedgehogs have long-ranged
interactions in any concrete model. Even in the `best-case scenario',
in which the underlying Hamiltonian of the system is U($N$)/O($N$)-invariant, there will be a linearly-diverging
gradient energy for an isolated hedgehog configuration.
Thus, there will be a linear long-ranged force between hedgehogs.
Consequently, one might adopt the point of view that, as a result
of these long-ranged interactions, the overall phase associated with
a motion of the hedgehogs is not a purely topological quantity
(but will, instead depend on details of the motion) and, therefore,
need not faithfully represent the underlying fundamental group.
As one motion is continuously deformed into another
in Figure \ref{fig:movie-of-movies}, the phase of the wavefunction
varies continuously from $+1$ to $-1$ as the order parameter
evolution is deformed. It is helpful to compare this to another
example of a projective representation: a charged particle in
a magnetic field $B$. Although the system is invariant under the
Abelian group of translations, the quantum mechanics of the
system is governed by a non-Abelian projective representation
of this group (which may be viewed as a linear representation of
the `magnetic translation group'). A translation by $a$ in the $x$-direction,
followed by a translation by $b$ in the $y$-direction differs in its
action on the wavefunction from the same translations in
the opposite order by a phase $abB/\Phi_0$ equal to the magnetic flux through
the area $ab$ in units of the flux quantum $\Phi_0$.
If we continuously deform these two sequences of translations
into each other, the phase of the wavefunction varies continuously.
For any trajectory along this one-parameter family of trajectories
(or `movie of movies'), the resulting phase of the wavefunction
is given by the magnetic flux enclosed by the composition of
this trajectory and the inverse of the initial one. In our
model of non-Abelian projective statistics, the phase changes
continuously in the same way, but as a result of the evolution
of the order parameter away from the defects, rather than as a result of a magnetic field.

An obvious question presents itself: is there a related theory in which the hedgehogs are no longer confined? Equivalently, could 3D objects with non-Abelian ribbon permutation statistics ever be the weakly-coupled low-energy quasiparticles of a system? The most straightforward route will not work: if we had a U($N$)-invariant system and tried to gauge it to eliminate the linear confining force between hedgehogs, we would find that the theory is sick due to the chiral anomaly. If we doubled the number of fermions in order to eliminate the anomaly, there would be two Majorana modes in the core of each
hedgehog, and their energies could be split away from zero by a local interaction. This is not surprising since ribbon permutation statistics would violate locality if the hedgehogs were truly decoupled (or had exponentially-decaying interactions). On the other hand, if the hedgehogs were to interact through a Coulomb interaction (or, perhaps, some other power-law),
they would be neither decoupled nor confined, thereby satisfying the requirements that they satisfy locality and are low-energy particle-like excitations of the system. Elsewhere, we will describe a model which realizes this scenario \cite{Freedman10}

The non-Abelian projective statistics studied in the 3D class with no
symmetry can be generalized to arbitrary dimension.
As shown in Eq. (\ref{eq:defectclassification}), the classification of topological defects is independent of the spatial dimension.
Thus, in any dimension $d$ with no symmetry ($p=0$), point-like ($D=0$) topological defects are classified by
$\pi_0(R_{p-D+1})=\pi_0(R_1)=\mathbb{Z}_2$.
Moreover, it can be proved that analogous topological defects in different dimensions not only carry the same topological quantum number, but also have the same statistics. In Sec.
\ref{sec:Kane_space}, we have defined the configuration space
$\mathcal{M}_{2n}$ which is the space of maps from
$B^3\setminus\bigcup_{i=1}^{2n} B_i^3$ to $R_7=U/O$,
with specific boundary conditions. Now if we consider point defect
in the class with no symmetry in 4D, the configuration space
$\mathcal{M}_{2n}^{d=4}$ is defined by maps from
$B^4\setminus\bigcup_{i=1}^{2n} B_i^4$ to the classifying space
$R_6=Sp/U$. Noticing that $B^4\setminus\bigcup_{i=1}^{2n} B_i^4$
is homotopy equivalent to the suspension of
$B^3\setminus\bigcup_{i=1}^{2n} B_i^3$, we obtain that
$\mathcal{M}_{2n}^{d=4}$ is equivalent to the space of maps from
$B^3\setminus\bigcup_{i=1}^{2n} B_i^3\rightarrow \Omega\left(Sp/U\right)$, where $\Omega\left(Sp/U\right)$ is the loop space of
$Sp/U$. Since  $\Omega\left(Sp/U\right)\simeq U/O$,
we obtain $\mathcal{M}_{2n}^{d=4}\simeq \mathcal{M}_{2n}$. Thus we have proved that the configuration space $\mathcal{M}_{2n}$ is independent of the spatial dimension $d$. On the other hand, the fundamental
group of the configuration space $X_{2n}$ of $2n$ distinct points in $B^d$ is independent of $d$ as long as $d>2$. Consequently, the space $K_{2n}$ defined by the fibration $\mathcal{M}_{2n}\rightarrow K_{2n}\rightarrow X_{2n}$ is also topologically independent of spatial dimension $d$ for $d>2$. Thus, the proof we did for $d=3$ applies to generic dimension, and non-Abelian projective statistics exist in any spatial dimension $d>3$ for point defects
in the no-symmetry class. A similar analysis applies to extended
defects with dimension $D>0$. When the spatial dimension is increased by $1$ and the symmetry class remains the same, the classifying space is always changed from $R_{2+p-d}$ to $R_{2+p-d-1}\simeq \Omega^{-1}(R_{2+p-d})$. Consequently, at least for
simple defects with the topology of $S^D$, the statistics is independent of spatial dimension $d$ as long as $d$ is large enough. For point defects, the ``lower critical dimension" is $d=3$, while for line defects, i.e. $D=1$, the ``lower critical dimension" is at least $d=4$ since in $d=3$ we can have braiding between loops.

%%%%%%%%%%%%%%%%%%%%%%%%%%%%%%%%%%%%%
%%%%%%%%%%%%%%%%%%%%%%%%%%%%%%%%%%%%%%%%%%%%%%%%%%%%%%%%%%%%%%%%%%%%%%%%%%%
%%%%%%%%%%%%%%%%%%%%%%%%%%%%%%%%%%%%%

\appendix

\section{Topological Classification of Hamiltonians
with Charge Conservation but without Time-Reversal
Symmetry}

\label{sec:QnotT}

For the sake of completeness, in this appendix we
give the topological classification of free fermion Hamiltonians
with charge conservation but not time-reversal symmetry.
We begin with the zero-dimensional case. The charge
conservation condition is most transparent when the Hamiltonian
is written in terms of complex fermion operators, $c^{}_i$, $c_i^\dagger$,
rather than Majorana fermions:
\begin{equation}
\label{eqn:unitary-case-Ham}
H = \sum_{i,j} h_{ij} {c_i^\dagger} c^{}_j
\end{equation}
The indices $i,j$ range from $1$ to $N$, the number of
bands. So long as there is one $c$ and one $c^\dagger$ in each term,
the Hamiltonian will conserve charge. We make no further
assumption about the Hamiltonian. As in our discussion
in Section \ref{sec:free-fermion}, we are only interested in
the topology of the space of such Hamiltonians, so
two Hamiltonians are considered to be equivalent if they
can be continuously deformed into each other without closing
the gap. Thus, we can flatten the spectrum and assume that
$H$ and, therefore, $h_{ij}$ only has eigenvalues $\pm 1$;
any other gapped Hamiltonian can be continuously deformed
so that it satisfies this condition. Then, $h_{ij}$ can be written in
the form:
\begin{equation}
h = U^\dagger  \left( \begin{array}{cccccc}
1 &  &   &  & & \\
 & \ddots  &    &  &  &\\
   &    & 1  & &  &\\
   &    &   &  -1&  &\\
   &    &     &    &\ddots &\\
      &    &     &    &  & -1  \end{array} \right) U
\end{equation}
with $k$ diagonal entries equal to $+1$ and $N-k$
diagonal entries equal to $-1$ for some $k$. The space
of such matrices $h_{ij}$ is equal to the space of
matrices $U$, namely U($N$), modulo those matrices $U$
which commute with the diagonal matrix above, namely
U$(k)\times$U$(N-k)$. Thus, the space of matrices $h_{ij}$
is topologically equivalent to
\begin{equation}
{\cal M}_{2N} = \bigcup_{k=0}^{N} \text{U}(N)/(\text{U}(k)\times
\text{U}(N-k))
\end{equation}
For $N$ sufficiently large, the topology of these spaces
cannot depend on $k$, so we write this simply as
$\mathbb{Z}\times\text{U}(N)/(\text{U}(N/2)\times\text{U}(N/2))$.

Now suppose that we have a 1D system. As in section
\ref{sec:free-fermion}, we will consider the Dirac equation
which approximates the Hamiltonian in the vicinity of the
points in the Brillouin zone at which it becomes small.
We write:
\begin{equation}
H = \psi^\dagger(i{\gamma_1}{\partial_1} + M)\psi
\end{equation}
Here, we have suppressed the band index on
the fermion operators $\psi^\dagger$ and $\psi$.
Both $\gamma$ and $M$ are $N\times N$ matrices;
$\gamma_1$ satisfies  $\text{tr}({\gamma_1})=0$ and $\gamma_1^2 = 1$.
Then, the different fermion operators will have the same gap if
\begin{equation}
\{{\gamma_1},M\}=0\,, \hskip 0.5 cm M^2 = 1
\end{equation}
Viewed as a linear operator on $\mathbb{C}^N$,
$\gamma_1$ defines two $N/2$-dimensional subspaces,
its eigenvalue $+1$ eigenspace and its eigenvalue $-1$ eigenspace.
Let us call them $X_+$ and $X_-$. Since $\{{\gamma_1},M\}=0$,
$M$ is a map from $X_+$ to $X_-$. Thus, $M$ is an isometry
between two copies of $\mathbb{C}^{N/2}$ and can be viewed
as an element of $\text{U}(N/2)$. This can be made more
concrete by taking, without loss of generality.
\begin{equation}
{\gamma_1} = \left( \begin{array}{cc}
\mathbb{I} & 0 \\
0 & -\mathbb{I}  \end{array} \right)
\end{equation}
where $\mathbb{I}$ is the $N/2 \times N/2$ identity matrix.
Then, $M$ can be any matrix of the form
\begin{equation}
M = \left( \begin{array}{cc}
0 & U^\dagger \\
U & 0  \end{array} \right)
\end{equation}
for $U\in \text{U}(N/2)$.

Now, consider a 2D system. The Dirac equation
takes the form:
\begin{equation}
H = \psi^\dagger(i{\gamma_j}{\partial_j} + M)\psi
\end{equation}
with $j=1,2$. As before, $\gamma_1$ defines two
$N/2$-dimensional subspaces, $X_+$ and $X_-$.
Then $i{\gamma_2}M$ commutes with $\gamma_1$
and squares to $1$. Thus, it divides $X_+$ into
two spaces, $X_+^1$ and $X_-^2$ which are the $+1$ and $-1$
eigenspaces for $i{\gamma_2}M$ (and it similarly divides
for $X_-$). Thus, a choice of $M$ is a choice of
linear subspace $X_+^1$ of $X_+$ or, in other words,
\begin{eqnarray}
M &\in& \bigcup_{k=0}^{N} \text{U}(N/2)/(\text{U}(k)\times
\text{U}(N/2-k))\cr
&=& \mathbb{Z}\times\text{U}(N/2)/(\text{U}(N/4)\times
\text{U}(N/4))
\end{eqnarray}
Continuing in this way, we derive Table \ref{tbl:unitary-classifying}
analogous to Table \ref{tbl:classifying}
but for charge-conserving systems
without time-reversal symmetry. Table \ref{tbl:unitary-classifying}
has period-$2$ as the dimension is increased
whereas Table \ref{tbl:classifying} had period-$8$.

\begin{table*}
\begin{tabular}{c | c c c c c}
dim.: &  0    &   1 & 2 & 3 &\ldots\\
\hline\hline
no symm.& $\mathbb{Z}\times
 \frac{\text{U}(N)}{\text{U}(N/2)\times\text{U}(N/2)}$ &  $\text{U}(N/2)$&
 $\mathbb{Z}\times \frac{\text{U}(N/2)}{\text{U}(N/4)\times\text{U}(N/4)}$
 & $\text{U}(N/2)$ & \ldots\\
sublattice symm. & $\text{U}(N/2)$  &
$\mathbb{Z}\times \frac{\text{U}(N/2)}{\text{U}(N/4)\times\text{U}(N/4)}$
&$\text{U}(N/4)$
&$\mathbb{Z}\times \frac{\text{U}(N/4)}{\text{U}(N/8)\times\text{U}(N/8)}$
&\ldots
\end{tabular}
\caption{The period-$2$ (in both dimension and number of
symmetries) table of classifying spaces for
charge-conserving free fermion Hamiltonians
without time-reversal symmetry. The systems
have $N$ complex fermon fields in dimensions
$d=0,1,2,3,\ldots$ with no additional symmetries (apart from
charge-conservation) or with an additional symmetry, such
as a sublattice symmetry.
Moving $p$ steps to the right and $p$ steps down leads to
the same classifying space (but for $1/2^p$ as many fermion
fields), which is a reflection of Bott periodicity, as explained
in the text. The number of disconnected components of any such
classifying space -- i.e. the number of different phases in that
symmetry class and dimension -- is given by the corresponding
$\pi_0$, which may be found in Eq. \ref{eqn:unitary-stable-pi-0}.
Higher homotopy groups, which classify defects, can be computed
using Eq. \ref{eqn:unitary-periodicity}.}
\label{tbl:unitary-classifying}
\end{table*}

\begin{table*}
 \begin{tabular}[t]{|c|c|c|c|c|}
 \hline
Symmetry classes&Physical realizations&$d=1$&$d=2$&$d=3$
 \\\hline\hline
A&Generic ins.&0&Quantum Hall (GaAs, etc.)&0
\\\hline
AIII&Bipartite ins. &{ Carbon nanotube}&0&{\color{red} Z}
\\\hline
 \end{tabular}
  \caption{Topological periodic table in physical dimensions $1,2,3$ of the complex classes. A and AIII in the first column labels the symmetry classes in the Zirnbauer classification.\cite{Zirnbauer96,Altland97}. The second column are the requirements to the physical systems which can realize the corresponding symmetry classes. The three columns $d=1,2,3$ list the topological states in the spatial dimensions $1,2,3$ respectively. $0$ means the topological classification is trivial. The red label {\color{red} Z} stands for the topological state which have not been realized in realistic materials. }
  \label{tbl:periodiccomplex}
\end{table*}

The two rows in Table \ref{tbl:unitary-classifying}
correspond to the presence or absence of
an additional unitary symmetry which requires the eigenvalues
of the Hamiltonian to come in pairs with equal and
opposite energy $\pm E_i$, i.e. a sublattice symmetry,
as discussed in Sec. \ref{sec:bott}. 
Consider a zero-dimensional system described by
the Hamiltonian (\ref{eqn:unitary-case-Ham})
with such a symmetry.
Then there is a unitary matrix $U$ satisfying
\begin{equation}
\{ U, h\}=0 \hskip 1 cm {U^2} = 1
\end{equation}
We can now apply the same logic to $U$ and $M$
that we applied to $\gamma_1$ and $M$ when considering
one-dimensional systems. Thus, $M\in \text{U}(N/2)$.
More generally, it is clear that $U$ plays the role of an extra
$\gamma$-matrix, so the row of Table \ref{tbl:unitary-classifying}
for systems with sublattice symmetry is just shifted by
one from the row of the table for systems without it.

In order to compute the homotopy groups of these
classifying spaces, it is useful to follow the logic which
we employed in Sections \ref{sec:free-fermion}D,F.
We approximate the loop space of $\text{U}(N)$
by minimal geodesics from $\mathbb{I}$
to $-\mathbb{I}$: $L(\lambda) = e^{i\lambda A}$,
where $A$ is a Hermitian matrix satisfying
${A^2}=1$ so that $L(0)=\mathbb{I}$
and $L(\pi)=-\mathbb{I}$. Any such matrix $A$
divides $\mathbb{C}^N$ into $+1$ and $-1$ eigenspaces
or, in other words, it can be written in the form
\begin{equation}
A = U^\dagger  \left( \begin{array}{cc}
\mathbb{I}^{}_k & 0 \\
0 & -\mathbb{I}^{}_{N-k}  \end{array} \right)
 U
\end{equation}
where $\mathbb{I}_m$ is the $m\times m$ identity matrix.
Therefore $A\in\mathbb{Z}\times
 \frac{\text{U}(N)}{\text{U}(N/2)\times\text{U}(N/2)}$.
Thus, the loop space of $\text{U}(N)$ can be approximated
by $\mathbb{Z}\times \frac{\text{U}(N)}{\text{U}(N/2)\times\text{U}(N/2)}$.
Thus, ${\pi_k}(\text{U}(N)) = \pi_{k-1}(\mathbb{Z}\times
\text{U}(N)/(\text{U}(N/2)\times\text{U}(N/2)))$.

To compute homotopy groups of $\mathbb{Z}\times
\text{U}(N)/(\text{U}(N/2)\times\text{U}(N/2))$, consider
the homotopy long exact sequence associated to the
fibration
$$\text{U}(N/2)\rightarrow\text{U}(N)\rightarrow\text{U}(N)/\text{U}(N/2).
$$
It tells us that
$$
\ldots\pi_{i-1}(\text{U}(N/2))
\rightarrow {\pi_i}(\text{U}(N))\rightarrow
{\pi_i}(\text{U}(N)/\text{U}(N/2)) \ldots.
$$
If we are in the stable limit, i.e. for $N$ sufficiently large,
then $\pi_{i-1}(\text{U}(N/2)) = {\pi_i}(\text{U}(N))$,
so ${\pi_i}(\text{U}(N)/\text{U}(N/2))=0$. Now consider the
fibration
$$
\text{U}(N/2)\rightarrow\text{U}(N)/\text{U}(N/2)\rightarrow
\text{U}(N)/(\text{U}(N/2)\times\text{U}(N/2)).
$$
Then, since ${\pi_i}(\text{U}(N)/\text{U}(N/2))=0$,
the associated homotopy exact sequence
is simply
$$
0\rightarrow {\pi_i}(\text{U}(N)/(\text{U}(N/2)\times\text{U}(N/2)))
\rightarrow \pi_{i-1}(\text{U}(N/2))\rightarrow 0
$$
so ${\pi_i}(\text{U}(N)/(\text{U}(N/2)\times\text{U}(N/2)))=\pi_{i-1}(\text{U}(N/2))$.

Combining these two relations, we have
Bott periodicity for the unitary group:
\begin{eqnarray}
\label{eqn:unitary-periodicity}
{\pi_k}(\text{U}(N)) &=& \pi_{k-1}(\mathbb{Z}\times
 \text{U}(N)/(\text{U}(N/2)\times\text{U}(N/2)))\cr
 &=& \pi_{k-2}(\text{U}(N/2))
\end{eqnarray}
By inspection, we see that
\begin{eqnarray}
\label{eqn:unitary-stable-pi-0}
{\pi_0}(\text{U}(N))&=&0\cr
\pi_{0}(\mathbb{Z}\times
 \text{U}(N)/(\text{U}(N/2)\times\text{U}(N/2)))
 &=&\mathbb{Z}
\end{eqnarray}
Combining this with Eq. \ref{eqn:unitary-periodicity}, we
learn that the even homotopy groups of $\text{U}(N)$ vanish
and the odd ones are $\mathbb{Z}$; the reverse holds
for $ \text{U}(N)/(\text{U}(N/2)\times\text{U}(N/2)))$.

These homotopy groups can be used to classify the
different possible phases and topological
defects of systems in the symmetry classes/dimensions in
Table \ref{tbl:unitary-classifying}. Similar to Table \ref{tbl:periodic} in Sec. \ref{sec:bott}, we list in Table \ref{tbl:periodiccomplex} the nontrivial topological physical systems in the two symmetry classes.
In Ref. \onlinecite{Ryu08}, these were called A and AIII, respectively,
following the corresponding classification
in random matrix theory \cite{Zirnbauer96,Altland97}.
In A class (with charge conservation and without sublattice symmetry or time-reversal symmetry) a topological nontrivial state is classified by integer in 2d, which is the famous quantum Hall system. In the AIII class with sublattice symmetry, in 1d the carbon nanotube or graphene ribbon can be considered as an example with nontrivial edge states. Compared to the requirement of this symmetry class, carbon nanotube has too high symmetry because time-reversal symmetry is present. An orbital magnetic field can be coupled minimally to the carbon nanotube to get an example system with exactly the required symmetry of AIII class. In 3d there should be an integer classification of the AIII class but no realistic topological nontrivial material is known yet.

\section{Understanding the Role of Spatial Dimension
in the Classification of Free Fermion Systems}

\label{sec:dimension}

In Sections \ref{sec:free-fermion} and \ref{sec:QnotT},
we have seen that increasing the spatial dimension
moves a system through the progression of classifying
spaces oppositely to adding symmetries which square to
$-1$ (but in the same direction as adding symmetries
which square to $1$). The classifying space for
$d$-dimensional systems is the loop space of
the classifying space for $d+1$-dimensional systems
with the same symmetries. Symbolically,
\begin{equation}
R_{2+p-d}=\Omega(R_{2+p-(d+1)})
\end{equation}
where $\Omega$ denotes the loop space
and, as in Section \ref{sec:free-fermion}, $p$
is the number of symmetries and $d$ is the spatial
dimension. Thus, in using Tables \ref{tbl:classifying}
and \ref{tbl:unitary-classifying}, moving one step to the right,
which increases the dimension and one step down, which
increases the number of symmetries, leaves the classifying space
the same.

We explained this in Sections \ref{sec:free-fermion} and
\ref{sec:QnotT} by expanding the Hamiltonian
about the point(s) in the Brillouin
zone where the gap is minimum where it has the form
of the Dirac Hamiltonian. Analyzing the
space of such Hamiltonians then amounts to analyzing the
possible mass terms, namely the space of matrices
which square to $-1$ and anticommute with the $\gamma$-matrices
and anti-unitary symmetry generators.
One might worry that this analysis is not completely general
since it depends on our ability to expand the Hamiltonian
in the form of the Dirac Hamiltonian, and this seems non-generic.
However, Kitaev's texture theorem (unpublished)
states that a general gapped free fermion $H$
canonically deforms to a Dirac Hamiltonian as above without closing the gap, so only these need to be considered.

Another more algebraic approach involves the suspension isomorphism in KR theory. In complex K-theory: $KU^{n+d}(X\wedge S^d)\simeq KU^n(X),$
however, for $KR[A]$ with $^-(x_1,\ldots,x_d) = (-x_1,\ldots,-x_d),$
one gets: $KR^{n-d}(X\wedge \bar{S}^d)\simeq KO^n(X)$. (There is a
similar isomophism in twisted K-theory where $(^-)^2 = -1$.) In
(translation invariant) particle/hole
symmetric systems
\footnote{
  \begin{eqnarray*}
    &&\frac{i}{4}\sum\gamma_{x_i}^{\dag}A_{x_i-x_j}\gamma_{x_j} =
    \frac{i}{4}\sum\gamma_{x_i}A_{x_i-x_j}\gamma_{x_j}^{\dag}\\
    &&\hspace{1cm} = -\frac{i}{4}\sum\gamma_{x_j}^{\dag}A_{x_i-x_j}\gamma_{x_i} =
    -\frac{i}{4}\sum\gamma_{x_i}^{\dag}A_{x_i-x_j}\gamma_{x_j},
  \end{eqnarray*}
  so reversing a space coordinate introduces a minus sign.}
it can be argued via a high frequency cutoff that momentum space directions correspond not to
$\bbar{\R}^d$ but actually to $\bbar{S}^d$.

If translational symmetry is broken to a lattice $\Z^d$ symmetry,
then momentum space becomes a $d-$torus $T^d$
and it should replace $S^d$ in
the l.h.s. of the two isomorphisms above.
Fortunately, all K-theories are generalized cohomology theories and
thus depend only on \underline{stable} homotopy type. Tori have
extremely simple stable types:
\begin{equation}
  \label{eq:tori_types}
  \sum T^d = \sum(\vee_d(S^1)\vee \vee_{d \choose 2}(S^2) \vee \cdots \vee
  \vee_{d\choose d-1}(S^{d-1})\vee S^d),
\end{equation}
where $\sum$ denotes suspension. The same formula holds with
$(^-)$ above
all spaces.

Thus, for lattice-translational symmetry, we may employ
\begin{multline}
  \label{eq:trans_symm}
  KR^{n-d}(X\wedge \bar{T}^d) \simeq\\
   \til{KO}^n(X) \oplus_{d\choose
    d-1} \til{KO}^{n-1}(X) \oplus_{d\choose d-2} \cdots\\
    \oplus_{d\choose
    1}\til{KO}^{n-d+1}(X) \oplus KO^{n-d}(X).
\end{multline}

In this appendix, we want to sketch here a third way
of understanding ``spatial dimension = de-loop''
that is more geometric than either of the alternatives above.

  We warn the reader that this section is only a ``sketch'',
  so there may be tricky analytic
  details regarding the precise definition of the controlled spaces
  which we have over-looked. The math literature\cite{Roe96} in controlled
  K-theory is constructed in a less naive context. We are not sure if
  there is something essential we are missing.

Let $O_{\Z}$ be a limitingly large orthogonal group with integer $(\Z)$
control and a topology (like compact/open) where sliding a
disturbance off to $\pm\infty$ converges to the identity.
That is, $O_{\Z}$ is the set of infinite matrices $\Owe_{I,J}$
with a finite-to-one map (decapitalizing) $J\mapsto j\in\Z$ such that $|\Owe_{I,J}|\leq e^{-\text{const}|i-j|}$
and $\Owe_{I,J}^{\ast}\Owe_{I,K} = \delta_{J,K}$ with the following
identifications: Any matrix
$\Owe_{I,J}$ with $\Owe_{I_0,J}\equiv\delta_{I_0,J}$ is equivalent
to the matrix with one less row and column obtained by removing the
$I_0^{\textrm{th}}$ row and $I_0^{\textrm{th}}$ column.

By changing $l$, the number of distinct
$J$ map that map to a given $j$, we can choose any desired constant in the exponential
decay $\exp(-\text{const}|i-j|)$.  We could also have chosen to study the class of infinite
matrices with $\Owe_{I,J}=0$ for $|i-j|>1$; in this case, the matrix $\Owe_{I,J}$ is a banded
matrix, which is non-zero
only within distance $l$ of the main diagonal.  We refer to the case of exponential decay
as ``soft control", while we refer to the case that $\Owe_{I,J}=0$ for $|i-j|>1$ as ``strict control".

The key claim in this ``third'' approach is
\begin{claim}
  $O_{\Z} \simeq \faktor{O}{O\times O}\times \Z = BO\times Z$ ($\simeq$
  denotes weak homotopy type.)
\end{claim}

We take for $BO\times \Z$ the well known model of pairs, considered as
formal differences, of transverse (not necessarily spanning) subspaces
$(L_+,L_-)$ of a $d$-dimensional real vector space $\R^d$
(with $d$ finite and a limit taken as $d\rightarrow\infty$.)
The $\Z$-coordinate above is the index $\dim L_+ - \dim L_-$.

Here are the maps: choose any cut in $\Z$, say $\Pi_+$ projects to $J$
with $j\geq 0$ and $\Pi_-$ projects to $J$ with $j<0$. Set
\begin{eqnarray*}
 &&  f:O_{\Z}\rightarrow BO\times\Z,\\ &&f(O) = (\ker(\Pi_+\circ
O|_{S_+}),\ker(\Pi_-\circ O|_{S_-})),
\end{eqnarray*}
where $S_{+(-)}$ is the space of basis elements mapping to $[0,L]~
([-L,-1])$ where $l\gg$ the decay constant for $O_{\Z}$, and both
kernels are regarded as subsets of the span $\R^d$ of elements mapping
to $[-L,L]$. Given $(L_+,L_-)\in BO\times \Z$, define $g(L_+,L_-)\in
O_{Z}$ as follows. Think of $L_{\pm}\subset \R^d$ with $P$ denoting
the perpendicular subspace to $\vspan(L_+,L_-)$. The infinite vector
space on which $g(L_+,L_-)$ acts is spanned by a copy of $\R^d$ for
each $i\in\Z$. $O = g(L_+,L_-)$ acts as the identity on each $P\times
i$, translation by $-1\in \Z$ on each $L_+\times i$ and translation
by $+1\in\Z$ on each $L_-\times i$.  Note: the use of the
kernel above is correct in the case of strict control.  In the case
of soft control, we should replace the kernel by the projector onto
right singular vectors of $\Pi_+ \circ O|_{S+}$ with singular value
less than or equal to some quantity which is exponentially small in $L$.

It is immediate that $f\circ g = \id_{BO\times\Z}$. On the other hand,
$g\circ f$ seems to turn a general $O\in O_{\Z}$ into a very special
form (no rotation at all, just various subspaces sliding left and
right.) However, we claim that there is a deformation retraction of
$O_{\Z}$ to maps which take this simple form: $r:O_{\Z}\rightarrow
\{\textrm{slides}\}\subseteq O_{\Z}$. First, think about a fiberwise
$O$ (no sliding) which is $X$ over some $i\in \Z$ and $\id$ over other
$j\neq i\in\Z$. Note that
\begin{displaymath}
\begin{vmatrix}
  X & 0\\ 0 & -X
\end{vmatrix}
\sim
\begin{vmatrix}
  0 & \id\\
  \id & 0
\end{vmatrix}
\end{displaymath}
via the canonical homotopy
\begin{displaymath}
  \cos\left(\frac{\pi t}{2}\right)
  \begin{vmatrix}
    X & 0\\
    0 & -X
  \end{vmatrix}
  + \sin\left(\frac{\pi t}{2}\right)
  \begin{vmatrix}
    0 & \id\\
    \id & 0
  \end{vmatrix}.
\end{displaymath}
When $d$ is even (which we may assume) there are further canonical
homotopies so that
\begin{displaymath}
  \begin{vmatrix}
    0 & \id\\
    \id & 0
  \end{vmatrix}
  \sim
  \begin{vmatrix}
    \begin{matrix}
      \sigma_{\Z} & 0\\
      0 & \sigma_{\Z}
    \end{matrix}
    & 0\\
    0 &
    \begin{matrix}
      \sigma_{\Z} & 0\\
      0 & \sigma_{\Z}
    \end{matrix}
  \end{vmatrix}
  \sim
  \begin{vmatrix}
    \id
  \end{vmatrix}
\end{displaymath}
Now a ``Eilenberg swindle'' runs $X$ (canonically) towards $+\infty$.
\begin{center}
  \begin{tabular}{c@{\hspace{0.5pt}}c@{\hspace{0.5pt}}c}
    \id & $X$ & \id\\ \hline
    $(i-1)$ & $i$ & $(i+1)$
  \end{tabular}
  $\leadsto$
  \begin{tabular}{c@{\hspace{0.5pt}}c@{\hspace{0.5pt}}c@{\hspace{0.5pt}}c@{\hspace{0.5pt}}c@{\hspace{0.5pt}}c}
    \id & $X$ & $-X$ & $X$ & $-X$ & $\cdots$\\ \hline
    $i$ & $(i+1)$ & $(i+2)$ & $(i+3)$ & $(i+4)$ & $\cdots$
  \end{tabular}\\
  $\leadsto$
  \begin{tabular}{c@{\hspace{0.5pt}}c@{\hspace{0.5pt}}c@{\hspace{0.5pt}}c@{\hspace{0.5pt}}c}
    \id & \id & \id & \id & $\cdots$\\ \hline
    $i$ & $(i+1)$ & $(i+2)$ & $(i+3)$ & $\cdots$
  \end{tabular}
\end{center}

Reusing this trick, any fiber-wise, though now it is better to slide from the center out to both $\pm\infty$,
$X$ can be canonically connected to
$\id\in O_{\Z}$. Similarly, one can (canonically) deform the general
$O\in O_{\Z}$ to a sum of three pieces --- fixed, left-sliding,
right-sliding --- and this is the deformation retraction $r$ written
above.

We now promote the (weak) homotopy equivalence $O_{\Z}\simeq
BO\times\Z$ to an entire table where the rows are all (weak) homotopy
equivalences.
\begin{displaymath}
 {\scriptstyle \begin{array}{cccccc}
     \faktor{O_{\Z^3}}{U_{\Z^3}} & O_{\Z+\Z} & \simeq_2 &
   {\hskip -0.2cm}
    \faktor{O_{\Z}}{O_{\Z}\times O_{\Z}}\times \Z & & \faktor{U}{O}\\
    & \cdots & \faktor{O_{\Z+\Z}}{U_{\Z+\Z}} & O_\Z & \simeq_1 &
      \!BO\times \Z\\
      &&\cdots & \faktor{O_{\Z}}{U_{\Z}} & \simeq_3 & O\\
      &&&&\cdots & \faktor{O}{U}\\
      &&&&&\cdots
    \end{array}}
\end{displaymath}

Moving diagonally up and left, say from $\simeq_1$ to $\simeq_2$
seems to be easy: one cuts the control space not at a point but along
a codimension 1 hyperplane, in this case the line $y = -\frac{1}{2}$,
and finds a pair of locally finite $\Z$-controlled kernels which
define an element in $\faktor{O_\Z}{O_\Z\times O_\Z}\times \Z$.
However, there is one subtle point: the matrices $\Pi_+$ and $O$ are
both controlled; however, this does not imply that the projector onto the
kernel of $\Pi_+\circ O|_{S_+}$ is also controlled.  If we knew that the
singular values of $\Pi_+ \circ O|_{S_+}$ had a spectral gap separating the
zero singular values (or, in the case of soft control, the singular values which
are exponentially small in $L$) from the rest of the spectrum, then we could
show that the projector onto the kernel had soft control, with decay
constant set by the gap, but in absence of knowledge of a gap, this
seems to be a difficult technical step.

Moving down the array, say from $\simeq_1$ to $\simeq_3$, requires
``looping'' $\simeq_1$. We must check that
$\Omega(O_\Z)\simeq\faktor{O_\Z}{U_\Z}$, i.e. that the first step in
Bott's ladder of eight rungs holds with $\Z$-controls. To do this, one
must do Lie theory in $O_\Z$: find the shortest geodesic arcs from
$\id$ to $-\id$ in $O_{\Z}$, understand which directions they pick out
in the Lie algebra $O_\Z$ and that their midpoints are $\Z$-controlled
complex structures, and finally make the analogs of Bott's index
calculations. The initial step is to consider a representative for a
``$\Z$-controlled complex structure'' (where coordinates are
consecutive in $\Z$):
\begin{displaymath}
  J = \Owe^+
  \begin{vmatrix}
    \ddots & & & \\
    & \begin{matrix}
      0 & 1\\
      -1 & 0
    \end{matrix} & & \\
    & & \begin{matrix}
      0 & 1\\
      -1 & 0
    \end{matrix} & \\
    & & & \ddots
  \end{vmatrix}
  \Owe,\hspace{1cm}\Owe\in O_\Z
\end{displaymath}
and see that the 1-parameter subgroup $\{e^{t\log \Owe}\}\subset
O_\Z$ retains $\Z$-control. This may be checked by writing
\begin{displaymath}
  \log J = \frac{\pi}{2} J,
\end{displaymath}
since $J^2=-I$.
This yields $\Z$-controls on $\log J$ and hence on the entire
1-parameter subgroup.

One interesting feature of this approach is that we are able to classify different symmetry
classes of controlled unitaries on a line.  The
controlled unitaries $U_\Z$ are classified by an integer.  This integer is precisely the ``flow"
described by Kitaev\cite{Kitaev06a}.  If the controlled unitary is assumed to be symmetric, it is in the orthogonal
group, and again $O_\Z$ is classified by an integer.  However, if the unitary
is chosen to be self-dual, in this case $\faktor{U_\Z}{Sp_\Z}$ has a $\Z_2$ classification.
Ryu et. al.\cite{Ryu08} note (see Table I or Table IV of that paper) that each of the 10 different symmetric spaces can be
obtained by considering the exponential of Hamiltonians lying in the 10 different symmetry
classes; for example, given a Hermitian matrix $H$, the exponential $\exp(i H t)$ is a unitary
matrix, while if $H$ is anti-symmetric and Hermitian, the exponential $\exp(i H t)$ is an
orthogonal matrix.  However, the  classification of unitaries here implies that
there can be an obstruction to writing a
controlled unitary as an exponential $\exp(i H t)$ with
$H$ of the correct symmetry class and $H$ also controlled.

%%%%%%%%%%%%%%%%%%%%%%%%%%%%%%%%%%%%%
\section{Phase Symmetry Unbroken}
\label{sec:appendix_phase_symmetry_unbroken}
%%%%%%%%%%%%%%%%%%%%%%%%%%%%%%%%%%%%%

If the phase symmetry of the order parameter is unbroken
(i.e. if we allow gradients in the overall phase), then
the effective sigma model target is $U/O$, not $\UO$.
This means that at every stage of the calculation,
an additional circle $\frac{U(1)}{O(1)}$ or $K(\Z,1)$ product factor must be added to the target space.
This leads to {\em no} change for $\pi_2(Y)$ or $\pi_1(R)$ (essentially because $\pi_2(\frac{U(1)}{O(1)}) \cong 0$),
however $\pi_1(Y)$ and $\pi_0(R)$ pick up an additional integer
which is the winding number around this circle.
Thus, (\ref{eq:pi_1(K_2n)_appendix}) below

\begin{equation}\label{eq:pi_1(K_2n)_appendix}
\begindc{0}[3]
    \obj(10,20)[a]{$\pi_2(Y)$}
    \obj(28,20)[b]{$\pi_1(R)$}
    \obj(44,20)[c]{$\pi_1(\til{K}_{2n})$}
    \obj(60,20)[d]{$\pi_1(Y)$}
    \obj(79,20)[e]{$\pi_0(R)$}
    \obj(60,15)[c1]{\begin{sideways}$\cong$\end{sideways}}
    \obj(79,15)[c1]{\begin{sideways}$\cong$\end{sideways}}
    \obj(60,10)[d2]{$\Z_2^{2n} \rtimes S_{2n}$}
    \obj(79,10)[e2]{$\Z_2$}
    \mor{a}{b}{$\partial_2$}
    \mor{b}{c}{}
    \mor{c}{d}{}
    \mor{d}{e}{$\partial_1$}
 %   \mor(75,20)(87,20){}
\enddc
\end{equation}
is replaced by (\ref{eq:new_pi_1(K_2n)_appendix}):
\begin{equation}\label{eq:new_pi_1(K_2n)_appendix}
\begindc{0}[3]
    \obj(10,20)[a]{$\pi_2(Y^{\text{new}})$}
    \obj(30,20)[b]{$\pi_1(R^{\text{new}})$}
    \obj(50,20)[c]{$\pi_1(K^{\text{new}}_{2n})$}
    \obj(20,10)[d]{$\pi_1(Y^{\text{new}})$}
    \obj(40,10)[e]{$\pi_0(R^{\text{new}})$}
    \obj(17,5)[c1]{\begin{sideways}$\cong$\end{sideways}}
    \obj(42,5)[c1]{\begin{sideways}$\cong$\end{sideways}}
    \obj(15,0)[d2]{$(\Z^{2n} \times \Z_2^{2n}) \rtimes S_{2n}$}
    \obj(45,0)[e2]{$\Z^{2n} \times \Z_2$}
    \obj(30,7){$\partial_1^{\text{new}}$}
    \mor{a}{b}{$\partial_2$}
    \mor{b}{c}{}
    \mor{c}{d}{}
    \mor{d}{e}{}
\enddc
\end{equation}

\noindent where $S_{2n}$ acts on $\Z^{2n}$ by permuting its standard generating set.

The new $\Z^{2n}$ factor in $\pi_0(R^\text{new})$ represents rolling the order parameter around the phase circle as one moves radially into the hedgehog.  Similarly, the $\Z^{2n}$ in $\pi_1(Y^\text{new})$ comes from wrapping the whole hedgehog around the phase circle as the parameter is advanced.  Clearly, the two $\Z^{2n}$'s cancel: $\text{ker}(\partial_1^{\text{new}}) = \text{ker}(\partial_1)$.  Thus, $$\pi_1(K_{2n}) \cong \pi_1(\til{K}_{2n})\times \Z,$$ where the extra factor $\Z$ comes from a base point in $B^3$ wraps around the additional circle $\frac{U(1)}{O(1)}$.

%%%%%%%%%%%%%%%%%%%%%%%%%%%%%%%%%%%%%
\section{Hopf Map: an example of boundary maps
in the homotopy exact sequnce}
\label{sec:hopf-map}
%%%%%%%%%%%%%%%%%%%%%%%%%%%%%%%%%%%%%

The Hopf fibration is an example of a
fibration (of a fiber bundle, in fact) which
arises in several contexts in physics.
It consists of a map (the Hopf map) from $S^3$ to $S^2$
such that the pre-image of each point in $S^2$
is an $S^1$. This is familiar from the representation
of a vector (e.g. the Bloch sphere) by a spinor
(e.g. a two-state system):
\begin{equation}
\label{eqn:spin-vector-map}
{\bf n} = {z^\dagger} {\bf \sigma} z
\end{equation}
where ${\bf \sigma}=({\sigma_x},{\sigma_y},{\sigma_z})$
are Pauli matrices and the spinor
\begin{equation}
z = \left(\begin{matrix}
             z_1  \cr
             z_2
             \end{matrix}\right)
\end{equation}
satisfies $|{z_1}|^2 + |{z_2}|^2 = 1$ so that
${\bf n}$ is a unit vector. Due to this normalization
condition, the spinor $z$ lives in $S^3$ while
the vector ${\bf n}$ lives in $S^2$.
Transforming $z\rightarrow e^{i\theta}\,z$ leaves
${\bf n}$ invariant. Thus, the map (\ref{eqn:spin-vector-map})
maps circles in $S^3$ to points in $S^2$, as depicted
below:
\[\begindc{0}[3]
    \obj(-3,20)[aa]{$U(1)\text{ phase} =$}
    \obj(10,20)[a]{$S^1$}
    \obj(20,20)[b]{$S^3$}
    \obj(20,10)[c]{$S^2$}
    \obj(42,20)[d]{$=\text{normalized states in }\C^2$}
    \obj(34,10)[e]{$=\text{Bloch sphere}$}
    \mor{a}{b}{}
    \mor{b}{c}{}
\enddc\]

One context, quite closely-related to the present
discussion, in which the Hopf map arises
is the O(3) non-linear $\sigma$ model
in $2+1$-dimensions. The order parameter field
${\bf n}(x)$ is a map from $S^3$ to $S^2$
(assuming that ${\bf n}(x)$ is constant at infinity, so that spacetime
can be compactified into $S^3$). The topological term
which can be added to such a model \cite{Wilczek83} is:
\begin{equation}
S_{H} = -\frac{1}{2\pi}\int {d^2}x\, d\tau
\epsilon^{\mu\nu\lambda} {A_\mu} F_{\nu\lambda}
\end{equation}
where $F_{\nu\lambda}=
{\partial_\nu}{A_\lambda}-{\partial_\lambda}{A_\nu}$
and
\begin{equation}
\epsilon^{\mu\nu\lambda}{\partial_\nu}{A_\lambda}
\equiv \frac{1}{8\pi} \epsilon^{\mu\nu\lambda}
{\bf n}\cdot {\partial_\nu}{\bf n} \times {\partial_\lambda}{\bf n}
\end{equation}
This term gives the linking number of skyrmion
trajectories in the O(3) non-linear $\sigma$ model.
To see this, suppose that ${\bf n}(\infty)=(0,0,1)$.
The pre-image of $(0,0,-1)$ (which we may view as the
centers of skyrmions) is a set of closed circles,
and the Hopf term gives their linking number.

We can use the Hopf map to explain boundary
maps in the homotopy exact sequence:
$\pi_2(S^2) \overset{\partial}{\underset{\cong}{\longrightarrow}} \pi_1(S^1)$:

\[\begindc{0}[3]
    \obj(10,30)[a]{$\pi_2(S^2)$}
    \obj(30,30)[b]{$\pi_1(S^1)$}
    \obj(50,30)[c]{$\pi_1(S^3)$}
    \obj(10,25)[d]{\begin{sideways}$\cong$\end{sideways}}
    \obj(30,25)[e]{\begin{sideways}$\cong$\end{sideways}}
    \obj(50,25)[f]{\begin{sideways}$\cong$\end{sideways}}
    \obj(10,20)[g]{$\Z$}
    \obj(30,20)[g]{$\Z$}
    \obj(50,20)[g]{$0$}
    \mor{a}{b}{$\partial$}
    \mor{a}{b}{$\cong$}[\atright,\solidarrow]
    \mor{b}{c}{}
\enddc\]

\noindent which, we see, must be an isomorphism to be consistent with $\pi_1(S^3) \cong 0$.  But how is this boundary map defined?  We think of a $2$-sphere $S^2$ as a loop of circles $S^1$, as depicted in Fig. \ref{fig:sphereloop}.

\begin{figure}[htpb]
\labellist \small\hair 2pt
  \pinlabel $\text{base }S^2=\text{Bloch sphere}$ at 370 100
\endlabellist
\hskip -0.8 cm \includegraphics[scale=0.5]{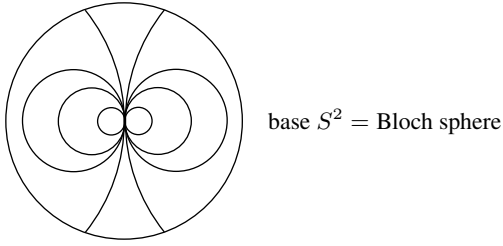}
\caption{The circle $S^1$ starts small, gets large, and becomes small again at the end of the ``loop,'' that is, after it has swung all the way around the sphere $S^2$.}
\label{fig:sphereloop}
\end{figure}

\noindent Using the (defining) property of a fibration (see any book on algebraic topology), we may lift this loop of circles $S^1$
from the base $S^2$ to the total space $S^3$.
The lifting is perfectly continuous, but
it does not give us a sphere. Instead, it gives us
a surface with boundary. This boundary is
an $S^1$. What happens is that, at the very end of
the loop of circles $S^1$, the lifting
``explodes'' the final circle $S^1$,
whose map to the Bloch sphere is shrinking to a point,
into a circle in $S^3$,
as pictured in figure \ref{fig:hopfcircles}.
In other words, an element of ${\pi_2}({S^2})$
is lifted to a surface in $S^3$ whose boundary is a circle.
This circle can wind around the fiber $S^1$,
thus defining an element of ${\pi_1}({S^1})$.
This map from an element of ${\pi_2}({S^2})$
to an element of ${\pi_1}({S^1})$ is called
a {\it boundary map} since the element
of ${\pi_1}({S^1})$ is the boundary
of the lift to $S^3$ of the corresponding element of ${\pi_2}({S^2})$.

\begin{figure}[htpb]
\labellist \small\hair 2pt
  \pinlabel $\text{another fiber}$ at 220 182
  \pinlabel $\text{fiber over base pt.}=$ at 320 155
  \pinlabel $\text{lift at end of family}$ at 320 140
  \pinlabel $\text{base pt.}=$ at 315 90
  \pinlabel $\text{lift at beginning}$ at 320 75
  \pinlabel $\text{of family}$ at 320 60
  \pinlabel $\text{family}$ at 128 115
  \pinlabel $S^3$ at 30 105
\endlabellist
\centering
\hskip -1 cm \includegraphics[scale=0.7]{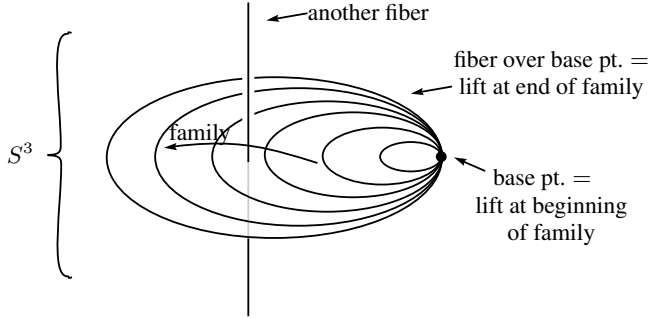}
\caption{The lifts of circles $S^1$
are depicted `upstairs' with respect to the usual picture of
Hopf circles filling $\R^3$. The lift at the end of the
family is the biggest ellipse in this picture, not a single
point, unlike the lift at the beginning of the family.
Thus, the loop of circles no longer closes when lifted.}
\label{fig:hopfcircles}
\end{figure}

%%%%%%%%%%%%%%%%%%%%%%%%%%%%%%%%%%%%%
\section{Computation of ${\pi_0}(R_{2n}), {\pi_1}(R_{2n})$
using the ``Postnikov tower'' for $\UO$}
\label{sec:Postnikov}
%%%%%%%%%%%%%%%%%%%%%%%%%%%%%%%%%%%%%

In this appendix, we compute the homotopy groups
${\pi_0}(R_{2n}), {\pi_1}(R_{2n})$.

Recall that the homotopy classes of maps from a 3-manifold $M^3$ to $S^2$ can be understood in terms of the ribbons which are the (framed) inverse image of the north pole $N \subset S^2$.  Such ribbons serve as generators for the path components of $\text{Maps}(M^3 \to S^2)$ and the relations are ``framed cobordisms,'' that is, imbedded surfaces $(S; r_0, r_1) \subset (M^3 \times [0,1]; M \times 0, M \times 1)$ where $S$ is a surface ``cobording'' between initial and final ribbons $r_0$ and $r_1$.  $S$, like the ribbons themselves, has a normal framing which restricts to the oriented ribbon directions of both $r_0$ and $r_1$.  This is a special case of what, in topology, is called the
Pontryagin-Thom construction (PTC) (see Appendix
\ref{sec:appendix_pontryagin_thom_construction} for more details).

Of course, there is more to $\UO$ than the bottom $S^2$.  Let's look next at the influence of the 3-cell $D^3$ on $R_{2n}$.  Whereas $S^2 \setminus N$ is already contractible, to make $S^3 \bigcup_{\text{deg } 2} D^3$ contractible, we must remove a diameter $\delta \subset D^3$ whose endpoints are glued to $N$.  Making our maps $f$ transverse to $\delta \subset D^3$ (and not just $N \subset S^2$ as would be ordinarily done in the PTC) we see that the ribbons (and their framed cobordisms) are now subject to a ``singularity'' at the $f$-preimages of the origin $O \subset D^3$.  The singularity is very mild: it is simply a point on the ribbon where the orientation (induced by comparing orientations on $M^3$ and $S^2$) longitudinally along the ribbon $r_0$ changes direction.  The upshot is that the generators for $\pi_0(R_{2n})$ now consist of ribbons with no longitudinal arrows assigning a direction along the ribbon. (These generating ribbons are, of course, subject to the boundary condition that there are exactly $2n$ arc endpoints meeting each hedgehog exactly once at its marked point $f^{-1}(N)$.)  Closed loops of ribbons may also arise.  Because all maps from a 3-manifold deform into these lowest cells, $S^2 \bigcup_{\text{deg }2} D^3 \subset \UO$, all generators of $\pi_0(R_{2n})$ are of this form.

Finding the relations defining $\pi_0(R_{2n})$ requires looking at maps of a four manifold $(Q \times I, \partial(Q \times I))$ into $\UO$.  To study these by a variant of the PTC, we would need to take into account all cells up to dimension 4, $sk^4(\UO)$, the ``4-skeleton.''  Similarly, a full analysis of $\pi_1(R_{2n})$ would require looking at $sk^5(\UO)$.

\begin{fact}\label{fact:1}
$\pi_3(S^2 \bigcup_{\text{deg } 2} D^3) \cong \Z_4$ and is generated by the Hopf map $h: S^3 \to S^2$.
\end{fact}

\begin{proof}
The Whitehead exact sequence reads:

\[\begindc{0}[3]
    \obj(10,10)[a]{$H_4(X)$}
    \obj(30,10)[b]{$\Gamma(\pi_2(X))$}
    \obj(50,10)[c]{$\pi_3(X)$}
    \obj(80,10)[d]{$H_3(X)$}
    \mor{a}{b}{}
    \mor{b}{c}{}
    \mor{c}{d}{Hurewicz}
\enddc\]

\noindent where $\Gamma$ is Whitehead's quadratic functor.  It is known that $\Gamma (\Z_2) \cong \Z_4$.  If $X = S^2 \bigcup_{\text{deg } 2} D^3$, the outer integral homology groups vanish, so $\pi_3(X) \cong \Z_4$ as claimed.

To see that the generator $g \in \pi_3(X)$ deforms to $S^2$, make transverse to the origin $O \subset D^3$.  Since $D^3$ is attached to $S^2$ with positive degree, the signed sum of inverse images $g^{-1}(O) = 0$.  These may be paired and then removed by building a framed cobordism (see PTC appendix) from $g^{-1}(O)$ to $\emptyset$ inside $S^3 \times [0,1]$.  The map $g'$ associated to the $S^3 \times 1$ level avoids $O \subset D^3$ and so may be radially deformed into $S^2$.

\end{proof}

\begin{consequence}
$sk^4(\UO) = S^2 \bigcup_{\text{deg } 2} D^3 \bigcup_{2 \text{ Hopf}} D^4$
\end{consequence}

\begin{proof}
Consider the commutative diagram:

\[\begindc{0}[3]
    \obj(50,43)[a]{$SO$}
    \obj(10,30)[b]{$S^3$}
    \obj(15,30)[c1]{$\cong$}
    \obj(22,30)[c]{$SU(2)$}
    \obj(50,30)[d]{$SU$}
    \obj(10,10)[e]{$S^2$}
    \obj(15,10)[c2]{$\cong$}
    \obj(22,10)[f]{$\frac{SU(2)}{SO(2)}$}
    \obj(50,10)[g]{$\frac{SU}{SO}$}
    \mor{a}{d}{fiber}
    \mor{b}{e}{Hopf map}[\atright,\solidarrow]
    \mor{c}{f}{coset}
    \mor{c}{d}{}
    \mor{f}{g}{}
    \mor{d}{g}{coset}
    \mor{d}{g}{$\delta$}[\atright,\solidarrow]
\enddc\]

\noindent Applying the $\pi_3$ functor, we obtain:

\[\begindc{0}[3]
    \obj(10,30)[a]{$\Z$}
    \obj(50,30)[b]{$\Z$}
    \obj(10,10)[c]{$\Z$}
    \obj(50,10)[d]{$\Z_2$}
    \mor{a}{b}{$\cong$}
    \mor{a}{c}{$\cong$}
    \mor{b}{d}{epimorphism}
    \mor{c}{d}{epimorphism}
\enddc\]

\noindent From the homotopy exact sequence of

\[\begindc{0}[3]
    \obj(10,20)[a]{$SO$}
    \obj(20,20)[b]{$SU$}
    \obj(20,10)[c]{$\frac{SU}{SO}$}
    \obj(29,10)[d]{$\cong\UO$}
    \mor{a}{b}{}
    \mor{b}{c}{}
\enddc\]

\noindent we can see that

\[\begindc{0}[3]
    \obj(10,20)[a]{$\pi_3(SO)$}
    \obj(30,20)[b]{$\pi_3(SU)$}
    \obj(50,20)[c]{$\pi_3(\frac{SU}{SO})$}
    \obj(70,20)[d]{$\pi_2(SO)$}
    \obj(10,15)[c1]{\begin{sideways}$\cong$\end{sideways}}
    \obj(30,15)[c1]{\begin{sideways}$\cong$\end{sideways}}
    \obj(50,15)[c1]{\begin{sideways}$\cong$\end{sideways}}
    \obj(70,15)[c1]{\begin{sideways}$\cong$\end{sideways}}
    \obj(10,10)[a2]{$\Z$}
    \obj(30,10)[b2]{$\Z$}
    \obj(50,10)[c2]{$\Z_2$}
    \obj(70,10)[d2]{$1$}
    \mor{a}{b}{$\times 2$}
    \mor{b}{c}{}
    \mor{c}{d}{}
\enddc\]

\noindent and we conclude that $\delta$ is an epimorphism.  Thus, the Hopf map, which is order 4 in $X$, needs a 4-cell to be attached in order to make it order 2 in $\UO$.
\end{proof}

We will not need this, but for the curious, $sk^5(\UO)$ requires one additional 5-cell which maps degree 2 over the 4-cell.

As remarked, studying maps {\em into} cell structures becomes laborious.  Homotopy theory works best when studying maps {\em out of} cell structures and into fibrations.  So, let us introduce the 2-stage Postnikov tower for $\UO$.

The notation $K(\pi,n)$ is used for any connected space (generally infinite dimensional) that has only a single nontrivial homotopy group, $\pi_n(K(\pi,n)) = \pi$.  It is easy to show that the homotopy type of $K(\pi, n)$ is unique.  At the next stage of complexity come spaces $Z$ with two nontrivial homotopy groups, say $\pi_m(Z) = A$ and $\pi_n(Z) = B$, $m < n$.  Such spaces (for simplicity we now assume $m \geq 2$) have a homotopy type determined by the total space $T$ of a fibration called a {\em 2-stage Postnikov Tower}:

\[\begindc{0}[3]
    \obj(10,25)[a]{$K(n,B)$}
    \obj(10,21)[a1]{(fiber)}
    \obj(30,25)[b]{$T$}
    \obj(30,10)[c]{$K(m,A)$}
    \obj(42,10)[c1]{(base)}
    \mor{a}{b}{}
    \mor{b}{c}{}
\enddc\]

To completely specify how the fibration twists and the homotopy type of $T$ we need to specify the ``$k$'' invariant $k \in H^{n+1}(K(m,A);B)$.

The space $K(m,A)$ classifies cohomology in the sense that for any space $S$, the homotopy classes of maps $[S, K(m,A)]$ naturally biject with $H^m(S;A)$.  (There is a ``fundamental'' class $\iota \in H^m(K(m,A);A)$ so that for $f: S \to K(m,A)$ we associate $f^* \iota \in H^m(S;A)$.) Thus, the $k$ invariant is really a map: $$k: K(m,A) \to K(n+1,B)$$ and $T$ is the pullback of the path loop fibration over $K(n+1,B)$:

\[\begindc{0}[3]
    \obj(0,25)[a]{$\Omega K(n+1,B) = K(n,B)$}
    \obj(30,25)[b]{pt.}
    \obj(30,10)[c]{$K(n+1,B)$}
    \mor{a}{b}{}
    \mor{b}{c}{}
\enddc\]
\\
\[\begindc{0}[3]
    \obj(10,30)[a]{$K(n,B)$}
    \obj(40,30)[b]{$K(n,B)$}
    \obj(10,15)[c]{$T$}
    \obj(40,15)[d]{pt.}
    \obj(10,0)[e]{$K(m,A)$}
    \obj(40,0)[f]{$K(n+1,B)$}
    \mor{a}{c}{}
    \mor{b}{d}{}
    \mor{c}{d}{}
    \mor{c}{e}{}
    \mor{d}{f}{}
    \mor{e}{f}{$k$}
\enddc\]

\noindent where we have denoted the contractible space of paths in $K(n+1,B)$ starting at its base point by its homotopy model, a single point, pt.

For spaces with many homotopy groups, the Postnikov tower can be continued iteratively.

In our case, $\pi_i(\UO) = 0, \Z_2, \Z_2, 0$ for $i = 1,2,3,4$, so a 2-stage tower is adequate for computing $\pi_i(R)$ for $i=0,1$.  Thus, our model $T$ for $\UO$ is now a fibration:

\[\begindc{0}[3]
    \obj(10,25)[a]{$K(\Z_2,3)$}
    \obj(30,25)[b]{$T$}
    \obj(30,10)[c]{$K(\Z_2,2)$}
    \mor{a}{b}{}
    \mor{b}{c}{}
\enddc\]

Similar to the previously computed cell structure of $\UO$, the cells of $K(\Z_2,2)$ through dimension $\leq 3$ are the same.  But for $K(\Z_2,2)$, the 4-cell kills the Hopf map and the 5-cell is now degree 4 over the 4-cell.  The cell complex for $K(\Z_2,2)$ begins as follows:

\[\begindc{0}[3]
    \obj(10,20)[a]{$C_5$}
    \obj(30,20)[b]{$C_4$}
    \obj(50,20)[c]{$C_3$}
    \obj(70,20)[d]{$C_2$}
    \obj(90,20)[e]{$C_1$}
    \obj(10,15)[c1]{\begin{sideways}$\cong$\end{sideways}}
    \obj(30,15)[c1]{\begin{sideways}$\cong$\end{sideways}}
    \obj(50,15)[c1]{\begin{sideways}$\cong$\end{sideways}}
    \obj(70,15)[c1]{\begin{sideways}$\cong$\end{sideways}}
    \obj(90,15)[c1]{\begin{sideways}$\cong$\end{sideways}}
    \obj(10,10)[a2]{$\Z$}
    \obj(30,10)[b2]{$\Z$}
    \obj(50,10)[c2]{$\Z$}
    \obj(70,10)[d2]{$\Z$}
    \obj(90,10)[e2]{$0$}
    \mor{a}{b}{}
    \mor{b}{c}{}
    \mor{c}{d}{}
    \mor{d}{e}{}
    \mor{a2}{b2}{$\times 4$}
    \mor{b2}{c2}{$0$}
    \mor{c2}{d2}{$\times 2$}
    \mor{d2}{e2}{}
\enddc\]

\noindent The $\Z_2$-cohomology is then the homology of the hom sequence of this chain complex into $\Z_2$:

\[\begindc{0}[3]
    \obj(10,20)[a]{$\Z_2$}
    \obj(30,20)[b]{$\Z_2$}
    \obj(50,20)[c]{$\Z_2$}
    \obj(70,20)[d]{$\Z_2$}
    \obj(90,20)[e]{$0$}
    \mor{b}{a}{0}[\atright,\solidarrow]
    \mor{c}{b}{0}[\atright,\solidarrow]
    \mor{d}{c}{0}[\atright,\solidarrow]
    \mor{e}{d}{}
\enddc\]

Thus, one finds that $H^4(K(\Z_2, 2);\Z_2) \cong \Z_2$, so there are only two possible fibrations (up to homotopy type) as above, the product $K(\Z_2,2) \times K(\Z_2,3)$ and a ``twisted product.''  Fact \ref{fact:1} says that $\pi_3(T)$ actually comes by composing the Hopf map into the (2-sphere inside the) base $K(\Z_2,2)$.  This is clearly false for the product where that composition would factor through $\pi_3(K(\Z_2,2)) \cong 1$.  This shows that the model $T$ for $\UO$ has nontrivial $k$ invariant.

If $X$ is a space and

\[\begindc{0}[30]
    \obj(1,2)[a]{$F$}
    \obj(2,2)[b]{$E$}
    \obj(2,1)[c]{$B$}
    \mor{a}{b}{}
    \mor{b}{c}{}
\enddc\]

\noindent is a fibration, then there is a corresponding fibration of the space of maps:

\[\begindc{0}[30]
    \obj(0,2)[a]{$\mathcal{M}(X,F)$}
    \obj(2,2)[b]{$\mathcal{M}(X,E)$}
    \obj(2,1)[c]{$\mathcal{M}(X,B)$}
    \mor{a}{b}{}
    \mor{b}{c}{}
\enddc\]

\noindent (We now use the usual comma instead of ``$\to$'' in the notation for map spaces.)  This is our main tool; it enables us to tear apart maps into $T$ into maps into $K(\Z_2,2)$ and $K(\Z_2,3)$ which are easy things to compute, namely cohomology groups.  A small wrinkle is that because of the boundary conditions we will study relative maps and end up with relative cohomology groups, but these are also easily computed.

So, let us now return to the computation of $\pi_0(R_{2n})$.  We study the fibration:

\[\begindc{0}[30]
    \obj(1,2)[a]{$R3$}
    \obj(2,2)[b]{$R$}
    \obj(2,1)[c]{$R2$}
    \mor{a}{b}{}
    \mor{b}{c}{}
\enddc\]

\noindent where $R$ is our original $R_{2n}$ (the subscript $n \geq 0$ is suppressed).  $R2$ is the space $\text{Maps}(Q \to K(\Z_2,2))$ with rigid boundary conditions (from $\partial Q$ to $S^2 \subset K(\Z_2,2)$) analogous to $R$ but with maps to $T$ replaced with maps to its base space $K(\Z_2,2)$.  Similarly, $R3$ is $\text{Maps}(Q \to K(\Z_2,3))$ with $\partial Q$ mapping to the base point of $K(\Z_2,3)$.  We have from the homotopy exact sequence:

\[\begindc{0}[3]
    \obj(10,30)[a]{$\pi_1(R2)$}
    \obj(30,30)[b]{$\pi_0(R3)$}
    \obj(45,30)[c]{$\pi_0(R)$}
    \obj(60,30)[d]{$\pi_0(R2)$}
    \obj(30,20)[e]{$H^3(Q,\partial Q;\Z_2)$}
    \obj(70,20)[f]{$H^2(Q,\partial Q;\Z_2)\cong H_1(Q;\Z_2){\hskip 1 cm}$}
    \obj(76,14)[g]{$\cong 1$}
    \obj(30,10)[g]{$\Z_2$}
    \obj(30,25)[c1]{\begin{sideways}$\cong$\end{sideways}}
    \obj(30,15)[c2]{\begin{sideways}$\cong$\end{sideways}}
    \obj(60,25)[c1]{\begin{sideways}$\cong$\end{sideways}}
    \obj(60,10)[p]{Poincare duality}
    \mor{a}{b}{$\partial_1$}
    \mor{b}{c}{}
    \mor{c}{d}{}
    \mor{p}{f}{}
\enddc\]

The boundary map is always zero for product fibrations, but in this case, it is potentially affected by the $k$ invariant (coming from the $k$ invariant of $T$).

Following the usual procedure, $\partial_1$ may be understood as the ``twist'' (mod 2) in a ribbon induced by the ``lasso move'' shown in figure \ref{fig:lassomove} below, which is well known to be $4\pi$.

\begin{figure}[htpb]
\labellist \small\hair 2pt
  \pinlabel $4\pi\text{ twist}$ at 360 75
\endlabellist
%\centering
{\hskip -0.8 cm} \includegraphics[width=3in]{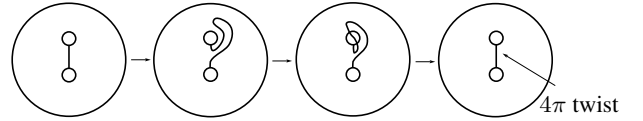}
\caption{By pulling a ribbon over a hedgehog and `lassoing' it,
a $4\pi$ twist is performed.}
\label{fig:lassomove}
\end{figure}

\noindent Since $4 \pi$ is an even multiple of $2 \pi$, $\partial_1$ is in fact zero and $\pi_0(R) \cong \Z_2$.  The two components can, on the level of framed ribbons, be identified with a total twisting of ribbons being an {\em even} or {\em odd} multiple of $2\pi$.

$\pi_1(R)$ may be similarly computed:

\[\begindc{0}[3]
    \obj(10,40)[a]{$\cdots$}
    \obj(30,40)[b]{$\pi_1(R3)$}
    \obj(50,40)[c]{$\pi_1(R)$}
    \obj(70,40)[d]{$\pi_1(R2)$}
    \obj(90,40)[e]{$\cdots$}
    \obj(30,35)[c1]{\begin{sideways}$\cong$\end{sideways}}
    \obj(70,35)[c2]{\begin{sideways}$\cong$\end{sideways}}
    \obj(30,30)[b2]{$H^3(Q\times I, \partial(Q \times I);\Z_2)$}
    \obj(70,30)[d2]{$H^2(Q\times I, \partial(Q \times I);\Z_2)$}
    \obj(30,25)[c3]{\begin{sideways}$\cong$\end{sideways}}
    \obj(70,25)[c4]{\begin{sideways}$\cong$\end{sideways}}
    \obj(30,20)[b3]{$H^1(Q\times I;\Z_2)$}
    \obj(70,20)[d3]{$H_2(Q\times I;\Z_2)$}
    \obj(30,15)[c5]{\begin{sideways}$\cong$\end{sideways}}
    \obj(70,15)[c6]{\begin{sideways}$\cong$\end{sideways}}
    \obj(30,10)[b4]{$1$}
    \obj(70,10)[d4]{$\Z_2^{2n}$}
    \obj(50,25)[p]{Poincare duality}
    \mor{a}{b}{}
    \mor{b}{c}{}
    \mor{c}{d}{}
    \mor{d}{e}{}
    \mor{p}{c3}{}
    \mor{p}{c4}{}
\enddc\]

Thus, $\pi_1(R) \cong \Z_2^{2n}$.  Intuitively, $\pi_1(R)$ is represented by 2-dimensional ``bags'' or homology classes that a loop of ribbons, defining an element of $\pi_1(R)$, sweeps out in $Q \times I$.

%%%%%%%%%%%%%%%%%%%%%%%%%%%%%%%%%%%%%
\section{Pontryagin-Thom Construction}
\label{sec:appendix_pontryagin_thom_construction}
%%%%%%%%%%%%%%%%%%%%%%%%%%%%%%%%%%%%%

Homotopy classes $[M^d, S^k]$ may be studied geometrically.  The idea is to make any map $f: M^d \to S^k$ transverse to a base point $N \in S^k$.  Then $f^{-1}(N)$ is a $(d-k)$-dimensional submanifold $K$ of $M^d$ equipped with a framing of its normal bundle obtained by pulling back a fixed normal $k$-frame to $N$.  This construction is reversible: Given a submanifold with normal framing $K^{d-k} \subset M$, there is a map $f:M^d \to S^k$ which wraps an $\epsilon$-neighborhood $\eta$ of $K$ over $S^k$ by using polar coordinates around the north pole $N$ and the framing to send each normal $D_\epsilon^k \subset \eta$ degree $=1$ over $S^k$ and sends the rest of $M \setminus \eta$ to the south pole in $S^k$.

This discussion can be relativized.  A homotopy $F: M \times [0,1] \to S$, $F|_{M \times 0} = f$, $F|_{M \times 1} = g$ yields by inverse image $\bbar{K}$ of $N$ a ``framed cobordism'' from $K_0 \subset M \times 0$ to $K_1 \subset M \times 1$, that is, a manifold $\bbar{K}^{d-k+1} \subset M \times [0,1]$ with a normal $k$-framing which restricts at 0 (1) to $K_{0\,(1)}$ and its normal k-framing.

The maps to $S^k$ that are produced from a framed submanifold may seem extreme and unrepresentative since most points map to the south pole.  However, the beauty of the construction is that since $S^k \setminus N$ is contractible, no real choice (up to homotopy) exists for the part of the map which avoids $N$.  Thus we have:

\begin{fact}\label{fact:pontryagin_thom_1}
The space of maps $\mathcal{M}(M^d, S^k)$ is (at least weakly) homotopy equivalent to the space of framed submanifolds $\{L^{d-k} \subset M^d\}$ (provided both are given reasonable topologies).
\end{fact}

\begin{fact}\label{fact:pontryagin_thom_2}
Fact \ref{fact:pontryagin_thom_1} holds in a relative setting, e.g.\ the $\{\text{space of homotopies} \} \overset{\text{weakly}}{\simeq} \{\text{space of framed cobordisms}\}$.  For example, merely on the level of $\pi_0$, this says that ``homotopy classes of maps $(M^d,S^k)$ are in 1-1 correspondence with framed cobordism classes of framed submanifolds $K^{d-k} \subset M^d$.''
\end{fact}

This basic PTC may be generalized to maps $(N^d,M^k)$ whenever we can identify a ``spine'' $X \subset M^k$ so that $M \setminus X$ is contractible.  One then studies $f^{-1}(X)$ (actually $f^{-1}(\text{a germ of } X)$ corresponding to the old framing data) instead of $f^{-1}(N)$.  For this approach to be practical, $X$ and its neighborhood have to be fairly simple.  In this extended setting, one must also keep track of the maps $f|: (f^{-1}(X)) \to X$ since this is no longer unique.

%%%%%%%%%%%%%%%%%%%%%%%%%%%%%%%%%%%%%
\section{Multiplication Table for $\pi_1(K_{2n})$}\label{sec:appendix_mult_table_K_2n}
%%%%%%%%%%%%%%%%%%%%%%%%%%%%%%%%%%%%%

We have shown ${\cal T}_{2n}=\pi_1({K}_{2n}) \cong \Z \times \Z_2 \times G_{2n}$ and a short exact sequence (SES):

\[\begindc{0}[3]
    \obj(10,10)[a]{$1$}
    \obj(23,10)[b]{$G_{2n}$}
    \obj(43,10)[c]{$\Z_2^{2n}\rtimes S_{2n}$}
    \obj(70,10)[d]{$\Z_2$}
    \obj(85,10)[e]{$1$}
    \mor{a}{b}{}
    \mor{b}{c}{}
    \mor{c}{d}{total parity}
    \mor{d}{e}{}
\enddc\]

Thus, we have another SES:

\[\begindc{0}[3]
    \obj(10,30)[a]{$1$}
    \obj(30,30)[b]{$\Z_2^{2n-1}$}
    \obj(50,30)[c]{$G_{2n}$}
    \obj(70,30)[d]{$S_{2n}$}
    \obj(90,30)[e]{$1$}
    \mor{a}{b}{}
    \mor{b}{c}{}
    \mor{d}{e}{}
    \cmor((56,31)(61,33)(67,31)) \pright(61,35){}
    \cmor((67,29)(61,27)(56,29)) \pleft(61,25){$s$}
\enddc\]

\noindent with a merely set theoretic section $s$.  However, using $s$ (line 3.5 of \cite{Brown82}) we can write the multiplication table out for $G_{2n}$, hence $\pi_1(K_{2n})$ in terms of a twisted 2-cocycle $f$ on $S_{2n}$:

\begin{equation}\label{eq:appendix_mult_table_pi_1(Y_2n)_twisted_cocycle}
(a,g)(b,h) = (a+g.b+f(g,h),gh)
\end{equation}

\noindent where $a \in \Z_2^{2n}$, $g,h \in S_{2n}$.  $g$ acts on $\Z_2^{2n-1}$ by including $\Z_2^{2n-1} \subset \Z_2^{2n}$ as the ``even'' vectors and letting $g$ permute the factors of $\Z_2^{2n}$.  It only remains to determine a twisted cocycle $f\in H^2(S_{2n};\Z_2^{2n-1})$, where the
action of $S_{2n}$ on $\Z_2^{2n-1}$ is as above.  The cohomology class represented by $f$ is non-zero because the exact sequence does not split.  As a check, we know (line 3.10 in \cite{Brown82}) that $f$ should obey:

\begin{equation}\label{eq:appendix_twisted_cocycle_rule}
g(f(h,k))-f(gh,k)+f(g,hk)-f(g,k) = 0
\end{equation}

\begin{claim}\label{claim:appendix_cocycle_guess}
$f(g,h)=0$ if $g$ or $h$ is even and $f(g,h)=e_{g(1)}$
if $g$ and $h$ are both odd, where $e_{g(1)}$ is the $g(1)^{\text{th}}$ standard generator of $\Z_2^{2n}$.
\end{claim}

\section{Ghostly recollection of the braid group}
\label{sec:braid-group}

To promote the Teo-Kane representation to a linear
representation of the braid group ${\cal B}_m$, we have the freedom to introduce an
overall abelian phase to each $T_{i,j}$.  There are several different
choices related to different physical systems.  For our purpose
here, any choice suffices.  The different choices lead to
different linear images of the braid groups as the braid generator matrix has
different orders.

The Ising TQFT can be realized using the Kauffman bracket
by choosing $A=ie^{-\frac{2\pi i}{16}}$.  The braid
group representation from a TQFT is not well-defined as a matrix
representation because in general we do not have a canonical
choice of the basis vectors.  In the Kauffman bracket formulation of
Ising TQFT, bases of representation spaces can be constructed using
linear combinations of Temperley-Lieb (TL) diagrams or trivalent graphs.
Then the representation of the braid groups is given
by the Kauffman bracket interpreted as a map from ${\cal B}_m$ to
units of TL algebra.

The Temperley-Lieb-Jones algebra $TL_m(A)$ at the above chosen $A$ is
isomorphic to the Clifford algebra.  Therefore the $\gamma$-matrices $\{\gamma_i\}$
 can be represented by TL
generators $\{U_i\}$.  In terms of the $\gamma$-matrices, the Jones
representation in Kauffman bracket normalization is:
$$ \rho_A(\sigma_i)=e^{-\frac{3\pi
i}{8}}e^{-\frac{\pi}{4}\gamma_i\gamma_{i+1}}.$$
Jones original representation from von Neumann algebra is:
$$ \rho^J_A(\sigma_i)=e^{\frac{\pi
i}{4}}e^{-\frac{\pi}{4}\gamma_i\gamma_{i+1}}.$$
We will also refer to the $T_{i,i+1}$ representation of the braid group
$$ \rho^{\gamma}_A(\sigma_i)=e^{-\frac{\pi}{4}\gamma_i\gamma_{i+1}}$$
as the $\gamma$-matrix representation.

The ratio of the two distinct eigenvalues of the braid generator is $-i$,
independent of the overall abelian phase.
The orders of the braid
generator matrices in the three normalizations are $16, 4, 8$, respectively.  In the
following, we will focus on the Jones representation.

Let $q=A^{-4}=i$, then the Kauffman bracket is
$ \rho_A(\sigma_i)=<\sigma_i>=A+A^{-1}U_i$, where $U_i$'s are the
TL generators.  For convenience, we introduce the Jones generators of the
TL algebras $e_i=\frac{U_i}{d}, d=\sqrt{2}, i=1,\cdots, m-1$. These $e_i$'s
should not be confused with the basis element $e_i$'s of $\mathbb{R}^m$.  The TL algebra in terms of
$e_i$'s is: $$e_i^{\dagger}=e_i, e_i^2=e_i, e_ie_{i\pm 1}e_i=\frac{1}{d^2} e_i.$$
 The Jones representation is
$\rho^J_A(\sigma_i)=-A\rho_A(\sigma_i)=-1+(1+q)e_i$.  Let
$x_i=<\sigma^2_i>$.  Then we have: $x_i^2=1$, $x_ix_j=x_jx_i$ if
$|i-j|>1$, and $x_ix_{i+1}=-x_{i+1}x_i$.  The $-1$ in the relation
$x_ix_{i+1}=-x_{i+1}x_i$ is from the Jones-Wenzl projectors
$p_3=0$.  By a simple calculation, the Jones-Wenzl projector $p_3$ in
terms of $\{e_i, e_{i+1}\}$ is $p_3=1+2(e_ie_{i+1}+e_{i+1}e_i-e_i-e_{i+1})$, and
$x_ix_{i+1}+x_{i+1}x_i=2p_3$.

Note that $\{x_i=1-2e_i, i=1,\cdots, m-1\}$ generate $TL_m(A)$.
Write $x_i$ in terms of $\gamma$-matrices, we have
$x_i=\sqrt{-1}\gamma_i\gamma_{i+1}$, which
identifies $TL_m(A)$ with the even part $Cl_m^0(\mathbb{C})$ of
$Cl_m(\mathbb{C})$. On the other hand, if we set
$v_i=(\sqrt{-1})^{i-1} x_i\cdots x_1, i=1,\cdots, m-1$, a
direct computation shows that $v_i^{\dagger}=v_i,
v_i v_j+v_j v_i=2\delta_{i,j}$.  Therefore, $\{v_i,i=1,\cdots,
m-1\}$ form the Clifford algebra $Cl_{m-1}(\mathbb{C})$.  Of
course, the two different realizations $TL_m(A)$ as Clifford
algebras are just the well known isomorphism
$Cl_m^0(\mathbb{C})\cong Cl_{m-1}\mathbb{C})$.

To identify the image of the Jones representation, we use the
exact sequence $1\rightarrow PB_m \rightarrow {\cal B}_m \rightarrow
S_m\rightarrow 1$.  The image of the pure braid group $PB_m$ is
generated by $\{x_i\}$.  As an abstract group, the group can be
presented
with generators $x_i, i=1\cdots, x_{m-1}$ and relations $x_i^2=1$,
$x_ix_j=x_jx_i$ if $|i-j|>1$, and $x_ix_{i+1}=-x_{i+1}x_i$.
This group, called the nearly-extra-special $2$-group in \cite{Franko06}
and denoted as $E^1_{m-1}$, is of order $2^m m!$.
Consequently, the image ${\til{G}}_m$ of the Jones representation of
${\cal B}_m$ fits into the exact sequence $1\rightarrow
E^1_{m-1}\rightarrow
\tilde{G}_m \rightarrow S_m\rightarrow 1.$
If we projectivize this sequence, we have $1\rightarrow
\mathbb{Z}^{m-1}_2\rightarrow
{G}_m \rightarrow S_m\rightarrow 1.$ (This sequence splits if and only if $m$ is odd.)
In this sense, the ribbon
permutation group is a ghostly recollection of the braid group.

The center $Z(E^1_{m-1})$ of $E^1_{m-1}$ is $\mathbb{Z}_2=\{\pm 1\}$ if $m$ is
even, and $\mathbb{Z}_2\times \mathbb{Z}_2=\{\pm 1,\pm
x_1x_3\cdots x_{m-1}\}$ if $m$ is odd.  The
representation of $E^1_{m-1}$ is faithful, but when $m$ is even,
none of the irreducible sector is faithful because the central
element $\pm x_1x_3\cdots x_{m-1}$ acts by $\pm 1$, too.  The
action of $x_i$ on the two irreducible sectors are the same for $i=1,2,\cdots, m-2$ and
differs only on $x_{m-1}$, which is $\pm 1$.
When projectivized, note that the element $\alpha=(11\cdots 1)\in
\mathbb{Z}_2^{m}$, as an element in $\mathbb{Z}_2^{m-1}$, is invariant under the action of $S_m$, hence
the projective image of each irreducible sector is
isomorphic to $\mathbb{Z}_2^{m-2}$.

As comparison, we mention the related result in \cite{Read03}.  The Jones representation there
is in the $\gamma$-matrix normalization.  The pure braid generators are of order $4$.  Hence the
pure braid group image is a variant $E^{-1}_{m-1}$ of $E^1_{m-1}$.  The
generators of the two groups are related by the change of $x_i$ to $-\sqrt{-1} x_i$.  A presentation of  $E^{-1}_{m-1}$
 with generators $x_i, i=1\cdots, x_{m-1}$ has relations $x_i^2=-1$,
$x_ix_j=x_jx_i$ if $|i-j|>1$, and $x_ix_{i+1}=-x_{i+1}x_i$.  The center $Z(E^{-1}_{m-1})$ is
more complicated:  $\mathbb{Z}_2$ if $m$ is odd, and $\mathbb{Z}_2\times \mathbb{Z}_2$ if $m=4k+1$ and
$\mathbb{Z}_4$ if $m=4k+3$.  When $m$ is
even, the image of each irreducible sector is identified as the lifting of $G_m$ to $Spin(m)$ \cite{Read03}.
It follows from the above discussion, the linear image of the Jones
representation in $\gamma$ matrix-normalization fits into the
exact sequence $1\rightarrow
\mathbb{Z}_4 \times E^{1}_{m-1} \rightarrow
\mathbb{Z}_8 \times \tilde{G}_m \rightarrow \mathbb{Z}_2 \times S_m \rightarrow 1$, where ${\til{G}}_m$ is the image of ${\cal B}_m$ above.

%\bibliographystyle{prsty}
%\bibliography{projective}

\begin{thebibliography}{58}
\bibitem{footnote1}
 We note that neither Teo and Kane's arguments\cite{Teo10}
nor ours determine the phase $\theta$. In a two-dimensional
system of anyons, it is not an arbitrary phase. For instance,
in a system of Ising anyons, $\theta$ must be equal to $\pi/8$.
Any value of $\theta$ which is not a $16^{\rm th}$
root of unity necessarily implies that the theory has
additional structure and, therefore, additional particle types.
\end{thebibliography}

\end{document}